\documentclass{lmcs}
\pdfoutput=1

\usepackage{lastpage}
\lmcsdoi{15}{3}{9}
\lmcsheading{}{\pageref{LastPage}}{}{}%
{Feb.~20,~2018}{Jul.~31,~2019}{}

\usepackage{macros}
\allowdisplaybreaks

\begin{document}

\title[Normalizing the Taylor expansion of non-deterministic $λ$-terms]{
	Normalizing the Taylor expansion of non-deterministic $λ$-terms,
	\emph{via} parallel reduction of resource vectors
}

\author{Lionel Vaux}
\address{Aix Marseille Univ, CNRS, Centrale Marseille, I2M, Marseille, France}
\email{lionel.vaux@univ-amu.fr}
\thanks{This work was supported by
	French ANR projects Coquas (number ANR-12-JS02-0006)
	and Récré (number ANR-11-BS02-0010),
  as well as the French-Italian Groupement de Recherche International on Linear Logic.}

\keywords{lambda-calculus, non-determinism, normalization, denotational semantics}
\ACMCCS{\textit{Theory of computation~Denotational semantics; }
Theory of computation~Lambda calculus;
Theory of computation~Linear logic;}
\titlecomment{
	An extended abstract \cite{vaux:taylor-beta} of this paper, focused on the
	translation of parallel β-reduction steps through Taylor expansion, appeared
	in the proceedings of \emph{CSL 2017}.
}

\begin{abstract}
	It has been known since Ehrhard and Regnier’s seminal work 
	on the Taylor expansion of λ-terms that this operation 
	commutes with normalization: the expansion of a λ-term 
	is always normalizable and its normal form is the expansion
	of the Böhm tree of the term.

	We generalize this result to the non-uniform setting of the 
	algebraic λ-calculus, \ie, λ-calculus extended with linear
	combinations of terms. This requires us to tackle two difficulties:
	foremost is the fact that Ehrhard and Regnier’s techniques rely heavily on
	the uniform, deterministic nature of the ordinary λ-calculus, and thus cannot
	be adapted; second is the absence of any satisfactory generic extension of
	the notion of Böhm tree in presence of quantitative non-determinism, which is
	reflected by the fact that the Taylor expansion of an algebraic λ-term is not
	always normalizable.

	Our solution is to provide a fine grained study of the dynamics of
	β-reduction under Taylor expansion, by introducing a notion of reduction on
	resource vectors, \ie infinite linear combinations of resource λ-terms.
	The latter form the multilinear fragment of the differential λ-calculus, and
	resource vectors are the target of the Taylor expansion of λ-terms.
	We show the reduction of resource vectors contains the image of any β-reduction
	step, from which we deduce that Taylor expansion and normalization commute on
	the nose.

	We moreover identify a class of algebraic λ-terms, encompassing both normalizable
	algebraic λ-terms and arbitrary ordinary λ-terms: the expansion of these is
	always normalizable, which guides the definition of a generalization of Böhm
	trees to this setting.
\end{abstract}

\maketitle

\section{Introduction}

Quantitative semantics was first proposed by Girard \cite{girard:quantitative}
as an alternative to domains and continuous functionals, for defining
denotational models of λ-calculi with a natural interpretation of
non-determinism: a type is given by a collection of “atomic states”; a term
of type $A$ is then represented by a vector (\ie a possibly infinite formal linear
combination) of states. The main matter is the treatment of the
function space: the construction requires the interpretation of function terms to be analytic, \ie
defined by power series.

This interpretation of λ-terms was at the origin of linear
logic: the application of an analytic map to its argument boils down to the
linear application of its power series (seen as a matrix) to the vector of
powers of the argument; similarly, linear logic decomposes the
application of λ-calculus into the
linear cut rule and the promotion operator. Indeed, the seminal model of linear
logic, namely coherence spaces and stable/linear functions, was introduced as a
qualitative version of quantitative semantics \cite[especially Appendix C]{girard:fifteen}. 

Dealing with power series, quantitative semantics must account for infinite
sums. The interpretations of terms in Girard's original model can be seen as a special case of Joyal's analytic
functors \cite{joyal:especes}: in particular, coefficients are sets
and infinite sums are given by coproducts. This allows to give a semantics to fixed
point operators and to the pure, untyped λ-calculus. On the other hand, it does
not provide a natural way to deal with weighted (\eg, probabilistic)
non-determinism, where coefficients are taken in an external semiring of
scalars.

In the early 2000’s, Ehrhard introduced an alternative presentation of
quantitative semantics \cite{ehrhard:fs}, limited to a typed setting, but where
types can be interpreted as particular vector spaces, or more generally
semimodules over an arbitrary fixed semiring; called \emph{finiteness spaces},
these are moreover equipped with a linear topology, allowing to interpret
linear logic proofs as linear and continuous maps, in a standard sense.
In this setting, the formal operation of differentiation of power series
recovers its usual meaning of linear approximation of a function, and morphisms
in the induced model of λ-calculus are subject to Taylor expansion:
the application $φ(α)$ of the
analytic function $φ$ to the vector $α$ boils down to the sum
$\sum_{n\in\naturals}\frac 1{\factorial n}
\pars{\frac{\partial^n φ}{\partial x^n}}_{x=0}\cdot{α^n}$
where $\pars{\frac{\partial^n φ}{\partial x^n}}_{x=0}$
is the $n$-th derivative of $φ$ computed at $0$, 
which is an $n$-linear map, and $α^n$ is the $n$-th tensor power of $α$.

Ehrhard and Regnier gave a computational meaning to such derivatives by 
introducing linearized variants of application and substitution in the λ-calculus,
which led to the differential λ-calculus \cite{er:tdlc},
and then the resource λ-calculus \cite{er:resource} ---
the latter retains iterated derivatives at zero as the only form of application.
They were then able to recast the above Taylor expansion formula in a syntactic,
untyped setting: to every λ-term $M$, they associate a vector $\TaylorExp M$ of
resource λ-terms, \ie terms of the resource λ-calculus.

The Taylor expansion of a λ-term can be seen as an intermediate, infinite
object, between the term and its denotation in quantitative semantics. Indeed,
resource terms still retain a dynamics, if a very simple, finitary one: the size of
terms is strictly decreasing under reduction. Furthermore, normal resource terms are
in close relationship with the atomic states of quantitative semantics of the
pure λ-calculus (or equivalently with the elements of a reflexive object in the
relational model \cite{bem:enough}; or with normal type derivations in a
non-idempotent intersection type system \cite{carvalho:exec}), so that the normal
form of $\TaylorExp M$ can be considered as the denotation of $M$, which allows
for a very generic description of quantitative semantics.

Other approaches to quantitative semantics generally impose a constraint on the
computational model \emph{a priori}. For instance, the model of finiteness
spaces \cite{ehrhard:fs} is, by design, limited to strongly normalizing
computation. Another example is that of probabilistic coherence spaces
\cite{de:pcoh}, a model of untyped λ-calculi extended with probabilistic
choice, rather than arbitrary weighted superpositions.
Alternatively, one can interpret non-deterministic extensions of PCF
\cite{lmcp:weighted,laird:fixed}, provided the semiring of scalars has
all infinite sums. By contrast, the “normalization of
Taylor expansion” approach is more canonical, as it does not rely on a
restriction on the scalars, nor on the terms to be interpreted.

Of course, there is a price attached to such canonicity:
in general, the normal form of a vector of resource
λ-terms is not well defined, because we may have to consider infinite sums of
scalars.
Ehrhard and Regnier were nonetheless able to prove that the Taylor expansion
$\TaylorExp M$ of a pure λ-term is always normalizable \cite{er:resource}. This
can be seen as a new proof of the fact that Girard’s quantitative semantics of pure
λ-terms uses finite cardinals only \cite{hasegawa:generating}.
They moreover established that this normal form is exactly the Taylor expansion
of the Böhm tree $\BohmTree M$ of $M$ \cite{er:bkt} ($\BohmTree M$ is the
possibly infinite tree obtained by hereditarily applying the head reduction
strategy in $M$).
Both results rely heavily on the uniformity property of the pure λ-calculus:
all the resource terms in $\TaylorExp M$ follow a single syntactic tree pattern.
This is a bit disappointing since quantitative semantics 
was introduced as a model of non-determinism,
which is ruled out by uniformity.

Actually, the Taylor expansion operator extends naturally to the algebraic
λ-calculus \cite{vaux:alglam}: a generic, non-uniform extension of
λ-calculus, augmenting the syntax with formal finite linear
combinations of terms. Then it is not difficult to find terms whose 
Taylor expansion is not normalizable.
Nonetheless, interpreting types as finiteness spaces of resource terms, 
Ehrhard \cite{ehrhard:finres} proved by a reducibility technique that 
the Taylor expansion of algebraic λ-terms typed in a variant of system $F$ is
always normalizable.

\subsection{Main results}

In the present paper, we generalize Ehrhard’s result and show that all weakly
normalizable algebraic λ-terms have a normalizable Taylor expansion (Theorem
\ref{theorem:converges:normalizes}, p.\pageref{theorem:converges:normalizes}).\footnote{
	We had already obtained such a result for strongly normalizable λ-terms
	in a previous work with Pagani and Tasson \cite{ptv:taylorsn}: there, we
	further proved that the finiteness structure on resource λ-terms could be
	refined to characterize exactly the strong normalizability property in a
	λ-calculus with finite formal sums of terms.
	Here we rely on a much coarser notion of finiteness: see subsection
	\ref{subsection:normalizable}.
}

We moreover relate the normal form of the expansion of a term with the
normal form of the term itself, both in a computational sense (\ie\ the
irreducible form obtained after a sequence of reductions) and in a more
denotational sense, via an analogue of the notion of Böhm tree:
Taylor expansion does commute with normalization, in both those senses
(Theorem \ref{theorem:NormalForm:TaylorExp}, p.\pageref{theorem:NormalForm:TaylorExp};
Theorem \ref{theorem:taylor:determinable}, p.\pageref{theorem:taylor:determinable}).

When restricted to pure λ-terms, Theorem \ref{theorem:taylor:determinable}
provides a new proof, not relying on uniformity, that the normal form of
$\TaylorExp M$ is isomorphic to $\BohmTree M$.
In their full extent, our results provide a generalization of the notion of
non-deterministic Böhm tree \cite{deliguoro-piperno:ndlc} in a weighted,
quantitative setting.

Let us stress that neither Ehrhard’s work \cite{ehrhard:finres} nor our own
previous work with Pagani and Tasson \cite{ptv:taylorsn} addressed the
commutation of normalization and Taylor expansion. Indeed, in the absence of
uniformity, the techniques used by Ehrhard and Regnier
\cite{er:resource,er:bkt} are no longer available, and we had to design another
approach.\footnote{
	It is in fact possible to refine Ehrhard and Regnier’s approach, \emph{via}
	the introduction of a \emph{rigid} variant of Taylor expansion
	\cite{tao:species}, which can then be adapted to the non-deterministic
	setting.
	This allows to describe the coefficients in the normal form of Taylor
	expansion, like in the uniform case, and then prove that Taylor expansion
	commutes with the computation of Böhm trees.
	It does not solve the problem of possible divergence, though, and one has to
	assume the semiring of coefficients is complete, \ie that all sums converge.
	See Subsection \ref{subsection:related} on related work for more details.
} Our solution is to introduce a notion of reduction on resource vectors,
so that: (i) this reduction contains the translation of any β-reduction step
(Lemma \ref{lemma:vpresRed:correct}, p.\pageref{lemma:vpresRed:correct});
(ii) normalizability (and the value of the normal form) of resource vectors is
preserved under reduction (Lemma \ref{lemma:vpresRed:NormalForm},
p.\pageref{lemma:vpresRed:NormalForm}). This approach turns out to be quite
delicate, and its development led us to two technical contributions that we
deem important enough to be noted here:
\begin{itemize}
\item the notion of \emph{reduction structure} (subsection
	\ref{subsection:structures}) that allows to control the families of resource
	terms simultaneously involved in the reduction of a resource vector:
	in particular this provides a novel, modular mean to circumvent the
	inconsistency of β-reduction in presence of negative coefficients
	(a typical deficiency of the algebraic λ-calculus \cite{vaux:alglam});
\item our analysis of the effect of parallel reduction on the size of resource
	λ-terms (Section \ref{section:parallel}): this constitutes the technical core
	of our approach, and it plays a crucial rôle in establishing key additional
	properties such as confluence (Lemma \ref{lemma:vpbdresRed:diamond},
	p.\pageref{lemma:vpbdresRed:diamond}, and Corollary
	\ref{corollary:splitVariant:presRed:fullReduct},
	p.\pageref{corollary:splitVariant:presRed:fullReduct}) and conservativity
	(Lemma \ref{lemma:resEq:conservative:zerosumfree},
	p.\pageref{lemma:resEq:conservative:zerosumfree}, and
	Lemma \ref{lemma:resEq:conservative:normalizable},
	p.\pageref{lemma:resEq:conservative:normalizable}).
\end{itemize}

\subsection{Structure of the paper}

The paper begins with a few mathematical preliminaries, in section
\ref{section:preliminaries}: we recall some definitions
about semirings and semimodules (Subsection \ref{subsection:semirings}), if only to fix
notations and vocabulary; we also provide a very brief review of finiteness
spaces (Subsection \ref{subsection:finiteness}),
then detail the particular case of
linear-continuous maps defined by summable families of vectors (subsection
\ref{subsection:summable}), the latter notion pervading the paper.

In Section \ref{section:resource} we review the syntax and the reduction
relation of the resource λ-calculus, as introduced by Ehrhard and Regnier
\cite{er:resource}. The subject is quite standard now, and the only new
material we provide is about minor and unsurprising combinatorial properties of
multilinear substitution.

Section \ref{section:vectors} contains our first notable contribution: after
recalling the Taylor expansion construction, we prove that it is compatible
with substitution. This result is related with the functoriality of promotion 
in quantitative denotational models and the proof technique is quite similar.
In the passing, we recall the syntax of the algebraic λ-calculus and briefly
discuss the issues raised by the contextual extension of β-reduction in
presence of linear combinations of terms, as evidenced by previous work
\cite[\emph{etc.}]{vaux:alglam2,ad:lineal,vaux:alglam}.

In Section \ref{section:reduction:vectors}, we discuss the possible extensions
of the reduction of the resource λ-calculus to resource vectors, \ie{} infinite
linear combinations of resource terms, and identify two main issues.  First, in
order to simulate β-reduction, we
are led to consider the parallel reduction of resource terms in resource
vectors, which is not always well defined. Indeed, a single resource term
might have unboundedly many antecedents by parallel reduction, hence this
process might generate infinite sums of coefficients: we refer to this
phenomenon as the \emph{size collapse} of parallel resource reduction
(Subsection \ref{subsection:collapse}).
Second, similarly to the case of the algebraic λ-calculus, the induced
equational theory might become trivial, due the interplay between coefficients
in vectors and the reduction relation.
To address the latter problem we introduce the notion of reduction structure
(Subsection \ref{subsection:structures}) which allows us to modularly restrict
the set of resource terms involved in a reduction: later in the paper, we will
identify reduction structures ensuring the consistency of the reduction of
resource vectors (Subsection \ref{subsection:conservativity}).

In Section \ref{section:parallel}, we introduce successive restrictions of the
parallel reduction of resource vectors, in order to avoid the abovementioned
size collapse. We first observe that, to bound the size of a term as a
function of the size of any of its reducts, it is sufficient to bound the
length of chains of immediately nested fired redexes in a single parallel
reduction step (Subsection \ref{subsection:pbresRed}).
This condition does not allow us to close a pair of reductions to a common
reduct, because it is not stable under unions of fired redexes. We thus tighten
it to bounding the length of all chains of (not necessarily immediately) nested
fired redexes (Subsection \ref{subsection:pbdresRed}): this enables us to
obtain a strong confluence result, under a mild hypothesis on the semiring.
An even more demanding condition is to require the fired redexes as well
as the substituted variables to occur at a bounded depth (subsection
\ref{subsection:pbsresRed}): then we can define a maximal parallel reduction
step for each bound, which entails strong confluence without any additional hypothesis.
Finally, we consider reduction structures involving resource terms of bounded
height (Subsection \ref{subsection:bounded}): when restricting to such a
bounded reduction structure, the strongest of the above three conditions
is automatically verified.

We then show, in Section \ref{section:simulation:beta}, that the translation 
of β-reduction through Taylor expansion fits into this setting: the height of
the resource terms involved in a Taylor expansion is bounded by that of the
original algebraic λ-term, and every β-reduction step is an instance of the
previously introduced parallel reduction of resource vectors. As a consequence
of our strongest confluence result, we moreover obtain that any reduction step
from the Taylor expansion of a λ-term can be extended into the translation of a
parallel β-reduction step.

We turn our attention to normalization in Section \ref{section:normalization}.
We first show that normalizable resource vectors are stable under reduction.
We moreover establish that their normal form is obtained as the limit of the
parallel left reduction strategy (Subsection \ref{subsection:normalizable}).
Then we introduce \emph{Taylor normalizable} algebraic λ-terms as those having a
normalizable Taylor expansion, and deduce from the previous results that they
are stable under β-reduction (Subsection \ref{subsection:Taylor:normalizable}):
in particular, the normal form of Taylor expansion does define a denotational
semantics for that class of terms.
Then we establish that normalizable terms are Taylor normalizable (subsection
\ref{subsection:normalizable:terms}): it follows that normalization and Taylor
expansion commute on the nose.

We conclude with Section \ref{section:determinable}, showing how our techniques
can be applied to the class of hereditarily determinable terms, that we
introduce \emph{ad-hoc}: those include pure λ-terms as well as normalizable
algebraic λ-terms as a particular case, and we show that all
hereditarily determinable terms are Taylor normalizable and the
coefficients of the normal form are given by a sequence of approximants, close
to the Böhm tree construction.

\subsection{Related and future work}

\label{subsection:related}

Besides the seminal work by Ehrhard and Regnier \cite{er:resource,er:bkt} in
the pure case, we have already cited previous approaches to the normalizability
of Taylor expansion based on finiteness conditions \cite{ehrhard:finres,ptv:taylorsn}.

A natural question to ask is how our generic notion of normal form of Taylor
expansion compares with previously introduced notions of denotation in non-deterministic
settings: non-deterministic Böhm trees \cite{deliguoro-piperno:ndlc},
probabilistic Böhm trees \cite{leventis:phd}, weighted relational models
\cite{de:pcoh,lmcp:weighted,laird:fixed}, \emph{etc.}
The very statement of such a question raises several difficulties, prompting
further lines of research.

One first obstacle is the fact that, by contrast with the uniform case of the
ordinary λ-calculus, the Taylor expansion operator is not injective on
algebraic λ-terms (see Subsection \ref{subsection:taylor}), not even on the
partial normal forms that we use to introduce the approximants in section
\ref{section:determinable}. This is to be related with the quotient that the
non-deterministic Böhm trees of de'Liguoro and Piperno
\cite{deliguoro-piperno:ndlc} must undergo in order to capture observational
equivalence. On the other hand, to our knowledge, finding
sufficient conditions on the semiring of scalars ensuring that the Taylor
expansion becomes injective is still an open question.

Also, we define normalizable vectors based on the notion of summability: a sum
of vectors converges when it is componentwise finite \ie{}, for each component,
only finitely many vectors have a non-zero coefficient (see subsection
\ref{subsection:summable}).  If more information is available on scalars,
namely if the semiring of scalars is complete in some topological or
order-theoretic sense, it becomes possible to normalize the Taylor expansion of
all terms.

Indeed, Tsukada, Asada and Ong have recently established \cite{tao:species} the
commutation between computing Böhm trees and Taylor expansion with
coefficients taken in the complete semiring of positive reals $[0,+\infty]$
where all sums converge.
Let us precise that they do not consider weighted non-determinism, only formal
binary sums of terms, and that the notion of Böhm tree they consider is a very
syntactic one, similar to the partial normal forms we introduce in section
\ref{section:determinable}.
Their approach is based on a precise description of the relationship
between the coefficients of resource terms in the expansion of a term and those in
the expansion of its Böhm-tree, using a \emph{rigid} Taylor expansion as an
intermediate step: this avoids the ambiguity between the sums of coefficients
generated by redundancies in the expansion and those representing
non-deterministic superpositions.

Tsukada, Asada and Ong’s work can thus be considered as a refinement of Ehrhard
and Regnier’s method, that they are moreover able to generalize to the
non-deterministic case provided the semiring of scalars is complete.
By contrast, our approach is focused on β-reduction and identifies a class of
algebraic λ-terms for which the normalization of Taylor expansion converges
independently from the topology on scalars.
It seems only natural to investigate the connections between both approaches,
in particular to tackle the case of weighted non-determinism in a complete
semiring, as a first step towards the treatment of probabilistic or quantum
superposition, as also suggested by the conclusion of their paper.

In the probabilistic setting, though, the Böhm tree
construction \cite{leventis:phd} relies on both the topological properties of
real numbers and the restriction to discrete probability subdistributions.
Relying on this, Dal Lago and Leventis have recently shown
\cite{dll:pbt} that the sum defining the normal form
of Taylor expansion of an arbitrary probabilistic λ-term always converges with
finite coefficients, and that this normal form is the Taylor expansion of its
probabilistic Böhm tree, in the non-extensional sense \cite[section
4.2.1]{leventis:phd}.
To get a better understanding of the shape of Taylor
expansions of probabilistic λ-terms and their stability under reduction, a possible
first step is to investigate probabilistic coherence spaces
\cite{de:pcoh} on resource λ-terms: these would be the analogue, in the probabilistic setting, of
the finiteness structures ensuring the summability of normal forms in the
non-deterministic setting (see Subsection \ref{subsection:normalizable:terms}).

Apart from relating our version of quantitative semantics with pre-existing
notions of denotation for non-deterministic λ-calculi, we plan to investigate
possible applications to other proof theoretic or computational frameworks:
namely, linear logic proof nets \cite{girard:ll} and infinitary λ-calculus
\cite{kksv:infinitary}.

The Taylor expansion of λ-terms can be generalized to linear logic proof nets:
the case of linear logic can even be considered as being more primitive, as it
is directly related with the structure of those denotational models that
validate the Taylor expansion formula \cite{ehrhard:dill}.
Proof nets, however, do not enjoy the uniformity property of λ-terms: 
no general coherence relation is satisfied by the elements of the Taylor expansion of a
proof net \cite[section V.4.1]{tasson:phd}. This can be related with the
non-injectivity of coherence semantics \cite{tortora:obsessional}.
In particular, it is really unclear how Ehrhard and Regnier’s methods, or even
Tsukada, Asada and Ong’s could be transposed to this setting.
By contrast, our recent work with Chouquet \cite{cv:antireduits-csl}
shows that our study of reduction under Taylor expansion can be adapted to
proof nets.

It is also quite easy to extend the Taylor expansion operator to infinite
λ-terms, at least for those of $\Lambda^{001}$, where only the argument
position of applications is treated coinductively.
For infinite λ-terms, it is no longer the case that the support of Taylor
expansion involves resource λ-terms of bounded height only.
Fortunately, we can still rely on the results of subsection
\ref{subsection:pbdresRed}, where we only require a bound on the nesting of
fired redexes: this should allow us to give a counterpart, through Taylor
expansion, of the strongly converging reduction sequences of infinite λ-terms.
More speculatively, another possible outcome is a characterization of
hereditarily head normalizable terms via their Taylor expansion, adapting our
previous work on normalizability with Pagani and Tasson \cite{ptv:taylorsn}.

\section{Technical preliminaries}

\label{section:preliminaries}

We write:
\begin{itemize}
	\item$\naturals$ for the semiring of natural numbers;
	\item$\powerset X$ for the powerset of a set $X$: $\calX\in\powerset X$ iff $\calX\subseteq X$;
	\item$\card X$ for the cardinal of a finite set $X$;
	\item$\oc X$ for the set of finite multisets of elements of $X$;
	\item$\mset{x_1,\dotsc,x_n}\in\oc X$ for the 
		multiset with elements $x_1,\dotsc,x_n\in X$
		(taking repetitions into account),
		and then $\card{\mset{x_1,\dotsc,x_n}}=n$ for its 
		cardinality;
	\item$\prod_{i\in I} X_i$ and $\sum_{i\in I} X_i$ respectively 
		for the product and sum of a family $\pars{X_i}_{i\in I}$ of sets:
		in particular $\sum_{i\in I} X_i=\Union_{i\in I}\set i\times X_i$;
	\item$X^I=\prod_{i\in I}X$ for the set of applications from $I$ to $X$ or,
		equivalently, for the set of $I$ indexed families of 
		elements of $X$.
\end{itemize}

Throughout the paper we will be led to consider various categories of sets and
elements associated with a single base set $X$: elements of $X$, subsets of
$X$, finite multisets of elements of $X$, \emph{etc.} In order to help keeping
track of those categories, we generally adopt the following typographic
conventions:
\begin{itemize}
	\item we use small latin letters for 
		the elements of $X$, say $a,b,c\in X$;
	\item for subsets of $X$, we use 
		cursive capitals, say $\calA,\calB,\calC\in\powerset X$;
	\item for sets of subsets of $X$, we use 
		Fraktur capitals, say $\fA,\fB,\fC\subseteq\powerset X$;
	\item for (possibly infinite) linear combinations of elements of $X$, we use
    small greek letters, say $α,β,γ\in \rigS^X$,
    where $\rigS$ denotes some set of scalar coefficients;
	\item we transpose all of the above conventions to the 
		set $\oc X$ of finite multisets by overlining:
		\eg, we write $\ms a=\mset{a_1,\dotsc,a_n}\in\oc X$, 
		$\ms\calA\subseteq\oc X$ or $\ms{α}\in\rigS^{\oc X}$.
\end{itemize}

In the remaining of this section, we introduce basic mathematical content that
will be used throughout the paper.

\subsection{Semirings and semimodules}

\label{subsection:semirings}

A \emph{semiring}\footnote{
	The terminology of semirings is much less well established 
	than that of rings, and one can find various non equivalent 
	definitions depending on the presence of units or on commutativity
	requirements. Following Golan’s terminology \cite{golan:semirings},
	our semirings are \emph{commutative semirings}, which is required 
	here because we consider multilinear applications between 
	modules.
} $\rigS$ is the data of a carrier set $\carrier {\rigS}$, together 
with commutative monoids $(\carrier {\rigS},+_{\rigS},0_{\rigS})$
and $(\carrier {\rigS},\cdot_{\rigS},1_{\rigS})$ such that
the multiplicative structure distributes over the additive one, \ie
for all $a,b,c\in \carrier {\rigS}$,
$a\mathbin{\cdot_{\rigS}} 0_{\rigS}=0_{\rigS}$ and
$a\cdot_{\rigS}(b+_{\rigS}c)=a\cdot_{\rigS} b+_{\rigS}a\cdot_{\rigS} c$.

We will in general abuse notation and identify $\rigS$ with its carrier set $\carrier{\rigS}$.
We will moreover omit the subscripts on symbols $+$, $\cdot$, $0$ and $1$, and 
denote multiplication by concatenation: $ab=a\cdot b$.
We also use standard notations for finite sums and products in $\rigS$, \eg
$\sum_{i=1}^n a_i=a_1+\cdots+a_n$.
For any semiring $\rigS$, there is a unique semiring morphism (in the obvious
sense) from $\naturals$ to $\rigS$: to $n\in\naturals$ we associate the sum
$\sum_{i=1}^n 1\in\rigS$ that we also write $n\in\rigS$, although this morphism
is not necessarily injective. Consider for instance the semiring $\booleans$
of booleans, with $\carrier{\booleans}=\set{0,1}$, $\mathord+_{\booleans} =
\mathord{\max}$ and $\mathord\cdot_{\booleans}=\mathord\times$.

We finish this subsection by recalling the definitions of 
semimodules and their morphisms.
A \emph{(left) $\rigS$-semimodule} $\calM$ is the data of a commutative
monoid $(\carrier{\calM},0_\calM,+_\calM)$ together with an
external product $\sm_\calM:\rigS\times \carrier{\calM}\to\carrier{\calM}$
subject to the following identities:
\begin{align*}
	0\sm_\calM m &= 0_\calM &
	1\sm_\calM m &= m \\
	(a+b)\sm_\calM m &= a\sm_\calM m+_\calM b\sm_\calM m&
	a\sm_\calM(b\sm_\calM m) &= ab\sm_\calM m \\
	a\sm_\calM 0_\calM &= 0_\calM &
	a\sm_\calM(m+_\calM n) &= a\sm_\calM m +_\calM a\sm_\calM n
\end{align*}
for all $a,b\in\rigS$ and $m,n\in\carrier{\calM}$.
Again, we will in general abuse notation and identify $\calM$ with its carrier set $\carrier{\calM}$,
and omit the subscripts on symbols $+$, $\sm$ and $0$.

Let $\calM$ and $\calN$ be $\rigS$-semimodules.
We say $\phi:\calM\to\calN$ is \emph{linear} if 
\[\phi\pars{\sum_{i=1}^n a_i\sm m_i}=\sum_{i=1}^n a_i\sm \phi\pars{m_i}\]
for all $m_1,\dotsc,m_n\in\calM$ and all $a_1,\dotsc,a_n\in\rigS$.
If moreover $\calM_1,\dotsc,\calM_n$ are $\rigS$-semimodules,
we say $\psi:\calM_1\times\cdots\times\calM_n\to\calN$ is \emph{$n$-linear}
if it is linear in each component.

Given a set $X$, $\vectors X$ is
the semimodule of formal linear combinations of elements of $X$:
a \emph{vector} $ξ\in\vectors X$ is nothing but an $X$-indexed 
family of scalars $\pars{ξ_x}_{x\in X}$, that we may 
also denote by $\sum_{x\in X} ξ_x\sm x$.
The \emph{support} $\support{ξ}$ of a vector $ξ\in\vectors X$ is the set of elements
of $X$ having a non-zero coefficient in $ξ$:
\[ \support{ξ}\eqdef\set{x \in X\st ξ_x\not=0}. \]
We write $\finiteVectors X$ for the set of vectors with finite support:
\[ \finiteVectors X\eqdef\set{ξ \in \vectors X\st \support{ξ}\text{ is finite}}. \]
In particular $\finiteVectors X$ is the semimodule freely generated by $X$,
and is a subsemimodule of $\vectors X$.

\subsection{Finiteness spaces}

\label{subsection:finiteness}

A finiteness space \cite{ehrhard:fs} is a subsemimodule of $\vectors X$
obtained by imposing a restriction on the support of vectors, as follows.

If $X$ is a set, we call \emph{structure on $X$} any set $\fS\subseteq\powerset X$, 
and then the dual structure is
\[\dual\fS\eqdef\set[\mnorm]{\calX'\subseteq X\st\text{ for all $\calX\in\fS$,
$\calX\inter \calX'$ is finite}}.\]
A \emph{relational finiteness space} is a pair $(X,\fF)$,
where $X$ is a set (the \emph{web} of the finiteness space) and
$\fF\subseteq\powerset{X}$ is a structure on $X$
such that $\fF=\bidual{\fF}$: $\fF$ is then called
a \emph{finiteness structure}, and we say $\calX \subseteq X$
is \emph{finitary} in $(X,\fF)$ iff $\calX\in\fF$.
The \emph{finiteness space} generated by $(X,\fF)$,
denoted by $\finitaryVectors{X,\fF}$,
or simply $\finitaryVectors{\fF}$,
is then the set of vectors on $X$ with finitary support:
$ξ\in\finitaryVectors{\fF}$ iff $\support{ξ}\in\fF$.
By this definition, if $ξ\in\finitaryVectors{\fF}$ and
$ξ'\in\finitaryVectors{\dual\fF}$ then the sum $\sum_{x\in X} ξ_xξ'_x$ involves
finitely many nonzero summands.

Finitary subsets are downwards closed for 
inclusion, and finite unions of finitary subsets are finitary,
hence $\finitaryVectors{X,\fF}$ is a subsemimodule of $\vectors X$.
Moreover, the least (resp. greatest) finiteness structure
on $X$ is the set $\finiteSubsets X$ of finite subsets of $X$ (resp. the powerset $\powerset X$),
generating the finiteness space $\finiteVectors X$ (resp.\ $\vectors X$).

We do not describe the whole category of finiteness spaces and linear-continuous
maps here. In particular we do not recall the details of the linear topology induced
on $\finitaryVectors{X,\fF}$ by $\fF$: the reader
may refer to Ehrhard’s original paper \cite{ehrhard:fs} or his survey presentation of
differential linear logic \cite{ehrhard:dill}.

In the following, we focus on a very particular case, where the
finiteness structure on base types is trivial (\ie there is no restriction on the
support of vectors):
linear-continuous maps are then univocally generated by \emph{summable functions}.

We started with the general notion of finiteness space
nonetheless, because it provides a good background for the general spirit of
our contributions: we are interested in infinite objects restricted so that,
componentwise, all our constructions involve finite sums only.
Also, the semimodule of normalizable resource vectors introduced in section
\ref{section:normalization} is easier to work with once its finiteness space
structure is exposed.

\subsection{Summable functions}

\label{subsection:summable}

Let $\vec {ξ}=\pars{ξ_i}_{i\in I}\in\pars{\vectors X}^I$ be a family of vectors:
write $ξ_i=\sum_{x\in X} ξ_{i,x}\sm x$.
We say $\vec {ξ}$ is \emph{summable}
if, for all $x\in X$,
$\set{i\in I\st x\in\support{ξ_i}}$ is finite.
In this case, we define the sum $\sum\vec{ξ}=\sum_{i\in I}ξ_i\in\vectors X$
in the obvious, pointwise way:\footnote{
	The reader can check that the family $\vec{ξ}$ is summable iff 
	the support set \[ \set{(i,x)\in I\times X\st ξ_{i,x}\not =0} \] is finitary
	in the relational arrow finiteness space $(I\times X,\powerset I\multimap\powerset X)$
	as defined by Ehrhard \cite[see in particular Lemma 3]{ehrhard:fs}.
	Then $\sum\vec{ξ}$ is the result of applying the matrix $\pars{ξ_{i,x}}_{i\in I,x\in X}$
	to the vector $\pars{1}_{i\in I}\in\finitaryVectors{\powerset I}=\vectors I$.
} \[ \pars{\sum\vec{ξ}}_x \eqdef \sum_{i\in I} ξ_{i,x}.\]

Of course, any finite family of vectors is summable and, fixing an index set
$I$ and a base set $X$, summable families in $\pars{\vectors X}^I$ form an
$\rigS$-semimodule, with operations defined pointwise.

Moreover, if $\pars{ξ_i}_{i\in I}\in\pars{\vectors X}^I$ is summable,
then it follows from the inclusion
$\support{a_i\sm ξ_i}\subseteq \support{ξ_i}$ that $\pars{a_i\sm ξ_i}_{i\in I}$
is also summable for any family of scalars $\pars{a_i}_{i\in I}\in\rigS^I$.
Whenever the $n$-ary function
$f:X_1\times\cdots\times X_n\to\vectors Y$ 
(\ie the family
$\pars{f(x_1,\dotsc,x_n)}_{\pars{x_1,\dotsc,x_n}\in X_1\times\cdots\times X_n}$)
is summable,
we can thus define its \emph{extension} $\contExt f:\vectors {X_1}\times\cdots\times\vectors{X_n}\to\vectors Y$
by \[\contExt f(ξ_1,\dotsc,ξ_n)\eqdef
\sum_{\pars{x_1,\dotsc,x_n}\in X_1\times\cdots\times X_n} ξ_{1,x_1}\cdots ξ_{n,x_n}\sm f(x_1,\dotsc,x_n).\]

Note that we can consider $f:X\to\vectors Y$ as a $Y\times X$ matrix:
$f_{y,x}=f(x)_y$. Then if $f$ is summable and $ξ\in\vectors X$,
$\contExt f({ξ})$ is nothing but the application of the matrix
$f$ to the column $ξ$: the summability hypothesis ensures that 
this is well defined.

It turns out that the linear extensions of summable functions 
are exactly the linear-continuous maps, defined as follows:
\begin{definition}
	\label{definition:continuous}
	Let $φ:\vectors {X_1}\times\cdots\times\vectors{X_n}\to\vectors Y$.
	We say ${φ}$ is \emph{$n$-linear-continuous} if, for all
	summable families
	$\vec{ξ_1}=\pars{ξ_{1,i}}_{i\in I_1}\in\pars{\vectors {X_1}}^{I_1},
	\dotsc,
	\vec{ξ_n}=\pars{ξ_{n,i}}_{i\in I_n}\in\pars{\vectors {X_n}}^{I_n}$,
	the family
	$\pars{φ(ξ_{1,i_1},\dotsc,ξ_{n,i_n})}_{(i_1,\dotsc,i_n)\in I_1\times\cdots\times I_n}$ is summable and,
	for all families of scalars,
	$\vec{a_1}=\pars{a_{1,i}}_{i\in I_1}\in{\rigS}^{I_1},
	\dotsc,
	\vec{a_n}=\pars{a_{n,i}}_{i\in I_n}\in{\rigS}^{I_n}$,
	we have 
	\[φ\pars{
		\sum_{i_1\in I_1} a_{1,i_1}\sm ξ_{1,i_1},
		\dotsc,
	  \sum_{i_n\in I_n} a_{n,i_n}\sm ξ_{n,i_n}
	}=\sum_{(i_1,\dotsc,i_n)\in I_1\times\cdots\times I_n}
a_{1,i_1}\cdots a_{n,i_n}\sm {φ}(ξ_{1,i_1},\dotsc,ξ_{n,i_n}).\]
\end{definition}

\begin{lemma}
	If $φ:\vectors {X_1}\times\cdots\times\vectors{X_n}\to\vectors Y$
	is $n$-linear-continuous then its restriction
	$\restr {φ}{X_1\times\cdots\times X_n}$
		is a summable $n$-ary function and 
	$φ=\contExt{\restr {φ}{X_1\times\cdots\times X_n}}$.
	Conversely, if $f:X_1\times\cdots\times X_n\to\vectors Y$ is 
	a summable $n$-ary function then $\contExt f$ is $n$-linear-continuous.
\end{lemma}
\begin{proof}
	It is possible to derive both implications from general results 
	on finiteness spaces.\footnote{
		One might check that a map
		$φ:\vectors {X_1}\times\cdots\times\vectors{X_n}\to\vectors Y$
		is $n$-linear-continuous in the sense of Definition \ref{definition:continuous}
		iff it is $n$-linear and continuous in the sense of the linear topology of
		finiteness spaces, observing that the topology on $\vectors
		X=\finitaryVectors{\powerset X}$ is the product topology ($\rigS$ being
		endowed with the discrete topology) \cite[Section 3]{ehrhard:fs}.
		Moreover, $n$-ary summable functions $f:X_1\times\cdots\times X_n\to\vectors Y$
		are the elements of the finiteness space
		$\finitaryVectors{\powerset{X_1}\otimes\cdots\otimes\powerset{X_n}\multimap\powerset{Y}}$.
		As a general fact, the
		linear-continuous maps $\finitaryVectors{\fF}\to\finitaryVectors{\fG}$ are exactly
		the linear extensions of vectors in $\finitaryVectors{\fF\multimap\fG}$.
		But linear-continuous maps from a tensor product of
		finiteness spaces correspond with multi\-linear-\emph{hypocontinuous} maps
		\cite[Section 3]{ehrhard:fs}
		rather than the more restrictive multi\-linear-continuous maps.
		In the very simple setting of summable functions, though, both notions
		coincide, since $\vectors X$ is always locally linearly compact 
		\cite[Proposition 15]{ehrhard:fs}.
	}
	We also sketch a direct proof.

	The first implication follows directly from the definitions, observing 
	that each diagonal family of vectors $(x)_{x\in X_i}$ is obviously summable.

	For the converse: 
	let $\vec{ξ_1}=\pars{ξ_{1,i}}_{i\in I_1}\in\pars{\vectors {X_1}}^{I_1},
	\dotsc,
	\vec{ξ_n}=\pars{ξ_{n,i}}_{i\in I_n}\in\pars{\vectors {X_n}}^{I_n}$
	be summable families.
	We first prove that the family
	\[\pars{
	        ξ_{1,i_1,x_1}\cdots ξ_{n,i_n,x_n}\sm f(x_1,\dotsc,x_n)
	}_{(i_1,\dotsc,i_n)\in I_1\times\cdots\times I_n,(x_1,\dotsc,x_n)\in X_1\times\cdots\times X_n}\]
	is summable. Fix $y\in Y$.
	If $y\in\support{
	        ξ_{1,i_1,x_1}\cdots ξ_{n,i_n,x_n}\sm f(x_1,\dotsc,x_n)
	}$ then in particular
	$y\in\support{f(x_1,\dotsc,x_n)}$: since $f$ is summable, there are finitely many
	such tuples $(x_1,\dotsc,x_n)\in X_1\times\cdots\times X_n$.
	For each such tuple $(x_1,\dotsc,x_n)$ and each $k\in\set{1,\dotsc,n}$,
	since $\vec{ξ}_k$ is summable, there are finitely many $i_k$'s such that
	$ξ_{k,i_k,x_k}\not=0$. The necessary equation then follows from the
	associativity of sums.
\end{proof}

From now on, we will identify summable functions with their
multilinear-continuous extensions. Moreover, it should be clear that
multilinear-continuous maps compose.

\section{The resource λ-calculus}

\label{section:resource}

In this section, we recall the syntax and reduction of the resource λ-calculus,
that was introduced by Ehrhard and Regnier \cite{er:resource} as the
multilinear fragment of the differential λ-calculus \cite{er:tdlc}.
The syntax is very similar to that of Boudol’s resource λ-calculus
\cite{boudol:resource} but the intended meaning (multilinear approximations of
λ-terms) as well as the dynamics  is fundamentally different.

We also recall the definitions of the multilinear counterparts of term substitution:
partial differentiation and multilinear substitution.

In the passing, we introduce various quantities on resource λ-terms (size,
height, and number and maximum depth of occurrences of a variable) and we state
basic results that will be used throughout the paper.

Finally, we present the dynamics of the calculus: resource reduction and 
normalization.

\subsection{Resource expressions}

\label{subsection:resourceExpressions}

In the remaining of the paper, we suppose an infinite, countable set
$\variables$ of variables is fixed: we use small letters $x,y,z$ to denote
variables.

We define the sets $\resourceTerms$ of \emph{resource terms}
and $\resourceMonomials$ of \emph{resource monomials} by mutual induction as follows:\footnote{
We use a self explanatory if not standard variant of BNF notation
for introducing syntactic objects:
\[ \resourceMonomials \ni \ms s,\ms t,\ms u,\ms v,\ms w \recdef \mset {} \mid \mset s\cdot\ms t\]
means that we define the set $\resourceMonomials$ of resource monomials
as that inductively generated by the empty monomial, and 
addition of a term to a monomial, and that we will denote resource monomials 
using overlined letters among $\ms s,\ms t,\ms u,\ms v,\ms w$, possibly with
sub- and superscripts.
}
\[
	\begin{array}{rclcl}
		\resourceTerms & \ni & s,t,u,v,w & \recdef & x \mid \labs xs \mid \rappl s{\ms t} \\
		\resourceMonomials  & \ni & \ms s,\ms t,\ms u,\ms v,\ms w & \recdef & \mset {} \mid \mset s\cdot\ms t.
	\end{array}
\]

Terms are considered up to α-equivalence and monomials up to permutativity:
we write $\mset{t_1,\dotsc,t_n}$ for $\mset{t_1}\cdot\pars{\cdots\cdot\pars{\mset{t_n}\cdot\mset{}}}$
and equate $\mset{t_1,\dotsc,t_n}$ with $\mset{t_{f(1)},\dotsc,t_{f(n)}}$ for all 
permutation $f$ of $\set{1,\dotsc,n}$,
so that resource monomials coincide with finite multisets of resource terms.\footnote{
	Resource monomials are often called \emph{bags}, \emph{bunches} or \emph{poly-terms} in the
	literature, but we prefer to strengthen the analogy with power series here.
}
We will then write $\ms s\cdot\ms t$ for the multiset union of $\ms s$ and $\ms t$,
and $\card{\mset{s_1,\dotsc,s_n}}\eqdef n$.

We call \emph{resource expression} any resource term or resource monomial
and write $\resourceExpressions$ for either $\resourceTerms$ or $\resourceMonomials$:
whenever we use this notation several times in the same context, 
all occurrences consistently denote the same set. When we make a definition 
or a proof by induction on resource expressions, we actually use a mutual induction 
on resource terms and monomials.

\begin{definition}
We define by induction over a resource expression
$e\in\resourceExpressions$, its \emph{size} $\size{e}\in\naturals$
and its \emph{height} $\height{e}\in\naturals$:
\begin{align*}
	\size{x}&\eqdef1&
	\height{x}&\eqdef1\\
	\size{\labs xs}&\eqdef 1+\size{s}&
	\height{\labs xs}&\eqdef 1+\height{s}\\
	\size{\rappl{s}{\ms t}}&\eqdef 1+\size s+\size{\ms t}&
	\height{\rappl{s}{\ms t}}&\eqdef \max\set{\height s, 1+\height{\ms t}}\\
	\size{\mset{s_1,\dotsc,s_n}}&\eqdef \sum_{i=1}^n \size{s_i}&
	\height{\mset{s_1,\dotsc,s_n}}&\eqdef \max\set{\height{s_i}\st 1\le i\le n}.
\end{align*}
\end{definition}
It should be clear that, for all $e\in\resourceExpressions$,
$\height e\le\size e$. Also observe that $\size s>0$ and $\height s>0$ for 
all $s\in\resourceTerms$, and $\size{\ms s}\ge\card{\ms s}$ for all $\ms s\in\resourceMonomials$.
In the application case, we chose not to increment the height of the function:
this is not crucial but it will allow to simplify some of our 
computations in Section \ref{section:parallel}. In particular, 
in the case of a redex we have
$\height{\rappl{\labs xs}{\ms t}}=1+\max\set{\height s,\height{\ms t}}$.

For all resource expression $e$, we write
$\fv e$ for the set of its free variables.
In the remaining of the paper, we will often have to prove that 
some set $\calE\subseteq\resourceExpressions$ is finite:
we will generally use the fact that $\calE$ is finite
iff both $\set{\size e\st e\in\calE}$ and $\fv{\calE}\eqdef\Union_{e\in\calE}\fv e$
are finite.

Besides the size and height of an expression,
we will also need finer grained information on
occurrences of variables, providing a quantitative counterpart to the 
set of free variables:
\begin{definition}
	\label{definition:occ}
	We define by induction over resource expressions 
	the \emph{number $\occnum xe\in\naturals$
		of occurrences} and the \emph{set $\occdepth xe\in\naturals$ of
	occurrence depths} of a variable $x$ in $e\in\resourceExpressions$:
\begin{align*}
	\occnum x{y}&\eqdef \begin{cases}1&\text{if $x=y$}\\0&\text{otherwise}\end{cases} \\
	\occnum x{\labs ys}&\eqdef  \occnum x{s} && (\text{choosing }y\not=x)\\
	\occnum x{\rappl{s}{\ms t}}&\eqdef \occnum x{s}+\occnum x{\ms t} \\
	\occnum x{\mset{s_1,\dotsc,s_n}}&\eqdef  \sum_{i=1}^n \occnum x{s_i} 
\end{align*}
and
\[\begin{array}{rlr}
	\occdepth x{y}&\eqdef \begin{cases}\set 1&\text{if $x=y$}\\\emptyset&\text{otherwise}\end{cases} \\
	\occdepth x{\labs ys}&\eqdef  \set{d+1\st d\in\occdepth x{s}} & (\text{choosing }y\not=x)\\
	\occdepth x{\mset{s_1,\dotsc,s_n}}&\eqdef  \Union_{i=1}^n\occdepth x{s_i}\\
	\occdepth x{\rappl{s}{\ms t}}&\eqdef  \occdepth x s \union \set{d+1\st d\in\occdepth x{\ms t}}.
\end{array}\]
We then write $\maxoccdepth xe\eqdef \max\occdepth xe$
for the \emph{maximal depth of occurrences} of $x$ in $e$.
\end{definition}

Again, it should be clear that
$\occnum xe\le\size e$ and $\maxoccdepth xe\le\height e$.
Moreover,
$x\in\fv e$ iff $\occnum xe\not=0$ iff $\occdepth xe\not=\emptyset$ iff $\maxoccdepth xe\not=0$.

\subsection{Partial derivatives}

In the resource λ-calculus, the substitution $\subst exs$ of a term $s$ for a variable $x$ in $e$ admits a
linear counterpart: this operator was initially introduced in the differential
λ-calculus \cite{er:tdlc} in the form of a partial differentiation operation,
reflecting the interpretation of λ-terms as analytic maps in quantitative
semantics.

Partial differentiation enforces the introduction of formal finite sums of
resource expressions: these are the actual objects of the resource λ-calculus,
and in particular the dynamics will act on finite sums of terms rather than on
simple resource terms (see Subsection \ref{subsection:resRed}).
We extend all syntactic constructs to finite sums of resource expressions
by linearity: if $σ=\sum_{i=1}^n s_i\in\finiteTermSums$
and $\ms{τ}=\sum_{j=1}^p \ms t_j\in\finiteMonomialSums$,
we set $\labs x{σ}\eqdef\sum_{i=1}^n\labs x{s_i}$, 
$\rappl{σ}{\ms{τ}}\eqdef\sum_{i=1}^n\sum_{j=1}^p \rappl{s_i}{\ms t_j}$
and $\mset{σ}\cdot{\ms{τ}}\eqdef\sum_{i=1}^n\sum_{j=1}^p \mset{s_i}\cdot{\ms t_j}$.

This linearity of syntactic constructs will be generalized to vectors
of resource expressions in the next section.
For now, up to linearity, it is already possible to consider the substitution
$\subst ex{σ}$ of a finite sum of terms $σ$ for a variable term $x$ in an
expression $e$: in particular $\subst ex0=0$ whenever $x\in\fv e$.
This is in turn extended to sums by linearity: $\subst{ε}x{σ}=\sum_{i=1}^n
\subst{e_i}x{σ}$ when $ε=\sum_{i=1}^n e_i$. Observe that this is \emph{not}
linear in $σ$, because $x$ may occur several times in $e$:
for instance, with a monomial of degree 2,
$\subst{\mset{x,x}}x{t+u}=\mset{t,t}+\mset{t,u}+\mset{u,t}+\mset{u,u}$.

Partial differentiation is then defined as follows:
\begin{definition}
For all $u\in\resourceTerms$ and $x\in\variables$,
we define the \emph{partial derivative} $\pdiff exu\in\finiteResourceSums$
of $e\in\resourceExpressions$, by induction on $e$:
\begin{align*}
	\pdiff{y}xu&\eqdef \begin{cases}u&\text{if $x=y$}\\0&\text{otherwise}\end{cases}\\
	\pdiff{\labs ys}xu&\eqdef \labs y[]{\pdiff sxu}&&\text{(choosing $y\not\in\set x\union\fv{u}$)}\\
	\pdiff{\rappl{s}{\ms t}}xu&\eqdef  \rappl{\pdiff sxu}{\ms t}+\rappl{s}[]{\pdiff {\ms t}xu}\\
	\pdiff{\mset{s_1,\dotsc,s_n}}xu&\eqdef  \sum_{i=1}^n \mset{s_1,\dotsc,\pdiff{s_i}xu,\dotsc,s_n}.
\end{align*}
\end{definition}

Partial differentiation is extended to 
finite sums of expressions by bilinearity: if
$ε=\sum_{i=1}^n e_i\in\finiteResourceSums$ and 
$σ=\sum_{j=1}^p s_j\in\finiteTermSums$,
we set \[\pdiff{ε}x{σ}=\sum_{i=1}^n\sum_{j=1}^p \pdiff {e_i}x{s_j}.\]
\begin{lemC}[{\cite[Lemma 2]{er:resource}}]
	\label{lemma:pdiff:schwarz}
	If $x\not\in\fv u$ then 
	\[
		\pDiff{\pars{\pdiff{e}xt}}yu
		= \pDiff{\pars{\pdiff{e}yu}}xt 
		+ \pdiff ex{\pars{\pdiff{t}yu}}
	.\]
\end{lemC}

If moreover $y\not\in\fv{t}$, we obtain a version of Schwarz’s theorem on
the symmetry of second derivatives:
	\[
		\pDiff{\pars{\pdiff{e}xt}}yu
		= \pDiff{\pars{\pdiff{e}yu}}xt 
	.\]
If $x\not\in\fv{s_i}$ for all $i\in\set{1,\dotsc,n}$, we write 
\[\npdiff n{e}x{\pars{s_1,\dotsc,s_n}}\eqdef
\pDiff{\pars{\cdots \pdiff {e}x{s_1}\cdots }}x{s_n}.\]
More generally, we write
\[\npdiff n{e}x{\pars{s_1,\dotsc,s_n}}\eqdef\subst{\pars{\npdiff n{\subst exy}y{\pars{s_1,\dotsc,s_n}}}}yx\]
for any $y\not\in\Union_{i=1}^n \fv{s_i} \union(\fv{e}\setminus\set x)$: it should 
be clear that this definition does not depend on the choice of such a variable $y$.
By the previous lemma,
\[\npdiff n{e}x{\pars{s_1,\dotsc,s_n}}=
\npdiff n{e}x{\pars{s_{f(1)},\dotsc,s_{f(n)}}}\]
for any permutation $f$ of $\set{1,\dotsc,n}$ and we will thus write
\[\npdiff n{e}x{\ms s}\eqdef \npdiff n{e}x{\pars{s_1,\dotsc,s_n}}\]
whenever $\ms s=\mset{s_1,\dotsc,s_n}$.

An alternative, more direct presentation of iterated partial derivatives
is as follows.
Suppose $\occnum xe=m$, and write $x_1,\dotsc,x_m$ for the occurrences of $x$ in $e$.
Then:
	\[\npdiff nex{\mset{s_1,\dotsc,s_n}}=
		\sum_{\substack{f:\set{1,\dotsc,n}\to\set{1,\dotsc,m}\\\text{$f$ injective}}} 
			\subst{e}{x_{f(1)},\dotsc,x_{f(n)}}{s_1,\dotsc,s_n}
	\]
More formally, we obtain:
\begin{lemma}
	\label{lemma:npdiff:decomposition}
	For all monomial $\ms u=\mset{u_1,\dotsc,u_n}\in\resourceMonomials$ and all variable
	$x\in\variables$:\footnote{
		In this definition and in the remaining of the paper, we say a tuple
		$(I_1,\dotsc,I_n)\in\powerset I^n$ is a partition of $I$
		if $I=\Union_{i=1}^n I_k$, and the $I_k$'s are pairwise disjoint.
		We do not require the $I_k$’s to be nonempty. Hence a partition of $I$
		into a $n$-tuple is uniquely defined by a function from $I$ to $\set{1,\dotsc,n}$.
	}
	\begin{align*}
		\npdiff n{y}x{\ms{u}}  &= \begin{cases}
			y&\text{if $n=0$}\\
			u_1&\text{if $x=y$ and $n=1$}\\
			0&\text{otherwise}
		\end{cases}\\
		\npdiff n{\labs y s}x{\ms{u}}&=
		\labs y[]{\npdiff nsx{\ms u}}&&&\text{(choosing $y\not\in\set x\union\fv{\ms u}$)}\\
		\npdiff n{\rappl s{\ms t}}x{\ms u}&=
		\lefteqn{
			\sum_{(I,J)\text{ partition of }\set{1,\dotsc,n}}
			\rappl{\npdiff {\card I}sx{\ms u_I}}
				{\npdiff {\card J}{\ms t}x{\ms u_J}}
		}\hspace{3em}\\
		\npdiff n{\mset{s_1,\dotsc,s_k}}x{\ms u}&=
		\lefteqn{
			\sum_{
				(I_1,\dotsc,I_k)\text{ partition of }\set{1,\dotsc,n}
			}
			\mset{
				\npdiff {\card{I_1}}{s_1}x{\ms u_{I_1}}
				\dotsc,
				\npdiff {\card{I_k}}{s_k}x{\ms u_{I_k}}
			}
		}\hspace{8em}
	\end{align*}
	where $\ms u_I$ denotes ${\mset{u_{i_1},\dotsc,u_{i_p}}}$
	whenever $I=\set{i_1,\dotsc,i_p}$ with $p=\card I$.
\end{lemma}
\begin{proof}
	Easy, by induction on $n$.
\end{proof}

\begin{lemma}
	\label{lemma:npdiff:size}
	For all $e\in\resourceExpressions$, $\ms s\in\resourceMonomials$, $x\not=y\in\variables$
	and ${e'}\in\support{\npdiff n{e}x{\ms s}}$ with $n=\card{\ms s}$,
	moreover assuming that $x\not\in\fv{\ms s}$:
	\begin{itemize}
		\item $\occnum x{e}\ge n$ and $\occnum x{e'} = \occnum xe - n$;
		\item $\occnum y{e'} = \occnum ye + \occnum y{\ms s}$;
		\item $\occdepth x{e'}\subseteq\occdepth xe$;
		\item $\occdepth y{e} \subseteq \occdepth y{e'}\subseteq
			\occdepth ye\union\set{d+d'-1\st d\in\occdepth xe,d'\in\occdepth y{\ms s}}$;
		\item $\size {e'} = \size e + \size{\ms s} - n$;
		\item $\height e\le\height{e'}\le\max\set{\height e,\maxoccdepth xe+\height{\ms s}-1}$.
	\end{itemize}
\end{lemma}
\begin{proof}
	Each result is easily established by induction on $e$,
	using the previous lemma to enable the induction.
\end{proof}

\subsection{Multilinear substitution}

Recall that Taylor expansion involves iterated derivatives at $0$.
If $n=\card {\ms s}$ and $x\not\in\fv {\ms s}$ we write 
\[\lsubst ex{\ms s}\eqdef \subst{\pars{\npdiff {n} ex{\ms s}}}x0.\]
Observe that by Lemma \ref{lemma:npdiff:size}: if $n>\occnum xe$ then $\npdiff n{e}x{\ms s}=0$;
and if $n<\occnum xe$ then $x\in\fv{e'}$ for all $e'\in\support{\npdiff nex{\ms s}}$,
and then $\subst{e'}x{0}=0$.
In other words, 
\[\lsubst ex{\ms s}= \begin{cases}
			\npdiff nex{\ms s}&
				\text{if $n=\occnum xe$}\\
				0&\text{otherwise}
		\end{cases}\quad.\]
We say $\lsubst ex{\ms s}$ is the \emph{$n$-linear substitution}
of $\ms s$ for $x$ in $e$.
More generally, we write
\[\lsubst ex{\ms s}\eqdef\subst{\pars{\lsubst {\subst exy}y{\ms s}}}yx\]
for any $y\not\in\fv{\ms s}\union(\fv{e}\setminus x)$ and it should again
be clear that this definition does not depend of the choice of such a $y$.
By a straightforward application of Lemma \ref{lemma:npdiff:size}, we obtain:

\begin{lemma}
	\label{lemma:lsubst:size}
	For all $e\in\resourceExpressions$, $\ms s\in\resourceMonomials$, $x\not=y\in\variables$
	and ${e'}\in\support{\lsubst ex{\ms s}}$, assuming $x\not\in\fv{\ms s}$:
	\begin{itemize}
		\item $\occnum x{e}=\card{\ms s}$ and $\occnum x{e'}=0$;
		\item $\occnum y{e'} = \occnum ye + \occnum y{\ms s}$;
		\item $\occdepth x{e'} = \emptyset$;
		\item $\occdepth y{e} \subseteq \occdepth y{e'} \subseteq
			\occdepth y{e}\union\set{d+d'-1\st d\in\occdepth xe,d'\in\occdepth y{\ms s}}$;
		\item $\size {e'} = \size e + \size{\ms s} - \card{\ms s}$;
		\item $\height e\le\height{e'}\le\max\set{\height e,\maxoccdepth xe+\height{\ms s}-1}$.
	\end{itemize}
\end{lemma}
In particular, $\fv {e'}=(\fv e\setminus \set x)\union \fv{\ms s}$, and
$\max\set{\size{e},\size{\ms s}}\le\size{e'}\le\size{e}+\size{\ms s}$.

Again, we can give a direct presentation of multilinear substitution.
Suppose $\occnum xe=m$, and 
write $x_1,\dotsc,x_m$ for the occurrences of $x$ in $e$. Then:
\[
	\lsubst ex{\mset{s_1,\dotsc,s_n}} =
		\sum_{\substack{f:\set{1,\dotsc,n}\to\set{1,\dotsc,m}\\\text{$f$ bijective}}}
			\subst{e}{x_{f(1)},\dotsc,x_{f(n)}}{s_1,\dotsc,s_n}.
\]
More formally, as a consequence of Lemma \ref{lemma:npdiff:decomposition}:
\begin{lemma}
	\label{lemma:lsubst:decomposition}
	For all monomial $\ms u=\mset{u_1,\dotsc,u_n}\in\resourceMonomials$ and all variable
	$x\in\variables$:
	\begin{align*}
		\lsubst {y}x{\ms{u}}  &= \begin{cases}
			y&\text{if $y\not=x$ and $n=0$}\\
			u_1&\text{if $y=x$ and $n=1$}\\
			0&\text{otherwise}
		\end{cases}\\
		\lsubst{\labs y s}x{\ms{u}}&=
			\labs y[]{\lsubst{s}x{\ms u}}&&&\text{(choosing $y\not\in\set x\union\fv{\ms u}$)}\\
		\lsubst{\rappl s{\ms t}}x{\ms u}&=
		\lefteqn{
			\sum_{\substack{(I,J)\text{ partition of }\set{1,\dotsc,n}\\
				\text{ s.t. }\card I=\occnum x{s}\text{ and }\card J=\occnum x{\ms t}}}
			\rappl{\lsubst sx{\ms u_I}}
				{\lsubst {\ms t}x{\ms u_J}}
		}\\
		\lsubst{\mset{s_1,\dotsc,s_k}}x{\ms u}&=
		\lefteqn{
			\sum_{\substack{
				(I_1,\dotsc,I_k)\text{ partition of }\set{1,\dotsc,n}\\
				{
				\text{ s.t. }\forall j,\ \card {I_j}=\occnum x{s_j}
				}
			}}
			\mset{
				{\lsubst {s_1}x{\ms u_{I_1}}},
				\dotsc,
				{\lsubst {s_k}x{\ms u_{I_k}}}
			}
		}
	\end{align*}
	where the conditions on cardinalities of subsets of $\set{1,\dotsc,n}$
	in the application and monomial cases may be omitted.
\end{lemma}
A similar result is the commutation of multilinear substitutions:
\begin{lemma}
	\label{lemma:lsubst:commute}
	If $x\not\in\fv{\ms u}$ then:
	\[
		\lsubst{\pars{\lsubst ex{\ms t}}}y{\ms u}=
			\sum_{\substack{(I,J)\text{ partition of }\set{1,\dotsc,\card{\ms u}}\\
				\text{ s.t. }\card I=\occnum x{e}\text{ and }\card J=\occnum x{\ms t}}}
			\lsubst {\pars{\lsubst ey{\ms u_I}}}x{\pars{\lsubst {\ms t}y{\ms u_J}}}
		.
	\]
\end{lemma}
\begin{proof}
	Write $n=\card{\ms t}$ and $p=\card{\ms u}$.
	It is sufficient to prove 
	\[
		\npDiff p{\pars{\npdiff nex{\ms t}}}y{\ms u}=
			\sum_{(I,J)\text{ partition of }\set{1,\dotsc,p}}
			\npDiff {n} {\pars{\npdiff {\card I}sy{\ms u_I}}}x{\pars{\npdiff {\card J}{\ms t}y{\ms u_J}}}
	\]
	by induction on $n$ and $p$, using Lemma \ref{lemma:pdiff:schwarz}.
\end{proof}

\subsection{Resource reduction}

\label{subsection:resRed}

If $\reduce$ is a reduction relation, we will write $\R\reduce$ (resp.
$\T\reduce$; $\RT\reduce$) for its reflexive (resp.\ transitive; reflexive and
transitive) closure.

In the resource λ-calculus, a redex is a term of the form
$\rappl{\labs xt}{\ms u}\in\resourceTerms$ and its reduct is
$\lsubst tx{\ms u}\in\finiteTermSums$.
The resource reduction $\resRed$ is then the contextual closure
of this reduction step on finite sums of resource expressions. More precisely:
\begin{definition}
	We define the \emph{resource reduction} relation
	$\mathord{\resRed}\subseteq\resourceExpressions\times\finiteResourceSums$
	inductively as follows:
	\begin{itemize}
		\item $\rappl{\labs x s}{\ms t}\resRed \lsubst sx{\ms t}$ for all $s\in\resourceTerms$ 
			and $\ms t\in\resourceMonomials$;
		\item $\labs xs\resRed \labs x {σ'}$ as soon as $s\resRed σ'$;
		\item $\rappl s{\ms t}\resRed \rappl{σ'}{\ms{t}}$ as soon as $s\resRed σ'$;
		\item $\rappl s{\ms t}\resRed \rappl{s}{\ms{τ}'}$ as soon as $\ms t\resRed \ms{τ}'$;
		\item $\mset{s}\cdot \ms t\resRed\mset{σ'}\cdot\ms t$ as soon as $s\resRed σ'$.
	\end{itemize}
	We extend this reduction to finite sums of resource expressions:
	write $ε\resRed ε'$ if $ε=\sum_{i=0}^n e_i$ and $ε'=\sum_{i=0}^n ε'_i$ 
	with $e_0\resRed ε'_0$ and, for all $i\in\set{1,\dotsc,n}$,
	$e_i\resRedR ε'_i$.
\end{definition}
Observe that we allow for parallel reduction of any nonzero number of summands in a finite sum.
This reduction is particularly well behaved. In particular, it is confluent in a strong sense:
\begin{lemma}
	\label{lemma:reduction:strongconfluence}
	For all $ε,ε_0,ε_1\in\finiteResourceSums$,
	if $ε\resRed ε_0$ and $ε\resRed ε_1$ then
	there is $ε'\in\finiteResourceSums$ such 
	that $ε_0\R\resRed ε'$ and $ε_1\R\resRed ε'$.
\end{lemma}
\begin{proof}
	The proof follows a well-trodden path for proving confluence.

	One first proves by induction on $s$ that if $s\resRed σ'$
	then $\lsubst sx{\ms t}\R\resRed\lsubst{σ'}x{\ms t}$,
	and if $\ms t\resRed \ms{τ}'$ 
	then $\lsubst sx{\ms t}\R\resRed\lsubst{s}x{\ms{τ}'}$.
	Note that the reflexive closure is made necessary by 
	the possibility that $\lsubst sx{\ms t}=0$,
	and the transitive closure is not needed because 
	there is no duplication of the redexes of $\ms t$ in the 
	summands of the multilinear substitution $\lsubst sx{\ms t}$.

	One then proves that if $e\resRed ε_0$ and $e\resRed ε_1$ then
	there is $ε'\in\finiteResourceSums$ such 
	that $ε_0\R\resRed ε'$ and $ε_1\R\resRed ε'$.
	The proof is straightforward, by induction on the pair of reductions
	$e\resRed ε_0$ and $e\resRed ε_1$, using the previous 
	result in case $e$ is a redex which is reduced in 
	$ε_0$ but not in $ε_1$ (or \emph{vice versa}).
\end{proof}
In other words, $\R\resRed$ enjoys the diamond property.\footnote{
	This strong confluence result was not mentioned 
	in Ehrhard and Regnier's papers about resource λ-calculus
	\cite{er:resource,er:bkt} but they established 
	a very similar result for differential nets \cite[Section 4]{er:diffnets}:
	Lemma \ref{lemma:reduction:strongconfluence} can be understood as a
	reformulation of the latter in the setting of resource calculus.
}
Moreover, the effect of reduction on the size of terms is very regular.
First introduce some useful notation:
write $e\oneStepGenerates e'$ if $e\resRed ε'$ with $e'\in \support{ε'}$.

\begin{lemma}
	\label{lemma:reduction:size}
	Let $e\oneStepGenerates e'$.
	Then $\fv{e'}=\fv{e}$, and $\size{e'}+2\le\size{e}\le 2\size{e'}+2$.
\end{lemma}
\begin{proof}
	By induction on the reduction $e\resRed ε'$ with $e'\in\support{ε'}$.
	The inductive contextuality cases are easy, 
	and we only detail the base case, \ie
	$e=\rappl{\labs x{t}}{\ms u}$ 
	and $ε'=\lsubst {t}x{\ms u}$.

	Write $n=\occnum	xt$.
	The result then follows from Lemma \ref{lemma:lsubst:size}, observing that
	$\size{e'}=\size{t}+\size{\ms u}-n=\size{e}-2-n$ and $n\le\size{\ms u}\le\size{e'}$.
\end{proof}

We will write $\generates$ (resp.\ $\strictlyGenerates$) for $\oneStepGenerates^*$
(resp.\ $\oneStepGenerates^+$).
Observe that $e\generates e'$ (resp.\ $e \strictlyGenerates e'$)
iff there is $ε'\in\finiteResourceSums$ such that 
$e'\in\support{ε'}$ and $e\resRedRT ε'$ (resp.\ $e\resRedT ε'$).
Moreover, $\set{e'\st e\generates e'}$ is always finite and 
$\strictlyGenerates$ defines a well-founded strict partial order.
A direct consequence is that $\resRed$ always converges
to a unique normal form:
\begin{lemma}
	\label{lemma:resRed:SN}
	The reduction $\resRed$ is confluent and strongly normalizing.
	Moreover, for all $ε\in\finiteResourceSums$, the set $\set{ε'\st ε\resRedRT ε'}$
	is finite.
\end{lemma}
\begin{proof}
	Confluence is a consequence of Lemma \ref{lemma:reduction:strongconfluence}.
	By Lemma \ref{lemma:reduction:size}, 
	the transitive closure $\oneStepGenerates^+$
	is a well-founded strict partial order.
	Observe that the elements of $\finiteResourceSums$ 
	can be considered as finite multisets of resource expressions:
	then $\resRedT$ is included in the multiset ordering induced by 
	$\oneStepGenerates^+$, and 
	it follows that $\resRedT$ defines a well-founded strict partial order 
	on  $\finiteResourceSums$, \ie $\resRed$ is strongly
	normalizing.

	The final property follows from strong normalizability applying König's lemma
	to the tree of possible reductions, observing that each $ε$ has finitely 
	many $\resRed$-reducts.
\end{proof}

If $ε\in\finiteResourceSums$, we then write $\NormalForm{ε}$ for the
unique sum of normal resource expressions such that $ε\resRedRT\NormalForm{ε}$.
A consequence of the previous lemma is that any reduction discipline
reaches this normal form:
\begin{corollary}
	\label{corollary:resRed:strategies}
	Let $\mathord\to \subseteq \finiteResourceSums\times \finiteResourceSums$
	be such that $\mathord\to\subseteq\mathord{\resRedRT}$. Moreover assume 
	that, for all non normal $ε\in\finiteResourceSums$
	there is $ε'\not=ε$ such that $ε \to ε'$.
	Then $ε\to^*\NormalForm{ε}$ for all $ε\in\finiteResourceSums$.
\end{corollary}

\section{Vectors of resource expressions and Taylor expansion of algebraic λ-terms}

\label{section:vectors}

\subsection{Resource vectors}

\label{subsection:resourceVectors}

A vector $σ=\sum_{s\in\resourceTerms} σ_s\sm s$ of
resource terms will be called a \emph{term vector} whenever its set of free
variables $\fv{σ}\eqdef\Union_{s\in \support{σ}} \fv s$ is finite. Similarly,
we will call $\emph{monomial vector}$ any vector of resource monomials 
whose set of free variables is finite.
We will abuse notation and write $\termVectors$ for the set of term vectors and
$\monomialVectors$ for the set of monomial vectors.\footnote{
	The restriction to vectors with finitely many free variables is purely technical.
	For instance, it allows us to assume that a sum of abstractions 
	$σ=\sum_{i\in I} \labs {x_i}{s_i}$ can always use a common abstracted
	variable: $σ=\sum_{i\in I} \labs {x}[]{\subst{s_i}{x_i}x}$, with
	$x\not\in\union_{i\in I}\fv {\labs {x_i}{s_i}}$.
	Working without this restriction would 
	only lead to more contorted statements and tedious bookkeeping:
	consider, \eg,  what would happen to the definition of the substitution of a term vector for a
	variable (Definition \ref{definition:resourceVectors:substitution}),
	especially the abstraction case.
}

A \emph{resource vector} will be any of a term vector or a monomial vector, 
and we will write $\resourceVectors$ for either $\termVectors$ or $\monomialVectors$:
as for resource expressions, whenever we use this notation several times in the same context, 
all occurrences consistently denote the same set.

The syntactic constructs are extended to resource vectors by linearity:
for all $σ\in\termVectors$ and $\ms{σ},\ms{τ}\in\monomialVectors$, we set
	\begin{eqnarray*}
		\labs x{σ} &\eqdef&  \sum_{s\in\resourceTerms} {σ}_s\sm \labs xs, \\
		\rappl {σ}{\ms{τ}} &\eqdef& 
			\sum_{s\in\resourceTerms,\ms t\in\resourceMonomials} {σ}_s {\ms{τ}}_{\ms t}\sm \rappl{s}{\ms t}, \\
		\text{ and }\mset {σ_1,\dotsc,σ_n} &\eqdef&
			\sum_{s_1,\dotsc,s_n\in\resourceTerms} (σ_1)_{s_1}\cdots(σ_n)_{s_n}
			\sm \mset{s_1,\dotsc,s_n}.
	\end{eqnarray*}
This poses no problem for finite vectors: \eg, if $\support{σ}$
is finite then finitely many of the vectors $σ_s\sm\labs xs$ are non-zero,
hence the sum is finite. In the general case, however,
we actually need to prove that the above sums are well defined:
the constructors of the calculus define summable functions,
which thus extend to multilinear-continuous maps.\footnote{
	The one-to-one correspondence between summable $n$-ary
	functions and multilinear-continuous maps was established
	for semimodules of the form $\vectors X$, \ie the semimodules
	of all vectors on a fixed set. Due to the restriction we put on free
	variables, $\resourceVectors$ is not of this form: it should rather 
	be written $\Union_{V\in\finiteSubsets{\variables}}\resourceVectorsV V$
	where $\resourceExpressions[V]\eqdef\set{e\in\resourceExpressions\st\fv e\subseteq V}$.
	So when we say a function is multilinear-continuous on $\resourceVectors$,
	we actually mean that its restriction to each $\resourceVectorsV V$ with
	$V\in\finiteSubsets{\variables}$ is multilinear-continuous.
	In the present case, keeping this precision implicit is quite innocuous,
	but we will be more careful when considering the restriction to bounded vectors
	in Subsection \ref{subsection:bounded}, and to normalizable vectors in
	Section \ref{section:normalization}.
}
\begin{lemma}
	The following families of vectors are summable:
	\[ 
		\pars{ \labs xs }_{s\in\resourceTerms}
		,\quad
		\pars{ \rappl{s}{\ms t} }_{s\in\resourceTerms,\ms t\in\resourceMonomials}
		,\quad
		\pars{ \mset{s}}_{s\in\resourceTerms}
		\quad\text{and}\quad
		\pars{ {\ms{s}\cdot{\ms{t}}} }_{\ms s,\ms t\in\resourceMonomials}.
	\]
\end{lemma}
\begin{proof}
	The proof is direct, but we detail it if only to make the requirements explicit.

	For all $u\in\resourceTerms$ there is at most one $s$ such that
	$u\in \support{\labs xs}$ (in which case $u=\labs xs$) and at most one pair $(s,\ms t)$ such that
	$u\in\support{\rappl{s}{\ms t}}$ (in which case $u=\rappl{s}{\ms t}$).

	For all $\ms u\in\resourceMonomials$ there is at most one $s$ such that 
	$\ms u\in\support{\mset s}$ (in which case $\ms u=\mset s$),
	and there are finitely many $\ms s$ and $\ms t$ such that
	$\ms u\in \support{{\ms{s}\cdot{\ms{t}}} }$
	(those such that $\ms u=\ms s\cdot \ms t$).
\end{proof}

For each term vector $σ$, we then write $σ^n$ for the monomial vector
\[\mset[\mnorm]{\overbrace{\strut σ,\dotsc,σ}^{\text{$n$ times}}}.\] 

\subsection{Partial differentiation of resource vectors.}
We can extend partial derivatives to vectors by linear-continuity
(recall that, via the unique semiring morphism from $\naturals$ to $\rigS$, we
can consider that $\finiteResourceSums\subseteq\resourceVectors$).

\begin{lemma}
	The function
	\begin{eqnarray*}
		\resourceExpressions\times\resourceMonomials
		&\to &
		\resourceVectors
		\\
		(e,\mset{s_1,\dotsc,s_n})
		&\mapsto&
		\npdiff n{e}x{\mset{ s_1,\dotsc,s_n}}
	\end{eqnarray*}
	is summable.
\end{lemma}
\begin{proof}
	Let $e'\in\resourceExpressions$ and assume that 
	$e'\in\support{\npdiff nex{\ms s}}$ with $\card{\ms s}=n$.
	By Lemma \ref{lemma:npdiff:size}, $\fv e\subseteq \fv{e'}\union\set{x}$,
	$\fv {\ms s}\subseteq \fv{e'}$, $\size{e}\le\size{e'}$ and $\size{\ms s}\le \size{e'}$:
	$e'$ being fixed, there are finitely many $(e,\ms s)$ satisfying these constraints.
\end{proof}

The characterization of iterated partial derivatives given in Lemma
\ref{lemma:npdiff:decomposition} extends directly to resource vectors, by
the linear-continuity of syntactic constructs and partial derivatives.
For instance, given term vectors $σ,ρ_1,\dotsc,ρ_n\in\termVectors$ and a
monomial vector $\ms{τ}\in\monomialVectors$, we obtain:
\[
	\npdiff n{\rappl{σ}{\ms{τ}}}x{\mset{ρ_1,\dotsc,ρ_n}}=
			\sum_{(I,J)\text{ partition of }\set{1,\dotsc,n}}
			\rappl{\npdiff {\card I}{σ}x{\ms{ρ}_I}}
			{\npdiff {\card J}{\ms{τ}}x{\ms{ρ}_J}}.
\]
Now we can consider iterated differentiation along a fixed term vector $ρ$:
$\npdiff n{ε}x{ρ^n}$. We obtain:
\begin{lemma}
	\label{lemma:npdiff:multilinear}
	For all $σ,τ_1,\dotsc,τ_n,ρ\in\termVectors$ and
	all $\ms{τ}\in\monomialVectors$,
	\begin{eqnarray*}
		\npdiff k{\rappl{σ}{\ms{τ}}}x{ρ^k}&=&
			\sum_{l=0}^k \nchoose k{l,k-l}
			\rappl{\npdiff l{σ}x{ρ^l}}{\npdiff {k-l} {\ms{τ}}x{ρ^{k-l}}}\qquad\text{and}\\
		\npdiff k{\mset{τ_1,\dotsc,τ_n}}x{ρ^k}&=&
			\sum_{\substack{
					k_1,\dotsc,k_n\in\naturals\\
					k_1+\cdots+k_n=k
				}} \nchoose k{k_1,\dotsc,k_n}
			\mset{\npdiff{k_1}{τ_1}x{ρ^{k_1}},\dotsc,\npdiff{k_n}{τ_n}x{ρ^{k_n}}}.
	\end{eqnarray*}
\end{lemma}
\begin{proof}
	First recall that, if $k=\sum_{i=1}^n k_i$,
	the \emph{multinomial coefficient}
	$\nchoose k{k_1,\dotsc,k_n}\eqdef\frac{k!}{\prod_{i=1}^n k_i!}$
	is nothing but the number of partitions of $\set{1,\dotsc,k}$ into $n$ sets
	$I_1,\dotsc,I_n$ such that $\card I_j=k_j$ for $1\le j\le n$
	\cite[§26.4]{NIST:DLMF}.
	Then both results derive directly from Lemma \ref{lemma:npdiff:decomposition}.
\end{proof}

\subsection{Substitutions}
Since $\support{\lsubst ex{\ms s}}\subseteq \support{\npdiff nex{\ms s}}$,
multilinear substitution also defines a sum\-ma\-ble binary function 
and we will write
\[
	\lsubst{ε}x{\ms{σ}}\eqdef
	\sum_{e\in\resourceExpressions,\ms s\in\resourceMonomials}
	ε_e{\ms{σ}}_{\ms{s}}\sm\lsubst{e}x{\ms{s}}.
\]

By contrast with partial derivatives, the usual substitution is not linear, so
the substitution of resource vectors must be defined directly.
\begin{definition}
	\label{definition:resourceVectors:substitution}
	We define by induction over resource expressions
	the \emph{substitution} $\subst ex{σ}\in\resourceVectors$ of
	$σ\in\termVectors$ for a variable $x$ in $e\in\resourceExpressions$:
	\begin{align*}
		\subst{x}x{σ}
			&\eqdef \begin{cases}σ&\text{if $x=y$}\\y&\text{otherwise}\end{cases}\\
		\subst{\pars{\labs ys}}x{σ}
			&\eqdef\labs y{\subst sx{σ}}
			&&\text{(choosing $y\not\in\fv{σ}\union\set x$)}\\
		\subst{\mset{s_1,\dotsc,s_n}}x{σ}
			&\eqdef\mset{\subst{s_1}x{σ},\dotsc,\subst{s_n}x{σ}} \\
		\subst{\pars{\rappl{s}{\ms t}}}x{σ}
			&\eqdef\rappl{\subst sx{σ}}{\subst{\ms t}x{σ}}
	\end{align*}
\end{definition}

\begin{lemma}
	For all $e\in\resourceExpressions$, $x\in\variables$ and $σ\in\termVectors$:
	\begin{itemize}
		\item if $σ\in\resourceTerms$ then $\subst ex{σ}\in\resourceTerms$;
		\item if $σ\in\finiteTermVectors$ then $\subst ex{σ}\in\finiteResourceVectors$;
		\item if $x\not\in\fv e$ then $\subst ex{σ}=e$;
		\item if $x\in\fv e$ then $\subst ex0=0$;
		\item for all $e'\in\support{\subst ex{σ}}$,
			$\fv e\setminus\set{x}\subseteq\fv{e'}\subseteq\pars{\fv e\setminus\set{x}}\union\fv{σ}$
			and $\size {e'}\ge\size{e}$.
	\end{itemize}
\end{lemma}
\begin{proof}
Each statement follows easily by induction on $e$.
\end{proof}
A consequence of the last item is that the function
\begin{eqnarray*}
	\resourceExpressions&\to&\resourceVectors\\
	e&\mapsto& \subst ex{σ}
\end{eqnarray*}
is summable: we thus write 
\[\subst{ε}x{σ}\eqdef \sum_{e\in\resourceVectors}{ε}_e\sm \subst ex{σ}.\]

\subsection{Promotion}

\label{subsection:promotion}

Observe that the family $\pars{σ^n}_{n\in\naturals}$
is summable because the supports $\support{σ^n}$ for $n\in\naturals$ are 
pairwise disjoint.
We then define the \emph{promotion} of $σ$ as
$\prom{σ}\eqdef\sum_{n\in\naturals}\frac 1{n!}\sm σ^n$.

For this definition to make sense, we need inverses of natural numbers to be available:
we say $\rigS$ \emph{has fractions} if every $n\in\naturals\setminus\set 0$
admits a multiplicative inverse in $\rigS$.
This inverse is necessarily unique and we write it $\frac 1n$.
Observe that $\rigS$ has fractions iff there is a semiring morphism
from the semiring $\rationals^+$ of non-negative rational numbers
to $\rigS$, and then this morphism is unique, but not necessarily injective:
consider the semiring $\booleans$ of booleans.
\emph{Semifields}, \ie\ commutative semirings in which every non-zero element
admits an inverse, obviously have fractions:
$\rationals^+$ and $\booleans$ are actually semifields.
In the following, we will keep this requirement implicit:
whenever we use quotients by natural numbers,
it means we assume $\rigS$ has fractions.

\begin{lemma}
	\label{lemma:promotion:substitution}
	For all $σ$ and $τ\in\termVectors$,
	$\subst{\prom{σ}}x{τ}=\prom[]{\subst{σ}x{τ}}$.
\end{lemma}
\begin{proof}
	By the linear-continuity of 
	$ε\mapsto\subst {ε}x{σ}$, it is sufficient 
	to prove that 
	\[\subst{σ^n}x{τ}=\pars{\subst{σ}x{τ}}^n\]
	which follows from the $n$-linear-continuity of 
	$(σ_1,\dotsc,σ_n)\mapsto\mset{σ_1,\dotsc,σ_n}$
	and the definition of substitution.
\end{proof}

\begin{lemma}
	\label{lemma:vlsubst:promotion:constructors}
	The following identities hold:
	\begin{align*}
		\lsubst{x}x{\prom{ρ}} &= ρ\\ 
		\lsubst{y}x{\prom{ρ}} &= y\\ 
		\lsubst{\labs y{σ}}x{\prom{ρ}} &= \labs y[]{\lsubst{σ}x{\prom{ρ}}}&\text{(choosing $y\not\in\set x\union\fv{ρ}$)}\\
		\lsubst{\rappl{σ}{\ms{τ}}}x{\prom{ρ}} &=
		\rappl{\lsubst{σ}x{\prom{ρ}}}{\lsubst{\ms{τ}}x{\prom{ρ}}}\\
		\lsubst{\mset{σ_1,\dotsc,σ_n}}x{\prom{ρ}} &=
		\mset{\lsubst{σ_1}x{\prom{ρ}},\dotsc,\lsubst{σ_n}x{\prom{ρ}}} 
	\end{align*}
\end{lemma}
\begin{proof}
	Since each syntactic constructor is multilinear-continuous,
	it is sufficient to consider the case of $\lsubst ex{\prom{ρ}}$
	for a resource expression $e\in\resourceExpressions$.
	First observe that, if $k=\occnum xe$ then 
	$\lsubst ex{\prom{ρ}}=
	\frac{1}{k!}\sm{\npdiff kex{ρ^k}}$.
	In particular the case of variables is straightforward.

	The case of abstractions follows directly, since
	$\npdiff k {\labs xs}x{ρ^k}=\labs x[]{\npdiff ksx{ρ^k}}$.

	If $e=\rappl s{\ms t}$, write $l=\occnum xs$ and $m=\occnum x{\ms t}$.
	It follows from Lemma \ref{lemma:npdiff:multilinear} that 
	$\lsubst ex{ρ^k}=\nchoose k{l,m}\sm
	\rappl{\lsubst sx{ρ^l}}{\lsubst{\ms t}x{ρ^m}}$
	and then 
	$\frac 1{k!}\sm\lsubst ex{ρ^k}=
	\rappl{\frac 1{l!}\sm\lsubst sx{ρ^l}}{\frac 1{m!}\sm\lsubst{\ms t}x{ρ^m}}$.
	
	Similarly, if $e=\mset{t_1,\dotsc,t_n}$,
	write $k_i=\occnum x{t_i}$ for all $i\in\set{1,\dotsc,n}$.
	It follows from Lemma \ref{lemma:npdiff:multilinear} that 
	$\lsubst ex{ρ^k}=\nchoose k{k_1,\dotsc,k_n}\sm
	\mset{\lsubst {t_1}x{ρ^{k_1}},\dotsc,\lsubst {t_n}x{ρ^{k_n}}}$
	and then 
	$\frac 1{k!}\sm\lsubst ex{ρ^k}= \mset{
		\frac 1{k_1!}\lsubst {t_1}x{ρ^{k_1}},
	  \dotsc,
	  \frac 1{k_n!}\lsubst {t_n}x{ρ^{k_n}}
	}$.
\end{proof}

\begin{lemma}
	\label{lemma:vlsubst:promotion:substitution}
	For all $ε\in\resourceVectors$ an $σ\in\termVectors$,
	\[\subst{ε}x{σ}=\lsubst{ε}x{\prom{σ}}.\]
\end{lemma}
\begin{proof}
	By the linear-continuity of 
	$ε\mapsto\lsubst {ε}x{\prom{σ}}$ and
	$ε\mapsto\subst {ε}x{σ}$, 
	it is sufficient to show that
	\[\subst{e}x{σ}=\lsubst{e}x{\prom{σ}}\]
	for all resource expression $e$.
	The proof is then by induction on $e$, 
	using the previous  Lemma in each case.
\end{proof}

By Lemma \ref{lemma:promotion:substitution},
we thus obtain 
\[\lsubst{\prom{σ}}x{\prom{τ}}=\prom[]{\lsubst{σ}x{\prom{τ}}}\]
which can be seen as a counterpart of the functoriality 
of promotion in linear logic.
To our knowledge it is the first published proof of such a result 
for resource vectors. This will enable us to prove the 
commutation of Taylor expansion and substitution
(Lemma \ref{lemma:taylor:subst}), another unsurprising
yet non-trivial result.

\subsection{Taylor expansion of algebraic λ-terms}

\label{subsection:taylor}

Since resource vectors form a module, there is no reason to restrict the source
language of Taylor expansion to the pure λ-calculus: we can consider formal
finite linear combinations of λ-terms.

We will thus consider the terms given by the following grammar:
\[\begin{array}{rclcl}
	\rawTerms&\ni& M,N,P & \recdef & x \mid \labs xM \mid \appl MN \mid 0 \mid a\sm M\mid M+N
\end{array}\]
where $a$ ranges in $\rigS$.\footnote{We follow Krivine’s convention \cite{krivine:lc}, 
by writing $\appl MN$ for the application of term $M$ to term $N$.
We more generally write $\appl{M}{N_1\cdots N_k}$ for $\appl{\cdots\appl M{N_1}\cdots}{N_k}$.
Moreover, among term constructors, we give sums the lowest priority so that 
$\appl MN+P$ should be read as $\pars{\appl MN}+P$ rather than $\appl M{\pars{N+P}}$.
}
For now, terms are considered up to the usual α-equivalence only: 
the null term $0$, scalar multiplication $a\sm M$ and sum of terms
$M+N$ are purely syntactic constructs.

\begin{definition}
	We define the \emph{Taylor expansion} $\TaylorExp M\in\resourceVectors$ 
	of a term $M\in\rawTerms$ inductively as follows:
	\begin{align*}
		\TaylorExp{x}        &\eqdef x
		& \TaylorExp{0}        &\eqdef 0\\
		\TaylorExp{\labs xM} &\eqdef \labs x{\TaylorExp M}
		& \TaylorExp{a\sm M}   &\eqdef a\sm \TaylorExp{M}\\
		\TaylorExp{\appl MN} &\eqdef \rappl{\TaylorExp M}{\prom{\TaylorExp N}}
		& \TaylorExp{M+N}      &\eqdef \TaylorExp{M}+\TaylorExp{N}.
	\end{align*}
\end{definition}

\begin{lemma}
	\label{lemma:taylor:subst}
	For all $M,N\in\rawTerms$, and all variable $x$,
	\[\TaylorExp{\subst MxN}=\lsubst{\TaylorExp M}x{\prom{\TaylorExp N}}=\subst{\TaylorExp M}x{\TaylorExp N}.\]
\end{lemma}
\begin{proof}
	By induction on $M$, 
	using Lemmas \ref{lemma:vlsubst:promotion:constructors}
	and \ref{lemma:vlsubst:promotion:substitution}.
\end{proof}

Let us insist on the fact that, despite its very simple and unsurprising
statement, the previous lemma relies on the entire technical development
of the previous subsections. Again, to our knowledge, it
is the first proof that Taylor expansion commutes with substitution, in an
untyped and non-uniform setting, without any additional assumption.

By contrast, one can forget everything about the semiring of coefficients
and consider only the support of Taylor expansion.
Recall that $\booleans$ denotes the semiring of booleans.
Then we can consider that
$\resourceVectors[\booleans]=\powerset{\resourceExpressions}$ 
and write, \eg, $\labs x{\calS}=\set{\labs xs\st s\in\calS}$ for all set $\calS$ of
resource terms.

\begin{definition}
The \emph{Taylor support} $\TaylorSup M\subseteq\resourceTerms$
	of $M\in\rawTerms$ is defined inductively as follows:\footnote{
		One might be tempted to make an exception in
		case $a=0$ and set $\TaylorSup{0\sm M}=\emptyset$
		but this would only complicate the definition and further developments
		for little benefit: what about $\TaylorSup{a\sm M+b\sm M}$
		(resp.\ $\TaylorSup{a\sm b\sm M}$) in a semiring where
		$a\not=0$, $b\not=0$ and $a+b=0$ (resp.\ $ab=0$)?
		If we try and cope with those too, 
		we are led to make $\TaylorSupSym$ invariant under the equations 
		of $\rigS$-module, which is precisely what we want to avoid here:
		see the case of $\TaylorExpSym$ in the remaining of the present section.
	}
	\begin{align*}
		\TaylorSup{x}        &\eqdef \set{x}&
		\TaylorSup{0}        &\eqdef \emptyset\\
		\TaylorSup{\labs xM} &\eqdef \labs x{\TaylorSup M}&
		\TaylorSup{a\sm M}   &\eqdef \TaylorSup M\\
		\TaylorSup{\appl MN} &\eqdef \rappl{\TaylorSup M}{\prom{\TaylorSup N}}&
		\TaylorSup{M+N}      &\eqdef \TaylorSup M\union\TaylorSup N.
	\end{align*}
\end{definition}

It should be clear that $\support{\TaylorExp{M}}\subseteq\TaylorSup M$,
but the inclusion might be strict,
if only because $\TaylorSup{0\sm M}=\TaylorSup M$.
By contrast with the technicality of the previous subsection,
the following \emph{qualitative} analogue of Lemma \ref{lemma:taylor:subst}
is easily established:
\begin{lemma}
	\label{lemma:taylorsup:subst}
	For all $M,N\in\rawTerms$, and all variable $x$,
	\[\TaylorSup{\subst MxN}=\lsubst{\TaylorSup M}x{\prom{\TaylorSup N}}=\subst{\TaylorSup M}x{\TaylorSup N}.\]
\end{lemma}
\begin{proof}
	The qualitative version of Lemma
	\ref{lemma:vlsubst:promotion:constructors}
	is straightforward.
	The result follows by induction on $M$.
\end{proof}

The restriction of $\TaylorSupSym$ to the set $\lambdaTerms$ of pure λ-terms
was used by Ehrhard and Regnier \cite{er:resource}
in their study of Taylor expansion. They showed 
that if $M\in\lambdaTerms$ then $\TaylorSup M$ is uniform:
all the resource terms in $\TaylorSup M$ have 
the same outermost syntactic construct and this property is preserved 
inductively on subterms. They moreover proved that $\TaylorExp M$, and in fact $M$ itself, is 
entirely characterized by $\TaylorSup M$: in this case,
$\TaylorExp M=\sum_{s\in\TaylorSup M} \frac 1{m(s)} s$
where $m(s)$ is an integer coefficient depending only on $s$.
Of course this property fails in the non uniform setting of $\rawTerms$.

Now, let us consider the equivalence induced on terms by Taylor expansion:
write $M\teq N$ if $\TaylorExp{M}=\TaylorExp N$.
\begin{lemma}
The following equations hold:
\begin{align*}
	0+M&\teq M&
	M+N&\teq N+M&
	(M+N)+P&\teq M+(N+P)\\
	0\sm M&\teq 0&
	1\sm M&\teq M&
	a\sm M+b\sm M&\teq (a+b)\sm M\\
	a\sm 0 &\teq 0&
	a\sm (b\sm M)&\teq (ab)\sm M&
	a\sm (M+N) &\teq a\sm M+a\sm N\\
	\labs x0&\teq 0&
	\labs x[]{a\sm M}&\teq a\sm\labs xM&
	\labs x[]{M+N}&\teq \labs xM+\labs xN\\
	\appl 0P&\teq 0&
	\appl {a\sm M}P&\teq a\sm\appl MP&
	\appl {M+N}P&\teq \appl MP+\appl NP
\end{align*}
Moreover, $\teq$ is compatible with syntactic constructs:
if $M\teq M'$ then $\labs x M\teq \labs xM'$, $\appl MN\teq \appl{M'}N$, $\appl
NM\teq \appl N{M'}$, $a\sm M\teq a\sm M'$, $M+N\teq M'+N$ and $N+M\teq N+M'$.
\end{lemma}
\begin{proof}
	Up to Taylor expansion, these equations reflect the 
	fact that $\resourceVectors$ forms a semimodule (first three lines), and 
	that all the constructions used in the definition
	of $\TaylorExpSym$ are multilinear-continuous, except for promotion 
	(last two lines). Compatibility
	follows from the inductive definition of $\TaylorExpSym$.
\end{proof}

Let us write $\algEq$ for the least compatible equivalence relation containing
the equations of the previous lemma, and call \emph{vector λ-terms} the elements 
of the quotient $\vectorTerms$: these are the terms of the previously studied
algebraic λ-calculus \cite{vaux:alglam,alberti:phd}.\footnote{
	In those previous works, the elements of $\vectorTerms$ were called
	algebraic λ-terms, but here we reserve this name for another, simpler,
	notion.}

It is clear that $\vectorTerms$ forms a $\rigS$-semimodule.
In fact, one can show \cite{vaux:alglam} that $\vectorTerms$
is freely generated by the $\algEq$-equivalence classes of base terms, \ie
those described by the following grammar:
\[\begin{array}{rclcl}
	\baseTerms   &\ni& B & \recdef & x \mid \labs xB \mid \appl BM.
\end{array}\]
Hence we could write $\vectorTerms=\finiteVectors{\quotient{\baseTerms}{\algEq}}$.

Notice however that Taylor expansion is not injective on vector λ-terms in general.
\begin{example}
	\label{example:TaylorExp:not:injective}
	We can consider that
	$\vectorTerms[\booleans]=\finiteSubsets{\quotient{\baseTerms[\booleans]}{\algEq}}$
	and $\TaylorExp M\subseteq\resourceTerms$ for all $M\in\rawTerms[\booleans]$.
	It is then easy to check that,
	\eg, $\TaylorExp{\appl x\emptyset}\subseteq\TaylorExp{\appl xx}$,
	hence $\appl x\emptyset +_\booleans \appl xx\teq \appl xx$.\footnote{
		This discrepancy is also present in the non-deterministic Böhm trees of
		de'Liguoro and Piperno \cite{deliguoro-piperno:ndlc}: in that qualitative
		setting, they can solve it by introducing a preorder on trees based on set
		inclusion. They moreover show that this preorder coincides with that
		induced by a well chosen domain theoretic model, as well as with the
		observational preorder associated with must-solvability. This preorder
		should be related with that induced by 
		the inclusion of normal forms of Taylor expansions (which are always defined 
		since we then work with support sets rather than general vectors).
	}
\end{example}
This contrasts with the case of pure λ-terms, for which $\TaylorExpSym$ is always injective:
in this case, it is in fact sufficient to look at the linear resource terms
in supports of Taylor expansions.
\begin{fact}
	\label{fact:TaylorExp:lambda}
	For all $M,N\in\lambdaTerms$, 
	$\lapprox M\in\support{\TaylorExp N}$ iff $M=N$,
	where $\lapproxSym$ is defined inductively 
	as follows:
	\begin{align*}
		\lapprox x&\eqdef x
		&\lapprox{\labs xM}&\eqdef\labs x{\lapprox M}
		&\lapprox{\appl MN}&\eqdef\rappl{\lapprox M}{\mset{\lapprox N}}.
	\end{align*}
\end{fact}

To our knowledge, finding sufficient conditions on $\rigS$ ensuring that
$\TaylorExpSym$ becomes injective on $\vectorTerms$ is still an open question.

Observe moreover that the $\rigS$-semimodule structure of $\vectorTerms$ gets in the way when 
we want to study β-reduction and normalization: it is well known
\cite{vaux:alglam2,ad:lineal,vaux:alglam} that β-reduction in a semimodule of terms is
inconsistent in presence of negative coefficients.

\begin{example}
	\label{example:inconsistency}
	Consider $\delta_M\eqdef\labs x[]{M+\appl xx}$ and $\infty_M\eqdef\appl{\delta_M}{\delta_M}$.
	Observe that $\infty_M$ β-reduces to $M+\infty_M$. Suppose $\rigS$ is a ring.
	Then any congruence $\simeq$ on $\rawTerms[\rigS]$ containing β-reduction
	and the equations of $\rigS$-module is inconsistent:
	$0\simeq\infty_M+(-1)\sm\infty_M\simeq (M+\infty_M)+(-1)\sm\infty_M\simeq M$.
\end{example}

The problem is of course the identity $0\simeq\infty_M+(-1)\sm\infty_M$. Another difficulty is that,
if $\rigS$ has fractions then, up to $\rigS$-semimodule equations, one can
split a single β-reduction step into infinitely many fractional steps: if
$M\betaRed M'$ then 
\[
	\textstyle
	M
	\simeq\frac 12\sm M+\frac 12\sm M
	\betaRed\frac 12\sm M+\frac 12\sm M'
	\simeq\pars{\frac 14\sm M+\frac 14\sm M}+\frac 12\sm M'
	\betaRed\pars{\frac 14\sm M+\frac 14\sm M'}+\frac 12\sm M'
	\simeq\cdots
\]

It is not our purpose here to explore the various possible fixes to the
rewriting theory of β-reduction on vector λ-terms.
We rather refer the reader to the literature on algebraic λ-calculi
\cite{vaux:alglam,ad:lineal,alberti:phd,diazcaro:phd} for various proposals.
Our focus being on Taylor expansion, we propose to consider vector λ-terms as
intermediate objects: the reduction relation induced on resource vectors by β-reduction 
through Taylor expansion contains β-reduction on vector terms --- which is mainly 
useful to understand what may go wrong.

We still need to introduce some form of quotient in the syntax, though, if only
to allow formal sums to retain a computational meaning: otherwise, for
instance, no β-redex can be fired in $\appl{\labs x M+\labs xN}P$; and more 
generally there are β-normal terms whose Taylor expansion is not normal,
and conversely (consider, \eg, $\appl{\labs x 0}P$).

Write $\algebraicTerms$ for the quotient of $\rawTerms$ by the least compatible
equivalence $\canEq$ containing the following six equations:
\begin{align*}
	\labs x0         &\canEq 0&
	\labs x[]{a\sm M}&\canEq a\sm\labs xM&
	\labs x[]{M+N}   &\canEq \labs xM+\labs xN
	\\
	\appl 0P       &\canEq 0&
	\appl {a\sm M}P&\canEq a\sm\appl MP&
	\appl {M+N}P   &\canEq \appl MP+\appl NP
\end{align*}
We call \emph{algebraic λ-terms} the elements of $\algebraicTerms$.
We will abuse notation and denote an algebraic λ-term by 
any of its representatives.

Observe that $\TaylorSup M$ is preserved under $\canEq$ so 
it is well defined on algebraic terms, although not on vector terms.
\begin{fact}
	An algebraic λ-term $M$ is β-normal (\ie each of its representatives is
	β-normal) iff $\TaylorSup M$ contains only normal resource terms.
\end{fact}

We do not claim that $\canEq$ is minimal with the above property
(for this, the bottom three equations are sufficient) but it 
is quite natural for anyone familiar with the decomposition of λ-calculus 
in linear logic, as it reflects the linearity of λ-abstraction 
and the function position in an application. Moreover it retains the two-level
structure of vector λ-terms, seen as sums of base terms.

It is indeed a routine exercise to show that orienting the defining equations of $\canEq$
from left to right defines a confluent and terminating rewriting system. We
call canonical terms the normal forms of this system, which we can describe as
follows. The sets $\canonicalTerms$ of canonical terms and $\simpleRawTerms$ of simple
canonical terms are mutually generated by the following grammars:
\[\begin{array}{rclcl}
	\simpleRawTerms &\ni& S,T   & \recdef & x \mid \labs xS \mid \appl SM\\
	\canonicalTerms &\ni& M,N,P & \recdef & S \mid 0 \mid a\sm M \mid M+N
\end{array}\]
so that each algebraic term $M$ admits a unique canonical $\canEq$-representative.

In the remaining of this paper we will systematically identify algebraic
terms with their canonical representatives and keep $\canEq$ implicit.
Moreover, we write $\simpleTerms$ for the set of simple algebraic λ-terms, \ie
those that admit a simple canonical representative.

\begin{fact}
	\label{fact:simple:head}
	Every simple term $S\in\simpleTerms$ is of one of the following two forms:
	\begin{itemize}
		\item $S=\labs {x_1}{\cdots\labs {x_n}{\appl{x}{M_1\cdots M_k}}}$: $S$ is a head normal form;
		\item $S=\labs {x_1}{\cdots\labs {x_n}{\appl{\labs x T}{M_0\cdots M_k}}}$: 
			$\appl{\labs x T}{M_0}$ is the head redex of $S$.
	\end{itemize}
\end{fact}
So each algebraic λ-term can be considered as a formal linear combination 
of head normal forms and head reducible simple terms, which will structure
the notions of weak solvability and hereditarily determinable terms in section
\ref{section:determinable}.

\section{On the reduction of resource vectors}

\label{section:reduction:vectors}

Observe that 
\[
	\TaylorExp{\appl {\labs xM}N}
	=\rappl{\labs x{\TaylorExp{M}}}{\prom{\TaylorExp N}}
	=\sum_{\substack{s\in\resourceTerms\\\ms t\in\resourceMonomials}}
	\TaylorExp{M}_s\prom{\TaylorExp N}_{\ms t}\sm \rappl{\labs xs}{\ms t}
\]
and
\[
	\TaylorExp{\subst MxN}
	=\lsubst{\TaylorExp{M}}x{\prom{\TaylorExp N}}
	=\sum_{\substack{s\in\resourceTerms\\\ms t\in\resourceMonomials}}
	\TaylorExp{M}_s\prom{\TaylorExp N}_{\ms t}\sm \lsubst {s}x{\ms t}
\]

In order to simulate β-reduction through Taylor expansion we might be
tempted to consider the reduction given by
$ε\to ε'$ as soon as $ε=\sum_{i\in I} a_i\sm e_i$
and $ε'=\sum_{i\in I} a_i\sm ε'_i$ with $e_i\resRed ε'_i$ for all $i\in I$.\footnote{
	We must of course require that $\Union_{i\in I}\fv{e_i}$ is finite
	but, again, we will keep such requirements implicit in the following.
}

Observe indeed that, as soon as $\pars{a_i\sm e_i}_{i\in I}\in\resourceExpressions$
is summable (\ie for all $e\in\resourceExpressions$, there are finitely many
$i\in I$ such that $a_i\not=0$ and $e_i=e$),
the family $\pars{a_i\sm ε'_i}_{i\in I}$ is summable too:
if $e'\in \support{a_i\sm ε'_i}$ then $a_i\not=0$ and  $e'\in \support{ε'_i}$
hence by Lemma \ref{lemma:reduction:size},
$\fv {e_i}=\fv{e'}$ and $\size{e_i}\le 2\size{e'}+2$;
$e'$ being fixed, there are thus finitely many possible values for $e_i$ hence for $i$.
So we do not need any additional condition for this reduction step to be well defined.

This reduction, however, is not suitable for simulating β-reduction
because whenever the reduced β-redex is 
not in linear position, we need to reduce arbitrarily many 
resource redexes.

\begin{example}
Observe that
\[\TaylorExp{\appl{y}{\appl{\labs xx}z}}=\sum_{n,k_1,\dotsc,k_n\in\naturals}\frac 1{n!k_1!\cdots k_n!}
\rappl{y}{\mset{\rappl{\labs x{x}}{z^{k_1}},\dotsc,\rappl{\labs x{x}}{z^{k_n}}}}\]
and 
\[\TaylorExp{\appl{y}{z}}=\sum_{n\in\naturals}\frac 1{n!} \rappl{y}{z^n}.\]
Then the reduction from $\mset{\rappl{\labs x{x}}{z^{k_1}},\dotsc,\rappl{\labs x{x}}{z^{k_n}}}$
to $z^n$ if each $k_i=1$ (resp.\ to $0$ if one $k_i\not=1$) requires firing $n$ 
independent redexes (resp.\ one of those $n$ redexes).
\end{example}

\subsection{Parallel resource reduction}

\label{subsection:parallel}

One possible fix would be to replace $\resRed$ with $\resRedRT$ in the above definition,
\ie set $ε\to ε'$ as soon as $ε=\sum_{i\in I} a_i\sm e_i$
and $ε'=\sum_{i\in I} a_i\sm ε'_i$ with $e_i\resRedRT ε'_i$ for all $i\in I$,
but then the study of the reduction subsumes that of 
normalization, which we treat in Section \ref{section:normalization}, and this relies 
on the possibility to simulate β-reduction steps.

A reasonable middle ground is to consider a parallel variant $\presRed$ of
$\resRed$, where any number of redexes can be reduced simultaneously in one
step. The parallelism involved in the translation of a β-reduction step is
actually quite constrained: like in the previous example, the redexes that need
to be reduced in the Taylor expansion are always pairwise independent and no
nesting is involved. However, in order to prove the confluence of the reduction
on resource vectors, or its conservativity w.r.t. β-reduction, it is much more
convenient to work with a fully parallel reduction relation, both on algebraic
λ-terms and on resource vectors. Indeed, parallel reduction relations generally
allow, e.g., to close confluence diagrams in one step or to define a maximal
parallel reduction step: the relevance of this technical choice will be made
clear all through Section \ref{section:parallel}.

\begin{definition}
	\label{definition:presRed}
	We define \emph{parallel resource reduction} $\mathord{\presRed}\subseteq\resourceExpressions\times\finiteResourceSums$
	inductively as follows:
	\begin{itemize}
		\item $x\presRed x$;
		\item $\rappl{\labs x s}{\ms t}\presRed \lsubst {σ'}x{\ms{τ'}}$
			as soon as $s\presRed σ'$ and $\ms t\presRed \ms{τ'}$;
		\item $\labs xs\presRed \labs x {σ'}$ as soon as $s\presRed σ'$;
		\item $\rappl s{\ms t}\presRed \rappl{σ'}{\ms{τ'}}$ as soon as $s\presRed σ'$ and $\ms t\presRed \ms{τ}'$;
		\item $\mset{s_1,\dotsc,s_n}\presRed\mset{σ'_1,\dotsc,σ'_n}$ as soon as $s_i\presRed σ'_i$ for each $i\in\set{1,\dotsc,n}$.
	\end{itemize}
	We extend this reduction to sums of resource expressions by linearity:
	$ε\presRed ε'$ if $ε=\sum_{i=1}^n e_i$ and $ε'=\sum_{i=1}^n ε'_i$ 
	with $e_i\presRed ε'_i$ for all $i\in\set{1,\dotsc,n}$.
\end{definition}
It should be clear that $\mathord{\resRed}\subseteq\mathord{\presRed}\subset\mathord{\resRedRT}$,
Moreover observe that, because all term constructors are linear, the reduction
rules extend naturally to finite sums of resource expressions:
for instance, $\labs x{σ}\presRed \labs x{σ'}$ as 
soon as $σ\presRed σ'$.

We will prove in Sections \ref{section:parallel} and \ref{section:simulation:beta}
that this solution is indeed a good one:
parallel resource reduction is strongly confluent, and there is a way 
to extend it to resource vectors so that not only the resulting 
reduction is strongly confluent and allows to simulate β-reduction, but
any reduction step from the Taylor expansion of an algebraic term 
can be completed into a parallel β-reduction step.
There are two pitfalls with this approach, though.

\subsection{Size collapse}

\label{subsection:collapse}

First, parallel reduction $\presRed$ (like iterated reduction $\resRedRT$) lacks the
combinatorial regularity properties of $\resRed$ given by Lemma \ref{lemma:reduction:size}: 
write $e\pOneStepGenerates e'$ if $e\presRed ε'$ with $e'\in\support{ε'}$;
$e'\in\resourceExpressions$ being fixed, there is no bound on the size of the
$\presRed$-antecedents of $e'$, \ie those $e\in\resourceExpressions$ such that
$e\pOneStepGenerates e'$.
\begin{example}
	\label{example:reduction:collapse}
	Fix $s\in\resourceTerms$.
	Consider the sequences $\vec u(s)$ and $\vec v(s)$ of resource terms given by:
	\[
		\left\{
		\begin{aligned}
			u_0(s)&\eqdef s\\
			u_{n+1}(s)&\eqdef \rappl{\labs yy}{\mset{u_n(s)}}
		\end{aligned}
		\right.
		\quad
		\text{and}
		\quad
		\left\{
		\begin{aligned}
			v_0(s)&\eqdef s\\
			v_{n+1}(s)&\eqdef \rappl{\labs y{v_n(s)}}{\mset{}}
		\end{aligned}
		\right.
		.
	\]
	Observe that for all $n\in\naturals$, $u_{n+1}(s)\resRed u_n(s)$ and  $v_{n+1}(s)\resRed v_n(s)$,
	and more generally, for all $n'\le n$, $u_{n}(s)\presRed u_{n'}(s)$ and $v_{n}(s)\presRed v_{n'}(s)$.
	In particular $u_n(s)\presRed s$ and $v_n(s)\presRed s$ for all $n\in\naturals$.
\end{example}
Reducing all resource expressions in a resource vector simultaneously is thus no longer
possible in general: consider, \eg, $\sum_{n\in\naturals} u_n(x)$.
As a consequence, when we introduce a reduction relation on resource vectors 
by extending a reduction relation on resource expressions
as above, we must in general impose the summability of the family of reducts 
as a side condition:
\begin{definition}
	Fix an arbitrary relation
	$\mathord{\reduce}\subseteq\resourceExpressions\times\finiteResourceSums$.
	For all $ε,ε'\in\resourceVectors$,
	we write $ε\splitVariant\reduce ε'$
	whenever there exist families $(a_i)_{i\in I}\in \rigS^I$, 
	$\pars{e_i}_{i\in I}\in\resourceExpressions^I$ 
	and $\pars{ε'_i}_{i\in I}\in \finiteResourceSums^I$ 
	such that:
	\begin{itemize}
		\item $\pars{e_i}_{i\in I}$ is summable and $ε=\sum_{i\in I}a_i\sm e_i$;
		\item $\pars{ε'_i}_{i\in I}$ is summable and $ε'=\sum_{i\in I}a_i\sm ε'_i$;
		\item for all $i\in I$, $e_i\R\reduce ε'_i$.
	\end{itemize}
\end{definition}
The necessity of such a side condition forbids confluence. Indeed:
\begin{example}
	\label{example:presRed:not:confluent}
	Let $σ=\sum_n u_n(v_n(x))$. Then $σ\vpresRed\sum u_n(x)$ and $σ\vpresRed\sum v_n(x)$,
	but since the only common reduct of $u_p(x)$ and $v_q(x)$ is $x$, there is no
	way\footnote{
		In fact, this argument is only valid if $\rigS$ is zerosumfree
		(\ie\ if $a+b=0\in\rigS$ entails $a=b=0$; see below, in particular Lemma \ref{lemma:zerosumfree:splitVariant}),
		for instance if $\rigS=\naturals$: 
	we rely on the fact that if $\sum_{i\in I} a_i\sm s_i=\sum_{n\in\naturals} u_n(x)$ then 
	for all $i\in I$ such that $a_i\not=0$, there is $n\in\naturals$ such that $s_i=u_{n}(x)$.
}
	to close this pair of reductions: $\pars{x}_{n\in\naturals}$ is not summable.
\end{example}
These considerations lead us to study the combinatorics of parallel resource reduction 
more closely: in Section \ref{section:parallel}, we introduce successive variants
of parallel reduction, based on restrictions on the nesting of fired redexes,
and provide bounds for the size of antecedents of a resource expression.
We moreover consider sufficient conditions for these restrictions
to be preserved under reduction.

We then observe in Section \ref{section:simulation:beta} that, when applied to
Taylor expansions, parallel reduction is automatically of the most restricted
form, which allows us to provide uniform bounds and obtain the desired
confluence and simulation properties.

\subsection{Reduction structures}

\label{subsection:structures}

The other, \emph{a priori} unrelated pitfall is the fact that the reduction can interact
badly with the semimodule structure of $\resourceVectors$: we can reproduce
Example \ref{example:inconsistency} in $\resourceVectors$ through Taylor
expansion (see the discussion in Section \ref{vpresRed:degenerate},
p.\pageref{vpresRed:degenerate}).
Even more simply, we can use the terms of Example \ref{example:reduction:collapse}:
\begin{example}
	Let $s\in\resourceTerms$ and $σ=\sum_{n\in\naturals} u_{n+1}(s)\in\termVectors$.
	Assuming $\rigS$ is a ring:
	$
		0=σ+(-1).σ
		\vpresRed
		\sum_{n\in\naturals} u_{n}(s)+(-1).σ=s
	$.
	
\end{example}
Of course, this kind of issue does not arise when the semiring of coefficients is \emph{zerosumfree}:
recall that $\rigS$ is zerosumfree if $a+b=0$ implies $a=b=0$, which holds for
all semirings of non-negative numbers, as well as for booleans.
This prevents interferences between reductions and the semimodule structure:
\begin{lemma}
	\label{lemma:zerosumfree:splitVariant}
	Assume $\rigS$ is zerosumfree and fix a relation
	$\mathord{\reduce}\subseteq\resourceExpressions\times\finiteResourceSums$.
	If $ε\splitVariant\reduce ε'$ then, for all $e'\in\support{ε'}$ there exists
	$e\in\support{ε}$ and $ε_0\in \finiteResourceSums$
	such that $e\R\reduce ε_0$ and $e'\in\support{ε_0}$.
\end{lemma}
\begin{proof}
	Assume $ε=\sum_{i\in I} a_i\sm e_i$
	and  $ε'=\sum_{i\in I} a_i\sm ε'_i$ with $e_i\R\reduce ε'_i$ for all $i\in I$.
	If $e'\in\support{ε'}$ then there is $i\in I$ such that 
	$e'\in\support{a_i\sm ε'_i}$ hence $a_i\not=0$ and $e'\in\support{ε'_i}$.
	Then, since $\rigS$ is zerosumfree, $e_i\in\support{ε}$.
\end{proof}

Various alternative approaches to get rid of this
restriction in the setting of the algebraic λ-calculus
can be adapted to the reduction of resource vectors:
we refer the reader to the literature on algebraic λ-calculi
\cite{vaux:alglam,ad:lineal,alberti:phd,diazcaro:phd} for several proposals.
The linear-continuity of the resource λ-calculus allows us to propose a novel
approach: consider possible restrictions on 
the families of resource expressions simultaneously reduced in a
$\splitVariant\reduce$-step.

\begin{definition}
	\label{definition:resource:structure}
	We call \emph{resource support} any set $\calE\subseteq\resourceExpressions$ of
	resource expressions such that $\fv{\calE}=\Union_{e\in\calE}\fv e$ is finite.
	Then a \emph{resource structure} is any set
	$\fE\subseteq\powerset{\resourceExpressions}$ of resource supports such that:
	\begin{itemize}
		\item $\fE$ contains all finite resource supports;
		\item $\fE$ is closed under finite unions;
		\item $\fE$ is downwards closed for inclusion.
	\end{itemize}
\end{definition}
The maximal resource structure is $\resourceStructure\eqdef \set
{\calE\subseteq\resourceExpressions\st\fv{\calE}\text{ is finite}}$, which is also
a finiteness structure \cite{ehrhard:finres}. Observe that any 
finiteness structure $\fF\subseteq\resourceStructure$ is a resource structure:
all three additional conditions are automatically satisfied.

\begin{definition}
	\label{definition:splitVariant}
	Fix a relation 
	$\mathord{\reduce}\subseteq\resourceExpressions\times\finiteResourceSums$.
	For all resource support $\calE$, we write $\splitVariant[\calE]\reduce$ for
	$\splitVariant\reduce[\rightharpoonup_{\calE}]$ where $\rightharpoonup_{\calE}$
	denotes ${\mathord\to\inter(\calE\times\finiteResourceSums)}$.
	For all resource structure $\fE$, we then write $\splitVariant[\fE]\reduce$
	for $\Union_{\calE\in\fE} \mathord{\splitVariant[\calE]\reduce}$.
\end{definition}

We have 
$\mathord{\splitVariant[\calE]\reduce}\subseteq\mathord{\splitVariant\reduce}\inter(\vectors{\calE}\times\resourceVectors)$,
but in general the reverse inclusion holds only if $\rigS$ is zerosumfree:
in this latter case $ε\vpresRed ε'$ iff $ε\vpresRed[\support{ε}]ε'$.
	
\begin{definition}
	We call \emph{$\reduce$-reduction structure} any resource structure $\fE$ such that
	if $\calE\in\fE$ then $\Union\set{\support{ε'}\st e\in\calE\text{ and }e\R\reduce ε'}\in\fE$.
\end{definition}

We will consider some particular choices of reduction structure in the following,
but the point is that our approach is completely generic.
The results of Section \ref{section:simulation:beta} will imply
that if $\fS\subseteq\powerset{\termVectors}$ is a $\presRed$-reduction structure 
containing $\support{\TaylorExp M}$ then one can translate
any $\pbetaRed$-reduction sequence from $M$ 
into a $\vpresRed[\fS]$-reduction sequence from $\TaylorExp M$.
Additional properties such as the confluence of $\vpresRed[\fS]$,
its conservativity over $\pbetaRed$, or its compatibility with 
normalization will depend on additional conditions on $\fS$.

\section{Taming the size collapse of parallel resource reduction}

\label{section:parallel}

In this section, we study successive families of restrictions of the parallel
resource reduction $\presRed$. Our purpose is to enforce some control on the
size collapse induced by $\presRed$, so as to obtain a confluent restriction of
$\vpresRed$, all the while retaining enough parallelism to simulate parallel
β-reduction on algebraic λ-terms, ideally in a conservative way.

First observe that parallel resource reduction itself is strongly confluent as expected:
following a classic argument, we define $\fullReduct e$ as the 
result of firing all redexes in $e$ and then, whenever $e\presRed ε'$,
we have $ε'\presRed \fullReduct e$. Formally:
\begin{definition}
	For all $e\in\resourceExpressions$ we define 
	the \emph{full parallel reduct} $\fullReduct e$ of $e$ by induction on $e$ as follows:
	\begin{align*}
		\fullReduct{x}                      &\eqdef x\\
		\fullReduct{\labs xs}               &\eqdef \labs x{\fullReduct{s}}\\
		\fullReduct{\rappl{\labs xs}{\ms t}}&\eqdef \lsubst{\fullReduct{s}}x{\fullReduct{\ms t}}\\
		\fullReduct{\rappl{s}{\ms t}}       &\eqdef \rappl{\fullReduct{s}}{\fullReduct{\ms t}}
			&&\text{(if $s$ is not an abstraction)}\\
		\fullReduct{\mset{s_1,\dotsc,s_n}}  &\eqdef \mset{\fullReduct{s_1},\dotsc,\fullReduct{s_n}}.
	\end{align*}
	Then if $ε=\sum_{i=1}^n e_i\in\finiteResourceSums$,
	we set $\fullReduct{ε}=\sum_{i=1}^n\fullReduct{e_i}$.
\end{definition}

\begin{lemma}
	\label{lemma:presRed:fullReduct}
	For all $ε,ε'\in\finiteResourceSums$, 
	if $ε\presRed ε'$ then $ε'\presRed \fullReduct{ε}$.
\end{lemma}
\begin{proof}
	Follows directly from the definitions.
\end{proof}

In general, however, if we fix $e'\in\resourceExpressions$ then there is no
bound on those $e\in\resourceExpressions$ such that $e'\in\support{\fullReduct e}$,
so we cannot extend $\fullReductSym$ on $\resourceVectors$, nor generalize Lemma
\ref{lemma:presRed:fullReduct} to $\vpresRed$.
Indeed, we have shown that $\vpresRed$ is not even confluent.

In order to understand what restrictions are necessary to recover confluence,
we first provide a close inspection of the combinatorial effect of
$\presRed$ on the size of resource expressions: we show in subsection
\ref{subsection:pbresRed} that bounding the length of chains 
of immediately nested fired redexes is enough to bound the size of
$\presRed$-antecedents of a fixed resource expression.

In order to close a pair of reductions $e\presRed ε'$ and $e\presRed ε''$,
we have to reduce at least the residuals in $ε'$ of the redexes fired 
in the reduction $e\presRed ε''$ (and vice versa). So we want the above
bounds to be stable under taking the unions of sets of redexes in a term:
it is not the case if we consider chains of immediately nested redexes.
In Subsection \ref{subsection:pbdresRed}, we extend the boundedness condition
to all chains of nested fired redexes and introduce the family
$\pars{\pbdresRed b}_{b\in\naturals}$ of boundedly nested parallel reductions.
We then show that this family enjoys a kind of diamond property (Lemma
\ref{lemma:pbdresRed:confluent}), which can then be extended to
$\mathord{\vpbdresRed}=\Union_{b\in\naturals}\mathord{\splitVariant\pbdresRed b}$.
We must require that $\rigS$ enjoys an additional \emph{additive splitting} property (see
Definition \ref{definition:splitting}), in order to 
``align'' the $\pbdresRed b$-reductions involved in both sides of a pair
of $\splitVariant\pbdresRed b$-reductions from the same resource vector
(see the proof of Lemma \ref{lemma:vpbdresRed:diamond}).

To get rid of the additive splitting hypothesis we must further restrict 
resource reduction so as to recover a notion of full reduct \emph{at bounded depth}.
It is not sufficient to bound the depth of fired redexes because this is not
stable under reduction.
In Subsection \ref{subsection:pbsresRed}, we rather introduce 
the parallel reduction $\pbsresRed d$ where 
substituted variables occur at depth at most $d$.
We then show that
$\mathord{\vpbsresRed}=\Union_{d\in\naturals}\mathord{\splitVariant\pbsresRed d}$
is strongly confluent by proving that any $\pbsresRed d$-step from $ε$
can be followed by a $\pbsresRed {d'}$-reduction to $\fpbsReduct {d}{ε}$,
where $\fpbsReduct{d}{ε}$ is obtained by firing all redexes in $ε$ for which
the bound variables occur at depth at most $d$, and $d'$ depends only on $d$.

Finally, we consider resource vectors of bounded height: these contain the
Taylor expansions of algebraic λ-terms. We show that all the above restrictions 
actually coincide with $\vpresRed$  on bounded resource vectors. In this particular
case, we can actually extend $\fullReductSym$ by linear-continuity and 
obtain a proof of the diamond property for $\vpresRed$.\footnote{
	Note that, although they involve increasing constraints on parallel
	reduction, Subsections \ref{subsection:pbresRed} to \ref{subsection:pbsresRed}
	are essentially pairwise independent.
	Moreover, we obtain the diamond property for $\vpresRed$ on bounded resource
	vectors as a consequence of the results of Subsection \ref{subsection:pbsresRed},
	but it could as well be proved directly, using similar techniques
	(see Footnote \ref{footnote:direct}, p.\pageref{footnote:direct}).
	So, the reader who only wants the proofs necessary for the main results of
	the paper can skip Subsections \ref{subsection:pbresRed} and \ref{subsection:pbdresRed};
	the reader who is not interested in checking proofs can also skip subsection
	\ref{subsection:pbsresRed}.

	We chose to present the successive families of restrictions anyway, because
	their construction provides a precise understanding of the combinatorics of
	parallel resource reduction, and of the various ingredients involved in
	designing a strongly confluent version of $\vpresRed$: we start by avoiding
	the size collapse by putting a restriction on families of redexes that can be
	fired in parallel; then we ensure that this restriction is stable under
	reduction.

	This understanding plays a key rôle in enabling the generalization of our
	approach to linear logic proof nets or infinitary λ-calculus:
	with Chouquet, we have recently established that our restrictions on
	the nesting of redexes, as well as their preservation under reduction, can be
	adapted to the setting of proof nets \cite{cv:antireduits-csl};
	and preliminary work on infinitary λ-calculus indicates that it could be amenable
	to the technique of Subsection \ref{subsection:pbdresRed}, whereas 
	it does not make sense to restrict the depth of substituted variables in this setting.
}

At this point of the discussion, it is worth noting that, if we extend 
a relation $\mathord{\reduce}\subseteq{\resourceExpressions\times\finiteResourceSums}$
to a binary relation on finite sums of resource expressions so that 
$ε\reduce ε'$ iff $ε=\sum_{i=1}^n e_i$ and $ε'=\sum_{i=1}^n ε'_i$
with $e_i\reduce ε'_i$ for all $i\in\set{1,\dotsc,n}$, then 
for all $\reduce$-reduction structure $\fE$ and all 
resource vectors $ε,ε'\in\resourceVectors$, we have 
$ε\splitVariant[\fE]\reduce ε'$ iff there exist a set $I$ of indices,
a resource support $\calE\in\fE$,
a family $(a_i)_{i\in I}\in \rigS^I$ of scalars
and families $\pars{ε_i}_{i\in I}\in\finiteSums{\calE}^I$ 
and $\pars{ε'_i}_{i\in I}\in \finiteResourceSums^I$ 
such that:
\begin{itemize}
	\item $\pars{ε_i}_{i\in I}$ is summable and $ε=\sum_{i\in I}a_i\sm ε_i$;
	\item $\pars{ε'_i}_{i\in I}$ is summable and $ε'=\sum_{i\in I}a_i\sm ε'_i$;
	\item for all $i\in I$, $ε_i\reduce ε'_i$.
\end{itemize}
We will use this fact for confluence proofs: $\presRed$ and its variants are
all of this form.

\subsection{Bounded chains of redexes}

\label{subsection:pbresRed}

\begin{definition}
	We define a family of relations $\mathord{\pbresRed mk}
	\subseteq\resourceExpressions\times\finiteResourceSums$
	for $m\le k\in\naturals$ inductively as follows:
	\begin{itemize}
		\item $x\pbresRed mk x$ for all $m\le k\in\naturals$;
		\item $\labs x s\pbresRed mk \labs x{σ'}$ if $m\le k$ and $s\pbresRed {m_1}{k} σ'$ for some $m_1\le k$;
		\item $\rappl s{\ms t}\pbresRed mk\rappl{σ'}{\ms {τ'}}$ if $m\le k$,
			$s\pbresRed {m_1}{k} σ'$ and $\ms t\pbresRed {m_2}{k} \ms{τ'}$ 
			for some $m_1\le k$ and $m_2\le k$;
		\item $\mset{s_1,\dotsc,s_r}\pbresRed mk\mset{σ'_1,\dotsc,σ'_r}$
			if $s_i\pbresRed {m}{k} σ'_i$ for all $i\in\set{1,\dotsc,r}$;
		\item $\rappl{\labs xs}{\ms t}\pbresRed mk \lsubst{σ'}x{\ms{τ'}}$ if 
			 $0<m\le k$, $s\pbresRed {m-1}{k} σ'$ and $\ms t\pbresRed {m-1}{k} \ms{τ'}$.
	\end{itemize}
\end{definition}
Intuitively, we have $e\pbresRed mk ε'$ iff $e\presRed ε'$ and 
that reduction fires chains of redexes of length at most $k$, 
those starting at top level being of length at most $m$.
In particular, it should be clear that if $e\pbresRed mk ε'$
then $e\presRed ε'$, and $e\pbresRed {m'}{k'} ε'$ as soon as $m\le m'\le k'$ and $k\le k'$.
Moreover, $e\presRed ε'$ iff $e\pbresRed{\height e}{\height e}ε'$.

\begin{definition}
	We define $\growthBound klm\in\naturals$ for all $k,l,m\in\naturals$,
	by induction on the lexicographically ordered pair $(l,m)$:
	\begin{eqnarray*}
		\growthBound k00 &\eqdef& 0\\
		\growthBound k{l+1}0 &\eqdef& \growthBound klk+1\\
		\growthBound kl{m+1} &\eqdef& 4\growthBound klm.
	\end{eqnarray*}
	We then write $\growthBoundToplevel kl\eqdef \growthBound klk$.
\end{definition}

For all $k,l,m\in\naturals$, the following identities follow straightforwardly
from the definition and will be used throughout this subsection:
\begin{eqnarray*}
	\growthBound klm     &=&4^m\growthBound kl0\\
	\growthBound k0m     &=&0\\
	\growthBound k1m     &=&4^m.
\end{eqnarray*}

\begin{lemma}
	\label{lemma:growthBound:sum}
	For all $k,l,l',m\in\naturals$, $\growthBound k{l+l'}m\ge\growthBound klm+\growthBound k{l'}m$.
\end{lemma}
\begin{proof}
	By induction on $l'$.
	The case $l'=0$ is direct.
	Assume the result holds for $l'$, we prove it for $l'+1$:
	\begin{eqnarray*}
		\growthBound k{l+l'+1}m
		&=&
		4^m\pars{\growthBound k{l+l'+1}{0}}
		\\
		&=&
		4^m\pars{\growthBound k{l+l'}k +1}
		\\
		&\ge&
		4^m\pars{\growthBound klk+\growthBound k{l'}k+1}
		\\
		&=&
		4^m\pars{4^k\growthBound kl0+\growthBound k{l'+1}0}
		\\
		&\ge&
		4^m\growthBound kl0+4^m\growthBound k{l'+1}0
		\\
		&=& 
		\growthBound klm+\growthBound k{l'+1}m.
	\end{eqnarray*}
\end{proof}

The following generalization follows directly:

\begin{corollary}
	\label{corollary:growthBound:sum}
	For all $l_1\dotsc,l_n\in\naturals$
	\[\growthBound k{\sum_{i=1}^n l_i}m\ge\sum_{i=1}^n \growthBound k{l_i}m.\]
\end{corollary}

\begin{lemma}
	\label{lemma:growthBound:l}
	For all $k,l,m\in\naturals$, 
	$\growthBound klm\ge l$.
\end{lemma}
\begin{proof}
	By Corollary \ref{corollary:growthBound:sum},
	$\growthBound klm\ge l\times\growthBound k1m=l\times 4^m$.
\end{proof}

\begin{lemma}
	\label{lemma:growthBound:monotonicity}
	For all $k,k',l,l',m,m'\in\naturals$ 
	if $k\le k'$, $l\le l'$ and $m\le m'$, then:
	\[
		\growthBound k{l}{m} \le \growthBound {k'}{l'}{m'}.
	\]
\end{lemma}
\begin{proof}
	We prove the monotonicity of $\growthBound klm$ in $m$, $l$ and then $k$, separately.

	First, if $m\le m'$ then
	$\growthBound klm=4^m\growthBound kl0\le 4^{m'}\growthBound kl0 =\growthBound kl{m'}$.

  By Lemma \ref{lemma:growthBound:sum}, if $l\le l'$, 
	$\growthBound k{l'}m\ge\growthBound klm+\growthBound k{l'-l}m\ge\growthBound klm$.

	Finally, we prove that if $k\le k'$ then $\growthBound klm\le\growthBound{k'}lm$
	by induction on the lexicographically ordered pair $(l,m)$:
  \begin{gather}
	\begin{align*}
		\growthBound k00
		&~=~
		0
		\\
		&~=~
		\growthBound{k'}00
		\\
		\growthBound k{l+1}0
		&~=~
		\growthBound klk+1
		\\
		&~\le~
		\growthBound {k'}lk+1
		\\
		&~\le~
		\growthBound {k'}l{k'}+1
		\\
		&~=~
		\growthBound {k'}{l+1}{0}
		\\
		\growthBound kl{m+1}
		&~=~
		4\growthBound klm
		\\
		&~\le~
		4\growthBound {k'}lm
		\\
		&~=~
		\growthBound {k'}l{m+1}\tag*{\qedhere}
	\end{align*}
  \end{gather}
\end{proof}

Write $e\pbOneStepGenerates mke'$ if $e\pbresRed mk ε'$ with $e'\in\support{ε'}$.
\begin{lemma}
	If $e\pbOneStepGenerates mk e'$ then
	$\size{e}\le\growthBound{k}{\size{e'}}{m}$.
\end{lemma}
\begin{proof}
	By induction on the reduction $e\pbresRed mk ε'$ such that $e'\in ε'$.
	
	If $e=x=ε'$ then $e'=x$ and $\size{e}=1=\growthBound{0}{1}{0}\le\growthBound{k}{\size{e'}}{m}$.

	If $e=\labs xs$, $ε'=\labs x{σ'}$, $m\le k$ and $s\pbresRed {m_1}{k} σ'$ with $m_1\le k$,
	then $e'=\labs x{s'}$ with $s\pbOneStepGenerates {m_1}k s'$. We obtain:
	\begin{align*}
	\size{e}
	& = \size{s}+1 
	\\
	& \le \growthBound{k}{\size{s'}}{m_1}+1
	&& \text{(by induction hypothesis)}
	\\
	& \le \growthBound{k}{\size{s'}}{k}+1
	\\
	& = \growthBound{k}{\size{s'}+1}{0}
	\\
	& \le \growthBound{k}{\size{e'}}{m}.
	\end{align*}

	If $e=\rappl s{\ms t}$,
	$ε'=\rappl{σ'}{\ms{τ'}}$,
	$m\le k$,
	$s\pbresRed {m_1}{k} σ'$ and
	$\ms t\pbresRed {m_2}{k} \ms{τ'}$
	with $m_i\le k$ for all $i\in\set{1,2}$,
	then $e'=\rappl{s'}{\ms{t'}}$ with 
	$s\pbOneStepGenerates {m_1}k s'$
	and $\ms t\pbOneStepGenerates {m_2}k\ms{t'}$.
	We obtain:
	\begin{align*}
		\size{e} &= \size{s}+\size{\ms t}+1 
		\\
		&\le \growthBound{k}{\size{s'}}{m_1}+\growthBound{k}{\size{\ms {t'}}}{m_2}+1
		&&\text{(by induction hypothesis)}
		\\
		&\le \growthBound{k}{\size{s'}}{k}+\growthBound{k}{\size{\ms {t'}}}{k}+1
		\\
		&\le \growthBound{k}{\size{s'}+\size{\ms{t'}}}{k}+1
		\\
		&= \growthBound{k}{\size{s'}+\size{\ms{t'}}+1}{0}
		\\
		&\le \growthBound{k}{\size{e'}}{m}.
	\end{align*}

	If $e=\mset{s_1,\dotsc,s_r}$,
	$ε'=\mset{σ'_1,\dotsc,σ'_r}$ and
	$s_i\pbresRed{m}{k} σ'_i$ 
	for all $i\in\set{1,\dotsc,r}$,
	then $e'=\mset{s'_1,\dotsc,s'_r}$ 
	with $s_i\pbOneStepGenerates mk s'_i$
	for all $i\in\set{1,\dotsc,r}$.
	We obtain:
	\begin{align*}
		\size{e}
		&= \sum_{i=1}^r \size{s_i}
		\\
		&\le \sum_{i=1}^r\growthBound{k}{\size{s'_i}}{m}
		&&\text{(by induction hypothesis)}
		\\
		& \le \growthBound{k}{\sum_{i=1}^r\size{s'_i}}{m}
		\\
		& = \growthBound{k}{\size{e'}}{m}.
	\end{align*}

	If $e=\rappl {\labs xs}{\ms t}$,
	$ε'=\lsubst{σ'}x{\ms{τ'}}$,
	$0<m\le k$, $s\pbresRed {m-1}{k} σ'$ and
	$\ms t\pbresRed {m-1}{k} \ms{τ'}$,
	then there are $s'\in\support{σ'}$ and 
	$\ms{t'}\in\support{\ms{τ'}}$ such that 
	$e'\in\support{\lsubst {s'}x{\ms t'}}$.
	In particular, $s\pbOneStepGenerates {m-1}{k}s'$
	and $\ms t\pbOneStepGenerates {m-1}{k}\ms{t'}$ 
	and we obtain:
	\begin{align*}
		\size{e}
		& = \size{s}+\size{\ms t}+2
		\\
		& \le \growthBound{k}{\size{s'}}{m-1}+\growthBound{k}{\size{\ms {t'}}}{m-1}+2
		&&\text{(by induction hypothesis)}
		\\
		& \le 2\growthBound{k}{\size{e'}}{m-1}+2
		&&\text{($\size{e'}\ge\max\set{\size{s'},\size{\ms {t'}}}$)}
		\\
		& \le 4\growthBound{k}{\size{e'}}{m-1}
		&&\text{($\size{e'}\ge\size{s'}\ge 1$)}
		\\
		&= \growthBound{k}{\size{e'}}{m}.&&\hfill\qedhere
	\end{align*}
\end{proof}

As a direct consequence, for all $m\le k\in\naturals$, for all summable family
$\pars{e_i}_{i\in I}$ and all family $\pars{ε'_i}_{i\in I}$ such that
$e_i\pbresRed mk ε'_i$ for all $i\in I$, $\pars{ε'_i}_{i\in I}$
is also summable: we can thus drop the side condition in the definition of 
$\splitVariant\pbresRed mk$.

Observe however that those reduction relations are not stable under taking the
unions of fired redexes in families of reduction steps: using, \eg,
the terms $u_n(s)$ from Example \ref{example:reduction:collapse},
for all $n\in\naturals$, we have $u_{2n}(s)\pbresRed 11 u_n(s)$ by 
firing all redexes at even depth,
$u_{2n}(s)\pbresRed 01 u_n(s)$ by 
firing all redexes at odd depth,
and $u_{2n}(s)\pbresRed {2n}{2n} s$ by firing both families,
but there is obviously no $k\in\naturals$ such that 
$u_{2n}(s)\pbresRed kk s$ uniformly for all $n\in\naturals$.
Although we can close the induced critical pair 
\begin{center}
$\sum_{n\in\naturals} u_{2n}(s)\splitVariant\pbresRed 01\sum_{n\in\naturals} u_n(s)$ 
and
$\sum_{n\in\naturals} u_{2n}(s)\splitVariant\pbresRed 11 \sum_{n\in\naturals} u_n(s)$
\end{center}
trivially in this case, this phenomenon is an obstacle to confluence:
\begin{example}
	\label{example:pbresRed:not:confluent}
	Fix $s\in\resourceTerms$ and
	consider the sequence $\vec w(s)$ of resource terms given by 
	$w_0(s)=s$ and:
	\begin{eqnarray*}
		w_{2n+1}(s) &=& \rappl{\labs yy}{\mset{w_{2n}(s)}}\\
		w_{2n+2}(s) &=& \rappl{\labs y{w_{2n+1}(s)}}{\mset{}}
	\end{eqnarray*}
	Then for all $n\in\naturals$, $w_{2n}(s)\pbresRed 11 u_n(s)$,
	$w_{2n+1}(s)\pbresRed 01 u_n(s)$, $w_{2n}(s)\pbresRed 01 v_n(s)$,
	and $w_{2n+1}(s)\pbresRed 11 v_n(s)$.
	Then for instance
	\[
	\sum_{n\in\naturals} w_{2n}(s)\splitVariant\pbresRed 11\sum_{n\in\naturals} u_{n}(s)
	\quad\text{and}\quad
	\sum_{n\in\naturals} w_{2n}(s)\splitVariant\pbresRed 01\sum_{n\in\naturals} v_{n}(s)\]
	but we know from Example \ref{example:presRed:not:confluent} that this pair 
	of reductions cannot be closed in general.
\end{example}

\subsection{Boundedly nested redexes}

\label{subsection:pbdresRed}

From the previous subsection, it follows that
bounding the length of chains of immediately nested redexes allows to tame
the size collapse of resource expressions under reduction,
but we need to further restrict this notion 
in order to keep it stable under unions of fired redex sets.
A natural answer is to require a bound on the depth of the nesting of 
fired redexes, regardless of the distance between them:
\begin{definition}
	We define a family of relations $\pars{\pbdresRed b}_{b\in{\naturals}}$
	inductively as follows:
	\begin{itemize}
		\item $x\pbdresRed b x$ for all $b\in\naturals$;
		\item $\labs x s\pbdresRed b \labs x{σ'}$ if $s\pbdresRed b σ'$;
		\item $\rappl s{\ms t}\pbdresRed {b} \rappl{σ'}{\ms {τ'}}$
			if $s\pbdresRed b σ'$ and $\ms t\pbdresRed b \ms{τ'}$;
		\item $\mset{s_1,\dotsc,s_r}\pbdresRed b \mset{σ'_1,\dotsc,σ'_r}$
			if $s_i\pbdresRed b σ'_i$ for all $i\in\set{1,\dotsc,r}$;
		\item $\rappl{\labs xs}{\ms t}\pbdresRed b \lsubst{σ'}x{\ms{τ'}}$ 
			if $b\ge 1$, $s\pbdresRed {b-1} σ'$ and $\ms t\pbdresRed {b-1} \ms{τ'}$.
	\end{itemize}
\end{definition}

Intuitively, we have $e\pbdresRed b ε'$ iff $e\presRed ε'$ and 
every branch of $e$ (seen as a rooted tree) crosses at most $b$ fired redexes.
In particular it should be clear that if $e\pbdresRed b ε'$
then $e\pbresRed b b ε'$, and moreover $e\pbdresRed {b'}ε'$
for all $b'\ge b$.
Moreover observe that $e\pbdresRed{\height e}ε'$ whenever $e\presRed ε'$, hence
$\mathord{\presRed} = \Union_{b\in\naturals} \mathord{\pbdresRed b}$.

Write $e\pbdOneStepGenerates{b} e'$ if $e\pbdresRed b ε'$ with $e'\in \support{ε'}$.
If $e\pbdOneStepGenerates{b} e'$, then $e\pbOneStepGenerates bb e'$ and we 
thus know that $\size{e}\le\growthBoundToplevel b{\size{e'}}$.
In this special case, we can in fact provide a much better bound:
\begin{lemma}
	\label{lemma:pbdresRed:size}
	If $e\pbdOneStepGenerates{b}e'$ then $\size{e}\le 4^b\size{e'}$.
\end{lemma}
\begin{proof}
	By induction on the reduction $e\pbdresRed b{ε'}$
	such that $e'\in\support{ε'}$.

	If $e=x=ε'$ then $e'=x$ and $\size{e}=1\le 4^b =4^b\size{e'}$.

	If $e=\labs xs$ and $ε'=\labs x{σ'}$ with $s\pbdresRed {b} σ'$,
	then $e'=\labs x{s'}$ with $s\pbdOneStepGenerates bs'$.
	By induction hypothesis, $\size{s}\le 4^{b} \size{s'}$.
	Then
	$\size{e}
	=\size s +1
	\le 4^{b}\size{s'}+1
	\le 4^b\pars{\size{s'}+1}
	=4^b\size{e'}$.

	If $e=\rappl s{\ms t}$, $ε'=\rappl{σ'}{\ms{τ'}}$,
	$s\pbdresRed {b} σ'$ and $\ms t\pbdresRed {b} \ms{τ'}$,
	then $e'=\rappl{s'}{\ms{t'}}$ with 
	$s\pbdOneStepGenerates{b}s'$ and $\ms{t}\pbdOneStepGenerates{b}\ms{t'}$.
	By induction hypothesis, $\size{s}\le 4^{b}\size{s'}$ and
	$\size{\ms t}\le 4^{b}\size{\ms {t'}}$.
	Then 
	$ \size{e}
	= \size{s}+\size{\ms t}+1
	\le 4^{b}\size{s'}+4^{b}\size{\ms{t'}}+1
	\le 4^{b}\pars{\size{s'}+\size{\ms{t'}}+1}
	= 4^b\size{e'}
	$.

	If $e=\mset{s_1,\dotsc,s_r}$,
	$ε'=\mset{σ'_1,\dotsc,σ'_r}$ and
	$s_i\pbdresRed{b} σ'_i$
	for all $i\in\set{1,\dotsc,r}$,
	then $e'=\mset{s'_1,\dotsc,s'_r}$ 
	with $s_i\pbdOneStepGenerates{b}s'_i$
	for all $i\in\set{1,\dotsc,r}$.
	By induction hypothesis,
	$\size{s_i}\le 4^{b}\size{s'_i}$
	for all $i\in\set{1,\dotsc,r}$ and then
	$ \size{e}
	= \sum_{i=1}^r \size{s_i}
	\le \sum_{i=1}^r 4^b{\size{s'_i}}
	= 4^b{\size{e'}}
	$.

	If $e=\rappl {\labs xs}{\ms t}$,
	$ε'=\lsubst{σ'}x{\ms{τ'}}$, $b>0$,
	$s\pbdresRed{b-1}σ'$ and
	$\ms t\pbdresRed {b-1} \ms{τ'}$,
	then there are $s'\in\support{σ'}$ and 
	$\ms{t'}\in\support{\ms{τ'}}$ such that 
	$e'\in\support{\lsubst {s'}x{\ms{t'}}}$.
	In particular, $s\pbdOneStepGenerates{b-1}s'$
	and $\ms t\pbdOneStepGenerates{b-1}\ms{t'}$ and,
	by induction hypothesis,
	$\size{s}\le 4^{b-1}\size{s'}$ and
	$\size{\ms t}\le 4^{b-1}\size{\ms {t'}}$.
	Writing $n=\occnum x{s'}=\card{\ms{t'}}$,
	we have:
	\begin{align*}
	4^b\size{e'}
	&= 4^b\pars{\size{s'}+\size{\ms{t'}}-n}
	\\
	&= 4^{b-1}\pars{\size{s'}+\size{\ms{t'}}+3\size{s'}+3\size{\ms{t'}}-4n}
	&&\text{($n\le\size{s'}$ and $n\le\size{\ms{t'}}$)}
	\\
	&\ge 4^{b-1}\pars{\size{s'}+\size{\ms{t'}}+2\size{s'}}
	&&\text{($\size{s'}\ge 1$)}
	\\
	&\ge 4^{b-1}\pars{\size{s'}+\size{\ms{t'}}}+2
	\\
	&\ge \size{s}+\size{\ms t}+2
	\\
	&=\size{e}.&&\hfill\qedhere
	\end{align*}
\end{proof}

Like for parallel reduction (Definition \ref{definition:presRed}),
we extend each $\pbdresRed b$ to sums of resource expressions by linearity:
$ε\pbdresRed bε'$ if $ε=\sum_{i=1}^n e_i$ and $ε'=\sum_{i=1}^n ε'_i$ 
with $e_i\pbdresRed bε'_i$ for all $i\in\set{1,\dotsc,n}$.
Again, because all term constructors are linear, the reduction rules extend 
naturally to finite sums of resource expressions:
for instance,$\rappl{\labs x{σ}}{\ms {τ}}\pbdresRed b \lsubst{σ'}x{\ms{τ'}}$ 
as soon as $b\ge 1$, $σ\pbdresRed {b-1} σ'$ and $\ms {τ}\pbdresRed {b-1}\ms{τ'}$.

The relations $\pbdresRed b$ are then stable under unions of
families of fired redexes, avoiding pitfalls such as that of Example
\ref{example:pbresRed:not:confluent}.

\begin{lemma}
	If $e\pbdresRed {b_0} ε'$ and $\ms u\pbdresRed {b_1} \ms{υ'}$ then
	$\lsubst ex{\ms u}\pbdresRed{b_0+b_1} \lsubst {ε'}x{\ms{υ'}}$.
\end{lemma}
\begin{proof}
	Write $\ms u=\mset{u_1,\dotsc,u_n}$. Then we can 
	write $\ms {υ'}=\mset{υ'_1,\dotsc,υ'_n}$
	with $u_i\pbdresRed {b_1} {υ'}_i$
	for all $i\in\set{1,\dotsc,n}$.
	Recall that whenever
	$I=\set{i_1,\dotsc,i_k}\subseteq\set{1,\dotsc,n}$ with $\card I=k$,
	we write $\ms u_I=\mset{u_{i_1},\dotsc,u_{i_k}}$
	and $\ms {υ'}_I=\mset{υ'_{i_1},\dotsc,υ'_{i_k}}$.

	The proof is by induction on the reduction $e\pbdresRed {b_0} ε'$.
	If $e=y=ε'$ then:
	\begin{itemize}
		\item if $y=x$ and $n=1$ then 
			$\lsubst ex{\ms u}=u_1\pbdresRed{b_1}υ'_1=\lsubst{ε'}x{\ms{υ'}}'$;
		\item if $y\not=x$ and $\ms u=\mset {}$ then
			$\lsubst ex{\ms u}=y\pbdresRed {0} y=\lsubst{ε'}x{\ms{υ'}}$;
		\item otherwise, 
			$\lsubst ex{\ms u}=0\pbdresRed {0} 0=\lsubst{ε'}x{\ms{υ'}}$.
	\end{itemize} 

	If $e=\labs ys$ (choosing $y\not=x$ and $y\not\in\fv{\ms u}$),
	$ε'=\labs y{σ'}$ and $s\pbdresRed{b_0} σ'$ then, by induction hypothesis,
	$\lsubst sx{\ms u}\pbdresRed{b_0+b_1} \lsubst {σ'}x{\ms {υ'}}$.
	We obtain:
	$
	\lsubst ex{\ms u}
	=\labs y[]{\lsubst sx{\ms u}}
	\pbdresRed {b_0+b_1} \labs y[]{\lsubst {σ'}x{\ms {υ'}}}
	=\lsubst {ε'}x{\ms {υ'}}
	$.

	If $e=\rappl s{\ms t}$, $ε'=\rappl {σ'}{\ms {τ'}}$,
	$s\pbdresRed{b_0} σ'$ and
	$\ms t\pbdresRed{b_0} \ms {τ'}$ then,
	by induction hypothesis,
	$\lsubst sx{\ms u_I}\pbdresRed{b_0+b_1}\lsubst {σ'}x{\ms {υ'}_I}$ and
	$\lsubst {\ms t}x{\ms u_I}\pbdresRed{b_0+b_1}\lsubst {\ms{τ'}}x{\ms {υ'}_I}$,
	for all $I\subseteq\set{1,\dotsc,n}$.
	We obtain:
	\[
		\lsubst ex{\ms u}
		= \sum_{\substack{\text{$(I,J)$ partition}\\\text{of $\set{1,\dotsc,n}$}}}
			\rappl{\lsubst sx{\ms u_I}}{\lsubst {\ms t}x{\ms u_J}}
		\\
		\pbdresRed{b_0+b_1}\sum_{\substack{\text{$(I,J)$ partition}\\\text{of $\set{1,\dotsc,n}$}}}
			\rappl{\lsubst {σ'}x{\ms {υ'}_I}}{\lsubst {\ms{τ'}}x{\ms {υ'}_J}}
		=\lsubst {ε'}x{\ms {υ'}}.
	\]

	If $e=\mset{s_1,\dotsc,s_r}$, $ε'=\mset{σ'_1,\dotsc,σ'_r}$
	and $s_i\pbdresRed {b_0} σ'_i$ for all $i\in\set{1,\dotsc,r}$ then,
	by induction hypothesis,
	$\lsubst{s_i}x{\ms u_I}\pbdresRed{b_0+b_1}\lsubst{σ'_i}x{\ms{υ'}_I}$
	for all $i\in\set{1,\dotsc,r}$ and all $I\subseteq\set{1,\dotsc,n}$.
	We obtain:
	\begin{multline*}
		\lsubst ex{\ms u} =
		\sum_{\substack{\text{$(I_1,\dotsc,I_r)$ partition}\\\text{of $\set{1,\dotsc,n}$}}}
			\mset{\lsubst{s_1}x{\ms u_{I_1}},\dotsc,\lsubst{s_r}x{\ms u_{I_r}}}
		\\
		\pbdresRed{b_0+b_1}
		\sum_{\substack{\text{$(I_1,\dotsc,I_r)$ partition}\\\text{of $\set{1,\dotsc,n}$}}}
			\mset{\lsubst{σ'_1}x{\ms{υ'}_{I_1}},\dotsc,\lsubst{σ'_r}x{\ms{υ'}_{I_r}}}
		=\lsubst {ε'}x{\ms {υ'}}.
	\end{multline*}

	If $e=\rappl{\labs ys}{\ms t}$
	(choosing $y\not=x$ and $y\not\in\fv{\ms t}\union\fv{\ms u}$),
	$ε'=\lsubst{σ'}y{\ms{τ'}}$, $b_0\ge 1$,
	$s\pbdresRed {b_0-1} σ'$ and $\ms t\pbdresRed {b_0-1}\ms{τ'}$ then,
	by induction hypothesis, 
	$\lsubst sx{\ms u_I}\pbdresRed {b_0+b_1-1}\lsubst {σ'}x{\ms {υ'}_I}$ and
	$\lsubst {\ms t}x{\ms u_I}\pbdresRed {b_0+b_1-1}\lsubst {\ms{τ'}}x{\ms {υ'}_I}$,
	for all $I\subseteq\set{1,\dotsc,n}$.
	We obtain:
	\begin{multline*}
		\lsubst ex{\ms u}
		=\sum_{\substack{\text{$(I,J)$ partition}\\\text{of $\set{1,\dotsc,n}$}}}
			\rappl{\labs y{\lsubst sx{\ms u_I}}}{\lsubst {\ms t}x{\ms u_J}}
		\\
		\pbdresRed {b_0+b_1}
			\sum_{\substack{\text{$(I,J)$ partition}\\\text{of $\set{1,\dotsc,n}$}}}
			\lsubst{\pars{\lsubst {σ'}x{\ms {υ'}_I}}}y{\pars{\lsubst {\ms{τ'}}x{\ms {υ'}_J}}}
		=\lsubst {\pars{\lsubst{σ'}y{\ms{τ'}}}}x{\ms {υ'}}
		=\lsubst {ε'}x{\ms {υ'}}
	\end{multline*}
	using Lemma \ref{lemma:lsubst:commute}.
\end{proof}

\begin{lemma}
	\label{lemma:pbdresRed:confluent}
	Let $K$ be a finite set, and assume $ε\pbdresRed {b_k} ε'_k$ for all $k\in K$.
	Then, setting $b=\sum_{k\in K} b_k$, there is $ε''$ such that
	$ε'_k\pbdresRed {2^{\scriptstyle b_k}b} ε''$ for all $k\in K$.
\end{lemma}
\begin{proof}
	By the linearity of the definition of reduction on finite sums,
	it is sufficient to address the case of $ε=e\in\resourceExpressions$.
	The proof is then by induction on the family of reductions 
	$e\pbdresRed {b_k} ε'_k$.

	If $e=x=ε'_k$ for all $k\in K$, then we set $ε''=x$.

	If $e=\labs xs$, and $ε'_k=\labs x{σ'_k}$ with
	$s\pbdresRed {b_k} σ'_k$ for all $k\in K$ then,
	by induction hypothesis, we have $σ''$ such that
	$σ'_k\pbdresRed {2^{\scriptstyle b_k}b} σ''$ for all $k\in K$,
	and then we set $ε''=\labs x{σ''}$.

	If $e=\mset{s_1,\dotsc,s_r}$ and $ε'_k=\mset{σ'_{1,k},\dotsc,σ'_{r,k}}$ 
	with $s_j\pbdresRed {b_k} σ'_{j,k}$ for all $j\in\set{1,\dotsc,r}$ and $k\in K$ then,
	by induction hypothesis, we have $σ''_j$ such that  
	$σ'_{j,k}\pbdresRed {2^{\scriptstyle b_k}b} σ''_j$
	for all $j\in\set{1,\dotsc,r}$ and $k\in K$, and then we set 
  $ε''=\mset{σ''_1,\dotsc,σ''_r}$.

	Finally, assume $K=K_0+K_1$, $e=\rappl {\labs xs}{\ms t}$ and:
	\begin{itemize}
		\item for all $k\in K_0$, $ε'_k=\rappl {\labs x{σ'_k}}{\ms {τ'}_k}$ 
			with $s\pbdresRed {b_k} σ'_k$ and $\ms t\pbdresRed {b_k} \ms {τ'}_k$;
		\item for all $k\in K_1$, $b_k\ge 1$ and $ε'_k=\lsubst {σ'_k} x{\ms {τ'}_k}$
			with $s\pbdresRed {b_k-1} σ'_k$ and $\ms t\pbdresRed {b_k-1} \ms {τ'}_k$.
	\end{itemize}
	Write $b'=b-\card{K_1}$.
	By induction hypothesis, there are $σ''$ and $\ms{τ''}$ such that,
	for all $k\in K_0$, $σ'_k\pbdresRed{2^{\scriptstyle b_k}b'}σ''$
	and $\ms{τ'}_k \pbdresRed{2^{\scriptstyle b_k}b'}\ms{τ''}$, and
	for all $k\in K_1$, $σ'_k\pbdresRed[\mbig]{2^{\scriptstyle (b_k-1)}b'}σ''$
	and $\ms{τ'}_k \pbdresRed[\mbig]{2^{\scriptstyle (b_k-1)}b'}\ms{τ''}$.

	If $K_1=\emptyset$ then $b=b'$ and we set $ε''=\rappl{\labs x{σ''}}{\ms{τ''}}$:
	we obtain $ε'_k\pbdresRed{2^{\scriptstyle b_k}b}ε''$, for all $k\in K=K_0$.

	Otherwise, $b>b'$ and we set $ε''=\lsubst{σ''}x{\ms{τ''}}$ so that:
	\begin{itemize}
		\item for all $k\in K_0$, 
			$ε'_k=\rappl {\labs x{σ'_k}}{\ms {τ'}_k}\pbdresRed {2^{\scriptstyle b_k}b'+1} σ'_k$
			with $2^{b_k}b'+1\le 2^{b_k}b$;
		\item for all $k\in K_1$, 
			by the previous lemma,
			$ε'_k=\lsubst {σ'_k}x{\ms {τ'}_k}\pbdresRed {2^{\scriptstyle b_k}b'} σ'_k$
			and $2^{b_k}b'<2^{b_k}b$.\qedhere
	\end{itemize}
\end{proof}

We already know $\vpresRed$ is not confluent,
and the counter examples we provided actually 
show that no single $\splitVariant\pbdresRed b$ is confluent either.
Setting\footnote{
	Our notation is somehow abusive as $\vpbdresRed$ 
	is not of the form described in 
	Definition \ref{definition:splitVariant}:
	there should not be any ambiguity as we have not 
	defined any relation $\pbdresRed{\partial}$.
	Similarly, we may also write 
	$\vpbdresRed[\fE]$ for 
	$\Union_{b\in\naturals}\mathord{\splitVariant[\fE]\pbdresRed b}$
	in the following.
}
\[\mathord{\vpbdresRed}\eqdef\pars{\Union_{b\in\naturals}\mathord{\splitVariant\pbdresRed b}}\subseteq\resourceVectors\times\resourceVectors\]
however, we will obtain a strongly confluent reduction relation, under the assumption
that $\rigS$ has the following additive splitting property:\footnote{
	The additive splitting property was previously used by Carraro, Ehrhard and
	Salibra \cite{ces:multiplicities,carraro:phd} in their study of linear logic
	exponentials with infinite multiplicities. There is no clear connection
	between that work and our present contributions, though.
}
\begin{definition}
\label{definition:splitting}
We say $\rigS$ has the \emph{additive splitting property} if:
whenever $a_1+a_2=b_1+b_2\in\rigS$, there exists 
$c_{1,1},c_{1,2},c_{2,1},c_{2,2}\in\rigS$ such that
$a_i=c_{i,1}+c_{i,2}$ and $b_j=c_{1,j}+c_{2,j}$ for $i,j\in\set{1,2}$.
\end{definition}
This property is satisfied by any ring, but also by 
the usual semirings of non-negative numbers ($\naturals$, 
$\rationals^+$, \emph{etc.}) as well as booleans.
We will in fact rely on the following generalization of the property to finite
families of finite sums of any size:
\begin{lemma}
	\label{lemma:additive:splitting}
	Assume $\rigS$ has the additive splitting property.
	Let $a\in\rigS$, $J_1,\dotsc,J_n$ be finite sets and,
	for all $i\in\set{1,\dotsc,n}$,
	let $\pars{b_{i,j}}_{j\in J_i}\in\rigS^{J_i}$ be a family 
	such that $a=\sum_{j\in J_i} b_{i,j}$.
	Write $J=J_1\times\cdots\times J_n$ and, 
	for all $i\in\set{1,\dotsc,n}$, write $J'_i=J_1\times\cdots\times J_{i-1}\times J_{i+1}\times\cdots\times J_n$.
	Whenever $\vec \jmath'=(j_1,\dotsc,j_{i-1},j_{i+1},\dotsc,j_n)\in J'_i$
	and $j_i\in J_i$, write $\vec \jmath' \cdot_i j_i = (j_1,\dotsc,j_n)\in J$.
	Then there exists a family $\pars{c_{\vec\jmath}}\in\rigS^{J}$ 
	such that, for all $i\in\set{1,\dotsc,n}$ and all $j\in J_i$,
	$b_{i,j}=\sum_{\vec\jmath '\in J'_i}c_{\vec\jmath'\cdot_i j}$.
\end{lemma}
\begin{proof}
	By induction on $n$, and then on $\card{J_n}$ for $n>0$,
	using the binary additive splitting property to enable the induction.
\end{proof}

\begin{lemma}
	\label{lemma:vpbdresRed:diamond}
	Assume $\rigS$ has the additive splitting property
	and fix a $\presRed$-reduction structure $\fE$.
	For all finite set $K$ and all reductions $ε\vpbdresRed[\fE] ε'_k$ for $k\in K$,
	there is $ε''$ such that $ε'_k\vpbdresRed[\fE] ε''$ for all $k\in K$.
\end{lemma}
\begin{proof}
	For all $k\in K$, there are $b_k\in\naturals$,
	a resource support $\calE_k\in\fE$,
	a set $I_k$ of indices,
	a family $\pars{a_{k,i}}_{i\in I_k}$ of scalars,
	and summable families $\pars{e_{k,i}}_{i\in I_k}\in\calE_k^{I_k}$ and 
	$\pars{ε'_{k,i}}_{i\in I_k}\in\finiteSums{\resourceExpressions}^{I_k}$ such that
	$ε=\sum_{i\in I_k} a_{k,i}\sm e_{k,i}$,
	$ε'_k=\sum_{i\in I_k} a_{k,i}\sm ε'_{k,i}$
	and $e_{k,i}\pbdresRed{b_k} ε'_{k,i}$ for all $i\in I_k$.

	Write $\calE=\set{e_{k,i}\st k\in K,\ i\in I_k}$: since 
	$\calE\subseteq \Union_{k\in K}\calE_k$ and $\fE$ is a resource structure,
	we have $\calE\in\fE$. Write $\calE'=\Union\set{\support{ε'_{k,i}}\st k\in K,\ i\in I_k}$:
	since $\fE$ is a reduction structure, we also have $\calE'\in\fE$.

	Now fix $e\in\resourceExpressions$ and write $a=ε_e$.
	For all $k\in K$, the set $I_{e,k}=\set{i_k\in I_k\st e_{k,i_k}=e}$
	is finite, and then $\sum_{i_k\in I_{e,k}} a_{k,i_k}=a$.
	Write $I_e=\prod_{k\in K} I_{e,k}$ and, for all $k\in K$,
	$K'_k= K\setminus \set k$ and
	$I'_{e,k}=\prod_{l\in K'_k} I_{e,l}$.
	If $\vec\imath=\pars{i_l}_{l\in K'_k}\in I'_{e,k}$
	and $i_k\in I_{e,k}$, write $\vec\imath\cdot_k i_k=\pars{i_k}_{k\in K}\in I_e$.
	By Lemma \ref{lemma:additive:splitting},
	we obtain a family of scalars $\pars{a'_{e,\vec\imath}}_{\vec\imath\in I_e}$
	such that, for all $k\in K$ and all $i_k\in I_{e,k}$,
	$a_{k,i_k}=\sum_{\vec\imath\in I'_{e,k}} a'_{e,\vec\imath\cdot_k i_k}$.
	Moreover, $a=\sum_{\vec\imath\in I_e} a'_{e,\vec\imath}$.

	Since each $I_e$ is finite, the family
	$\pars{e}_{e\in\resourceExpressions,\vec\imath\in I_e}$ is summable.
	Moreover, if we fix $k\in K$ and $i_k\in I_k$, 
	there are finitely many $e\in\resourceExpressions$ and $\vec\imath\in I'_{e,k}$
	such that $\vec\imath\cdot_k i_k\in I_e$:
	indeed in this case $e=e_{k,i_k}$.
	Since $\pars{ε'_{k,i_k}}_{i_k\in I_k}$ is summable too,
	it follows that $\pars{ε'_{k,i_k}}_{e\in\resourceExpressions,\vec\imath\in I_e}$ is summable.
	By associativity, we obtain
	\[
		\sum_{\substack{e\in\resourceExpressions\\\vec\imath \in I_e}} a'_{e,\vec\imath}\sm e
		= \sum_{e\in\resourceExpressions} \pars{\sum_{\vec\imath \in I_e} a'_{e,\vec\imath}} e
		= ε
	\]
	and
	\[
		\sum_{\substack{e\in\resourceExpressions\\\vec\imath \in I_e}} a'_{e,\vec\imath}\sm ε'_{k,i_k}
		= \sum_{i_k\in I_k} \pars{\sum_{\vec\imath\in I'_{e_{k,i_k},k}} a'_{e_{k,i_k},\vec\imath\cdot_ki_k}} ε'_{k,i_k}
		= ε'_k
	\]
	for all $k\in K$.

	Write $b=\sum_{k\in K} b_k$.
	For all $e\in\resourceExpressions$ and all $\vec\imath=\pars{i_k}_{k\in K} \in I_e$,
	we have $e\pbdresRed{b_k} ε'_{k,i_k}$ for all $k\in K$ hence
	Lemma \ref{lemma:pbdresRed:confluent} gives
	$ε''_{e,\vec\imath}\in\finiteSums{\resourceExpressions}$
	such that $ε'_{k,i_k}\pbdresRed {2^{\scriptstyle b_k}b} ε''_{e,\vec\imath}$ for all $k\in K$. 
	Moreover, for all $k\in K$ and $e''\in\resourceExpressions$,
	if $e''\in\support{ε''_{e,\vec\imath}}$ then there is $e'\in\support{ε'_{k,i_k}}$
	such that $e'\pbdOneStepGenerates {2^{\scriptstyle b_k}b} e''$, and then $e\pbdOneStepGenerates {b_k} e'$:
	it follows that $\size{e}\le 4^{b_k+2^{b_k}b}\size{e''}$ and $\fv{e}=\fv{e''}$.
	Since each $I_e$ is finite, there are finitely many pairs
	$\pars{e,\vec\imath}\in\sum_{e\in\resourceExpressions} I_e$
	such that $e''\in\support{ε''_{e,\vec\imath}}$.
	Hence the family $\pars{ε''_{e,\vec\imath}}_{e\in\resourceExpressions,\vec\imath\in I_e}$
	is summable. Recall moreover that $ε'_{k,i_k}\in\finiteSums{\calE'}$ for all $k\in K$ and $i_k\in I_k$:
	we obtain
	\[ ε'_k\splitVariant[\calE']\pbdresRed {2^{\scriptstyle b_k}b}
	\sum_{\substack{e\in\resourceExpressions\\\vec\imath\in I_e}}a'_{e,\vec\imath}\sm ε''_{e,\vec\imath}\]
	for all $k\in K$, which concludes the proof.
\end{proof}

\subsection{Bounded depth of substitution}

\label{subsection:pbsresRed}

In the previous subsection, we relied on the additive splitting property
to establish the confluence of $\vpbdresRed$: this is because there 
is no maximal way to $\splitVariant\pbdresRed b$-reduce a resource vector,
hence we must track precisely the different redexes that are fired 
in each reduction of a critical pair.

We can get rid of this hypothesis by considering a more uniform bound
on reductions. A first intuition would be to bound the depth at 
which redexes are fired, but as with $\pbresRed mk$ this boundedness 
condition is not preserved in residuals: rather, we have to bound 
the depth at which variables are substituted.
First recall from Definition \ref{definition:occ} that $\maxoccdepth
xs=\max\occdepth xs$ is the maximum depth of an occurrence of $x$ in
$s$. Then:

\begin{definition}
	We define a family of relations $\pars{\pbsresRed d}_{d\in{\naturals}}$
	inductively as follows:
	\begin{itemize}
		\item $e\pbsresRed 0 e$ for all $e\in\resourceExpressions$;
		\item $x\pbsresRed {d+1} x$ for all $x\in\variables$;
		\item $\labs x s\pbsresRed {d+1} \labs x{σ'}$ if $s\pbsresRed {d} σ'$;
		\item $\rappl s{\ms t}\pbsresRed {d+1} \rappl{σ'}{\ms {τ'}}$
			if $s\pbsresRed {d+1} σ'$ and $\ms t\pbsresRed {d} \ms{τ'}$;
		\item $\mset{s_1,\dotsc,s_r}\pbsresRed {d+1} \mset{σ'_1,\dotsc,σ'_r}$
			if $s_i\pbsresRed {d+1} σ'_i$ for all $i\in\set{1,\dotsc,r}$;
		\item $\rappl{\labs xs}{\ms t}\pbsresRed {d+1} \lsubst{σ'}x{\ms{τ'}}$ 
			if $\maxoccdepth xs\le d$,
			$s\pbsresRed {d} σ'$ and $\ms t\pbsresRed {d} \ms{τ'}$.
	\end{itemize}
\end{definition}

It should be clear that if $e\pbsresRed d ε'$
then $e\pbdresRed d ε'$, and moreover $e\pbsresRed {d'}ε'$
for all $d'\ge d$. We also have $e\pbsresRed {\height e} ε'$
as soon as $e\presRed ε'$.

\begin{definition}
	For all $e\in\resourceExpressions$ we define 
	the \emph{full parallel reduct $\fpbsReduct de$ at substitution depth $d$} of $e$ by induction on the pair $(d,e)$ as follows:
	\begin{align*}
		\fpbsReduct{0}{e}                        &\eqdef e\\
		\fpbsReduct{d+1}{x}                      &\eqdef x\\
		\fpbsReduct{d+1}{\labs xs}               &\eqdef \labs x{\fpbsReduct{d}{s}}\\
		\fpbsReduct{d+1}{\rappl{\labs xs}{\ms t}}&\eqdef
				\lsubst{\fpbsReduct{d}{s}}x{\fpbsReduct{d}{\ms t}}
			&&\text{(if $\maxoccdepth xs\le d$)}\\
		\fpbsReduct{d+1}{\rappl{s}{\ms t}}       &\eqdef \rappl{\fpbsReduct{d+1}{s}}{\fpbsReduct{d}{\ms t}}
			&&\text{(in the other cases)}\\
		\fpbsReduct{d+1}{\mset{s_1,\dotsc,s_n}}  &\eqdef \mset{\fpbsReduct{d+1}{s_1},\dotsc,\fpbsReduct{d+1}{s_n}}
	\end{align*}
	Then if $ε=\sum_{i=1}^n e_i\in\finiteResourceSums$,
	we set $\fpbsReduct{d}{ε}\eqdef\sum_{i=1}^n\fpbsReduct{d}{e_i}$.
\end{definition}

\begin{lemma}
	\label{lemma:fpbsReduct}
	For all $e\in\resourceExpressions$, $e\pbsresRed d \fpbsReduct d e$.
\end{lemma}
\begin{proof}
	By a straightforward induction on $d$ then on $e$.
\end{proof}

It follows that $e\pbdresRed d\fpbsReduct de$, hence if $e'\in\support{\fpbsReduct de}$ then $\size{e}\le 4^d\size{e'}$.
In particular $\fpbsReductSym d$ defines a linear-continuous function on $\resourceVectors$.

\begin{lemma}
	\label{lemma:pbsresRed:lsubst}
	If $e\pbsresRed {d_0} ε'$, 
	$\ms u\pbsresRed {d_1}{\ms{υ}'}$
	and $d\ge\max\pars{\set{d_0}\union\set{d_x+d_1-1\st d_x\in\occdepth xe}}$ then 
	$\lsubst ex{\ms u}\pbsresRed{d}\lsubst{ε'}x{\ms{υ}'}$.
\end{lemma}
\begin{proof}
	Write $n=\card{\ms u}$, $\ms u=\mset{u_1,\dotsc,u_n}$
	and $\ms{υ}'=\mset{υ'_1,\dotsc,υ'_n}$ so that 
	$u_i\pbsresRed{d_1} υ'_i$ for $i\in\set{1,\dotsc, n}$.

	The proof is by induction on the reduction $e\pbsresRed {d_0}{ε'}$.
	We treat the cases $d_0=0$ and $d_0>0$ uniformly by a further induction on $e$,
	setting $d'_0=\max\set{0,d_0-1}$.

	If $d_0=d'_0+1$, $e=\rappl{\labs ys}{\ms t}$
	and $ε'=\lsubst{σ'}y{\ms{τ}'}$
	with $y\not\in\set x\union\fv{\ms t}\union\fv{\ms u}$,
	$\maxoccdepth ys\le {d'_0}$, 
	$s\pbsresRed {d'_0} σ'$ and 
	$\ms t\pbsresRed {d'_0}\ms{τ}'$,
	then we have 
	\[ 
		\lsubst ex{\ms u}=\sum_{(I,J)\text{ partition of }\set{1,\dotsc,n}}
		\rappl{\labs y[]{\lsubst{s}x{\ms{u}_I}}}{\lsubst{\ms t}x{\ms u_J}}
	\]
	and 
	\[
		\lsubst {ε'}x{\ms{υ}'}=\sum_{(I,J)\text{ partition of }\set{1,\dotsc,n}}
		\lsubst{\pars{\lsubst{s}x{\ms{υ}'_I}}}y{\pars{\lsubst{\ms t}x{\ms {υ}'_J}}}.
	\]
	Observe that $d>0$ and $d-1\ge\max\set{d'_0}\union\set{d'_x+d_1-1\st d'_x\in\occdepth x{s}\union\occdepth x{\ms t}}$.
	By induction hypothesis, 
	we obtain $\lsubst{s}x{\ms{u}_I}\pbsresRed{d-1}\lsubst{σ'}x{\ms{υ}'_I}$
	and $\lsubst{\ms t}x{\ms{u}_J}\pbsresRed{d-1}\lsubst{\ms{τ}'}x{\ms{υ}'_J}$,
	and we conclude since $\maxoccdepth y{\lsubst{s}x{\ms{u}_I}}=\maxoccdepth ys\le d'_0\le d-1$.

	If $e=y=ε'$, with $y\not=x$, then $\lsubst ex{\ms u}=\lsubst{ε'}x{\ms{υ}'}=y$
	and we conclude directly by the definition of $\pbsresRed d$.

	If $e=x=ε'$, then $\occdepth xe=\set 1$ hence $d\ge d_1$ and we conclude since 
	$\lsubst ex{\ms u}=\ms u$, $\lsubst{ε'}x{\ms{υ}'}=\ms{υ}'$
	and $\ms u\pbsresRed {d_1}{\ms{υ}'}$.

	If $e=\labs ys$ and $ε'=\labs y{σ'}$ with $y\not\in\set x\union\fv{\ms u}$
	and $s\pbsresRed {d'_0} σ'$, then
	write $d'=\max\set{d'_0}\union\set{d'_x+d_1-1\st d'_x\in\occdepth x{s}}$.
	By induction hypothesis, 
	we obtain $\lsubst{s}x{\ms{u}}\pbsresRed{d'}\lsubst{σ'}x{\ms{υ}'}$.
	Observe that either $d=d'+1$ or $d=d'=0$ (in that latter case,
	$\lsubst{s}x{\ms{u}}=\lsubst{σ'}x{\ms{υ}'}$),
	and then we conclude since
	$\lsubst ex{\ms u}=\labs y[]{\lsubst{s}x{\ms{u}}}$
	and
	$\lsubst {ε'}x{\ms{υ}'}=\labs y[]{\lsubst{σ'}x{\ms{υ}'}}$.

	If $e=\rappl{s}{\ms t}$
	and $ε'=\rappl{σ'}{\ms{τ}'}$, with
	$s\pbsresRed {d_0} σ'$ and
	$\ms t\pbsresRed {d'_0}\ms{τ}'$,
	then we have
	\[ 
		\lsubst ex{\ms u}=\sum_{(I,J)\text{ partition of }\set{1,\dotsc,n}}
		\rappl{\lsubst{s}x{\ms{u}_I}}{\lsubst{\ms t}x{\ms u_J}}
	\]
	and 
	\[
		\lsubst {ε'}x{\ms{υ}'}=\sum_{(I,J)\text{ partition of }\set{1,\dotsc,n}}
		\rappl{\lsubst{s}x{\ms{υ}'_I}}{\pars{\lsubst{\ms t}x{\ms {υ}'_J}}}.
	\]
	Write $d'=\max\set{d'_0}\union\set{d'_x+d_1-1\st d'_x\in\occdepth x{\ms t}}$.
	By induction hypothesis, 
	we obtain $\lsubst{s}x{\ms{u}_I}\pbsresRed{d}\lsubst{σ'}x{\ms{υ}'_I}$
	and $\lsubst{\ms t}x{\ms{u}_J}\pbsresRed{d'}\lsubst{\ms{τ}'}x{\ms{υ}'_J}$.
	Then we conclude observing that $d=d'+1$ or $d=d'=0$ (in that latter case,
	$\lsubst{s}x{\ms{u}_I}=\lsubst{σ'}x{\ms{υ}'_I}$ and
	$\lsubst{\ms t}x{\ms{u}_J}=\lsubst{\ms{τ}'}x{\ms{u}'_J}$).

	If $e=\mset{s_1,\dotsc,s_k}$
	and $ε'=\mset{σ'_1,\dotsc,σ'_k}$, with
	$s_i\pbsresRed {d_0} σ'_i$ for $i\in\set{1,\dotsc,k}$,
	then we have
	\[
		\lsubst{e}x{\ms{u}}=
		\sum_{
			(I_1,\dotsc,I_k)\text{ partition of }\set{1,\dotsc,n}
		}
		\mset{
			\lsubst {s_1}x{\ms u_{I_1}}
			\dotsc,
			\lsubst {s_k}x{\ms u_{I_k}}
		}
	\]
	and 
	\[
		\lsubst{ε'}x{\ms{υ}'}=
		\sum_{
			(I_1,\dotsc,I_k)\text{ partition of }\set{1,\dotsc,n}
		}
		\mset{
			\lsubst {σ'_1}x{\ms {υ'}_{I_1}}
			\dotsc,
			\lsubst {σ'_k}x{\ms {υ'}_{I_k}}
		}.
	\]
	By induction hypothesis, we obtain
	\[
	\mset{
			\lsubst {s_1}x{\ms u_{I_1}}
			\dotsc,
			\lsubst {s_k}x{\ms u_{I_k}}
		}
		\pbsresRed d
	\mset{
			\lsubst {σ'_1}x{\ms {υ'}_{I_1}}
			\dotsc,
			\lsubst {σ'_k}x{\ms {υ'}_{I_k}}
		}
	\] for all partition $(I_1,\dotsc,I_k)$ of $\set{1,\dotsc,n}$ and we conclude. 
\end{proof}

\begin{lemma}
	\label{lemma:bound:maxoccdepth:pbsresRed}
	If $e\pbsresRed d ε'$ and $e'\in\support{ε'}$ 
	then $\maxoccdepth x{e'}\le 2^{d}\max\set{d,\maxoccdepth x e}$.
\end{lemma}
\begin{proof}
	By induction on the reduction $e\pbsresRed d ε'$.

	If $d=0$, then $e'=ε'=e$ and the result is trivial.
	For the other inductive cases, write $d=d'+1$.

	If $e=\rappl{\labs ys}{\ms t}$
	and $ε'=\lsubst{σ'}y{\ms{τ}'}$ with 
	$\maxoccdepth ys\le d'$,
	$s\pbsresRed {d'} σ'$ and $\ms t\pbsresRed{d'} \ms{τ}'$,
	choosing $y\not\in\set x\union\fv {\ms t}$,
	then $e'\in\support[\mbig]{\lsubst{s'}y{\ms {t}'}}$ with 
	$s'\in\support{σ'}$ and $\ms{t}'\in\support{\ms{τ}'}$.
	By induction hypothesis, $\maxoccdepth z{s'}\le 2^{d'}\max\set{d',\maxoccdepth z{s}}$
	and $\maxoccdepth[\mnorm]z{\ms{t}'}\le 2^{d'}\max\set{d',\maxoccdepth z{\ms t}}$
	for any $z\in\variables$.
	By Lemma \ref{lemma:lsubst:size}, 
	\begin{align*}
	\maxoccdepth x{e'}
	&\le \max\pars{\occdepth x{s'}
	\union\set{d'_y+d'_x-1\st d'_y\in\occdepth y{s'},\ d'_x\in\occdepth[\mnorm] x{\ms t'}}}
	\\
	&\le \max\set{2^{d'}\max\set{d',\maxoccdepth x{s}},2^{d'}\max\set{d',\maxoccdepth y{s}}+2^{d'}\max\set{d',\maxoccdepth x{\ms t}}}
	\\
	&\le 2^{d'+1}\max\set{d',\maxoccdepth x{s},\maxoccdepth y{s},\maxoccdepth x{\ms t}}
	\\
	&\le 2^{d}\max\set{d,\maxoccdepth x{e}}.
	\end{align*}

	If $e=\labs ys$ and $ε'=\labs y{σ'}$ with 
	$s\pbsresRed {d'} σ'$,
	choosing $y\not=x$,
	then $e'=\labs x{s'}$ with 
	$s'\in\support{σ'}$.
	By induction hypothesis, $\maxoccdepth x{s'}\le 2^{d'}\max\set{d',\maxoccdepth x{s}}$.
	Then $\maxoccdepth x{e'}\le \maxoccdepth x{s'}+1\le 2^{d'}\max\set{d',\maxoccdepth x{s}}+1\le 
	2^d\max\set{d,\maxoccdepth xs}\le 2^d\max\set{d,\maxoccdepth xe}$.

	If $e=\rappl s{\ms t}$ and $ε'=\rappl{σ'}{\ms{τ}'}$ with
	$s\pbsresRed {d} σ'$ and $\ms t\pbsresRed{d'} \ms{τ}'$,
	then $e'=\rappl {s'}{\ms t'}$ with 
	$s'\in\support{σ'}$ and $\ms{t}'\in\support{\ms{τ}'}$.
	By induction hypothesis, $\maxoccdepth x{s'}\le 2^{d}\max\set{d,\maxoccdepth x{s}}$
	and $\maxoccdepth[\mnorm]x{\ms{t}'}\le 2^{d'}\max\set{d',\maxoccdepth x{\ms t}}$.
	Then:
	\begin{align*}
		\maxoccdepth x{e'}&\le \max\set{\maxoccdepth x{s'},\maxoccdepth[\mbig] x{\ms t'}+1}\\
		&\le \max\set{2^{d}\max\set{d,\maxoccdepth x{s}},2^{d'}\max\set{d',\maxoccdepth x{\ms t}}+1}\\
		&\le 2^d\max\set{d,\maxoccdepth xs,\maxoccdepth x{\ms t}}\\
		&\le 2^d\max\set{d,\maxoccdepth xe}.
	\end{align*}

	If $e=\mset{s_1,\dotsc,s_k}$
	and $ε'=\mset{σ'_1,\dotsc,σ'_k}$, with
	$s_i\pbsresRed {d} σ'_i$ for all $i\in\set{1,\dotsc,k}$,
	then $e'=\mset{s'_1,\dotsc,s'_k}$
	with $s'_i\in\support{σ'_i}$ for all $i\in\set{1,\dotsc,k}$.
	By induction hypothesis, for all $i\in\set{1,\dotsc,k}$, 
	$\maxoccdepth x{s'_i}\le 2^{d}\max\set{d,\maxoccdepth x{s_i}}$,
	hence
	\begin{align*}
		\maxoccdepth x{e'}&=\max\set{\maxoccdepth x{s'_1},\dotsc,\maxoccdepth x{s'_k}}\\
		&\le 2^{d}\max\set{d,\maxoccdepth x{s_1},\dotsc,\maxoccdepth x{s_k}}\\
		&= 2^{d}\max\set{d,\maxoccdepth x{e}}.&&\hfill\qedhere
	\end{align*}
\end{proof}

\begin{lemma}
	If $e\pbsresRed d ε'$ then 
	$ε'\pbsresRed {2^dd}\fpbsReduct de$.
\end{lemma}
\begin{proof}
	By induction on the reduction $e\pbsresRed d ε'$.

	If $d=0$, then $ε'=e$ and the result follows from Lemma \ref{lemma:fpbsReduct}.
	For the other inductive cases, set $d=d'+1$.

	If $e=\rappl{\labs xs}{\ms t}$
	and $ε'=\lsubst{σ'}x{\ms{τ}'}$
	with $\maxoccdepth xs\le {d'}$, 
	$s\pbsresRed {d'} σ'$ and
	$\ms t\pbsresRed {d'}\ms{τ}'$
	then by induction hypothesis, 
	we have $σ'\pbsresRed{2^{d'}d'} \fpbsReduct {d'}s$
	and  $\ms{τ}'\pbsresRed{2^{d'}d'} \fpbsReduct {d'}{\ms t}$.
	By the previous lemma,
	we moreover have $\maxoccdepth x{σ'}\le 2^{d'}\max\set{d',\maxoccdepth xs}=2^{d'}d'$.
	It follows that $2^{d}d\ge 2^{d'}d'$ and $2^dd\ge\maxoccdepth x{σ'}+2^{d'}d'-1$
	hence we can apply Lemma \ref{lemma:pbsresRed:lsubst}
	to obtain
	$ε'\pbsresRed{2^{d}{d}}\lsubst{\fpbsReduct{d'}s}x{\fpbsReduct{d'}{\ms t}}=\fpbsReduct{d}{e}$.

	If $e=\labs ys$ and $ε'=\labs y{σ'}$ with 
	$s\pbsresRed {d'} σ'$,
	then by induction hypothesis, 
	$σ'\pbsresRed{2^{d'}d'}\fpbsReduct{d'}s$,
	hence
	$ε'\pbsresRed{2^{d'}d'+1}\labs x{\fpbsReduct{d'}s}=\fpbsReduct de$
	and we conclude since $2^{d'}d'+1\le 2^dd$.

	If $e=\rappl {s}{\ms t}$ and $ε'=\rappl{σ'}{\ms{τ}'}$ with
	$s\pbsresRed {d} σ'$ and $\ms t\pbsresRed{d'} \ms{τ}'$,
	there are two subcases:
	\begin{itemize}
		\item 
			If moreover $s=\labs xu$ and $\maxoccdepth xu\le d'$ then 
			$σ'=\labs x{υ'}$ with 
			$u\pbsresRed {d'} υ'$.
			Then by induction hypothesis, 
			$υ'\pbsresRed{2^{d'}d'}\fpbsReduct{d'}u$,
			and 
			$\ms{τ}'\pbsresRed{2^{d'}d'}\fpbsReduct{d'}{\ms t}$.
			By the previous lemma,
			we moreover have $\maxoccdepth x{υ'}\le 2^{d'}\max\set{d',\maxoccdepth xu}=2^{d'}d'$,
			hence
			$ε'=\rappl{\labs x{υ'}}{\ms{τ}'}
			\pbsresRed{2^{d'}d'+1}\lsubst{\fpbsReduct{d'}u}x{\fpbsReduct{d'}{\ms t}}=\fpbsReduct de$,
			and we conclude since $2^{d'}d'+1\le 2^dd$.
		\item
			Otherwise $s$ is not an abstraction or $s=\labs xu$ with $\maxoccdepth xu>d'$.
			By induction hypothesis, 
			$σ'\pbsresRed{2^{d}d}\fpbsReduct{d}{σ'}$,
			and 
			$\ms{τ}'\pbsresRed{2^{d'}d'}\fpbsReduct{d'}{\ms t}$.
			Since $2^{d'}d'<2^dd$, we obtain
			$\ms{τ}'\pbsresRed{2^{d}d-1}\fpbsReduct{d'}{\ms t}$
			and then 
			$ε'\pbsresRed{2^{d}d}\rappl{\fpbsReduct{d}s}{\fpbsReduct{d'}{\ms t}}=\fpbsReduct de$.
	\end{itemize}

	If $e=\mset{s_1,\dotsc,s_k}$
	and $ε'=\mset{σ'_1,\dotsc,σ'_k}$, with
	$s_i\pbsresRed {d} σ'_i$ for all $i\in\set{1,\dotsc,k}$,
	then by induction hypothesis, for all $i\in\set{1,\dotsc,k}$, 
	$σ'_i\pbsresRed {2^dd}\fpbsReduct d{s_i}$ and we conclude directly.
\end{proof}

\begin{lemma}
	\label{lemma:vpbsresRed:confluent}
	For all $\presRed$-reduction structure $\fE$, 
	if $ε\splitVariant[\fE]\pbsresRed d ε'$
	then $ε'\splitVariant[\fE]\pbsresRed{2^dd}\fpbsReduct d{ε'}$.
\end{lemma}
\begin{proof}
	Assume there is $\calE\in\fE$, summable families
	$\pars{e_i}_{i\in I}\in\calE^I$ and 
	$\pars{ε'_i}_{i\in I}\finiteResourceSums^I$,
	and a family of scalars
	$\pars{a_i}_{i\in I}$
	such that  $ε=\sum_{i\in I} a_i\sm e_i$,
	$ε'=\sum_{i\in I} a_i\sm ε'_i$
	and $e_i\pbsresRed dε'_i$ for all $i\in I$.
	Write $\calE'=\Union_{i\in I}\support{ε'_i}$:
	since $\fE$ is a reduction structure,
	we obtain $\calE'\in\fE$.
	The family $\pars{\fpbsReduct d{e_i}}_{i\in I}$ is summable,
	and by the previous lemma, $ε'_i\pbsresRed {2^dd}\fpbsReduct d {e_i}$
	for all $i\in I$.
	We conclude 
	that $ε'\splitVariant\pbsresRed{2^dd}\sum_{i\in I} a_i\sm\fpbsReduct d{e_i}=\fpbsReduct d{ε}$.
\end{proof}

Similarly to $\vpbdresRed$, we set
\[\mathord{\vpbsresRed}\eqdef\Union_{d\in\naturals}\mathord{\splitVariant\pbsresRed d}\]
and we obtain:

\begin{corollary}
	\label{corollary:vpbsresRed:confluent}
	For all $\presRed$-reduction structure $\fE$
	and all $ε,ε'_1,\dotsc,ε'_n\in\resourceVectors$ such that 
	$ε\vpbsresRed[\fE] ε'_i$ for $i\in\set{1,\dotsc,n}$, 
	there exists $d\in\naturals$ such that
	$ε'_i\vpbsresRed[\fE]\fpbsReduct d{ε}$  for $i\in\set{1,\dotsc,n}$.
\end{corollary}

\subsection{Parallel reduction of resource vectors of bounded height}

\label{subsection:bounded}

Recall that we have $e\presRed ε'$ iff $e\pbresRed {\height e}{\height e} ε'$
iff $e\pbdresRed {\height e} ε'$ iff $e\pbsresRed {\height e} ε'$.

\begin{definition}
	We say a resource vector $ε\in\resourceVectors$ is \emph{bounded} if $\set{\height e\st e\in\support{ε}}$ is finite.
	We then write $\height{ε}=\max\set{\height e\st e\in\support{ε}}$.
\end{definition}

If $\calE\subseteq\resourceExpressions$, we also write
$\height{\calE}\eqdef\set{\height e\st e\in\calE}$
and then \[\boundedStructure\eqdef
\set{\calE\subseteq\resourceExpressions\st\height{\calE}\text{ and }\fv{\calE}\text{ are finite}}\]
which is a resource structure (see Definition \ref{definition:resource:structure}).
Indeed, $\boundedStructure\subseteq\resourceStructure$,
and if we write
\[\resourceExpressions[h,V]\eqdef
\set{e\in\resourceExpressions\st \height{e}\le h\text{ and }\fv e\subseteq V}\]
for all $h\in\naturals$ and all $V\subseteq\variables$, 
we have
$\boundedStructure=\bidual{\set{\resourceExpressions[h,V]\st h\in\naturals\text{ and }V\in\finiteSubsets{\variables}}}$:
this is a consequence of a generic \emph{transport lemma} \cite{tv:transport}.
The semimodule of bounded resource vectors is then $\boundedResourceVectors$.

\begin{lemma}
	\label{lemma:fullReduct:summable}
	For all $h\in\naturals$ and $V\in\finiteSubsets{\variables}$,
	$\pars{\fullReduct e}_{e\in\resourceExpressions[h,V]}$ is summable.
	Moreover, for all $ε\in\boundedResourceVectors$, we have 
	$\support{ε}\subseteq\resourceExpressions[\height{ε},\fv{ε}]$ and then,
	setting $\fullReduct{ε}\eqdef\sum_{e\in\support{ε}} ε_e\sm\fullReduct e$,
	we obtain $ε\vpresRed[\support{ε}]\fullReduct{ε}$.
\end{lemma}
\begin{proof}
	Follows from Lemmas \ref{lemma:fpbsReduct} and \ref{lemma:pbdresRed:size}
	using the fact that, if $\height e\le h$ then $\fullReduct e=\fpbsReduct he$.
\end{proof}

If $\rigS$ is zerosumfree, we have:
$ε\vpresRed ε'$ iff $ε\splitVariant\pbsresRed {\height{ε}} ε'$ as soon as $ε$ is bounded.
More generally, without any assumption on $\rigS$,
we have 
$ε\vpresRed[{\resourceExpressions[h,V]}] ε'$
iff $ε\splitVariant[{\resourceExpressions[h,V]}]\pbsresRed h ε'$.
We can moreover show that bounded vectors are stable under $\vpresRed[\boundedStructure]$:

\begin{lemma}
	\label{lemma:pOneStepGenerates:height}
	If $e\pOneStepGenerates e'$ then $\height{e'}\le 2^{\height e}{\height e}$.
\end{lemma}
\begin{proof}
	The proof is by induction on the reduction $e\presRed ε'$ such that $e'\in\support{ε'}$,
	and is very similar to that of Lemma \ref{lemma:bound:maxoccdepth:pbsresRed}.
	We detail only the base case.

	If $e=\rappl{\labs xs}{\ms t}$
	and $ε'=\lsubst{σ'}x{\ms{τ}'}$ with 
	$s\presRed σ'$ and $\ms t\presRed \ms{τ}'$.
	Then $e'\in\support{\lsubst{s'}x{\ms {t'}}}$ with 
	$s'\in\support{σ'}$ and $\ms{t'}\in\support{\ms{τ}'}$.
	By induction, $\height{s'}\le 2^{\height{s}}\height{s}$
	and $\height{\ms{t'}}\le 2^{\height{\ms t}}\height{\ms t}$.
	By Lemma~\ref{lemma:npdiff:size}, 
	\begin{align*}
	     \height{e'}
	&\le \height{s'}+\height{\ms{t'}}
  \\
	&\le 2^{\height{s}}\height{s}+2^{\height{\ms t}}\height{\ms t}
  \\
	&\le 2\times 2^{\max\set{\height{s},\height{\ms t}}}\max\set{\height{s},\height{\ms t}}
  \\
	&<   2^{\max\set{\height{s},\height{\ms t}}+1}\pars{\max\set{\height{s},\height{\ms t}}+1}
  \\
	&=   2^{\height{e}}\height{e}.&&\hfill\qedhere
	\end{align*}
\end{proof}

It follows that $\boundedStructure$ is a $\presRed$-reduction structure:
since $\vpresRed[\boundedStructure]$ coincides with $\vpbsresRed[\boundedStructure]$,
Corollary \ref{corollary:vpbsresRed:confluent} entails that
$\vpresRed[\boundedStructure]$ is strongly confluent.
We can even refine this result following Lemma \ref{lemma:fullReduct:summable}.
First, let us call \emph{bounded reduction structure} any $\presRed$-reduction
structure $\fE$ such that $\fE\subseteq\boundedStructure$.
Then Lemma \ref{lemma:vpbsresRed:confluent} entails:
\begin{corollary}
	\label{corollary:splitVariant:presRed:fullReduct}
	For all bounded reduction structure $\fE$, 
	and all reduction $ε\vpresRed[\fE] ε'$, 
	$ε'\vpresRed[\fE]\fullReduct{ε}$.
\end{corollary}

It should moreover be clear that $\TaylorExp{M}$ is bounded for all $M\in\algebraicTerms$.
In the next section, we show that $\vpresRed[\boundedStructure]$ allows to simulate 
parallel β-reduction via Taylor expansion.\footnote{
	\label{footnote:direct}
	Observe that it is possible to establish Corollary
	\ref{corollary:splitVariant:presRed:fullReduct} quite directly, following the
	proof of Lemma \ref{lemma:vpbsresRed:confluent}, and using only Lemma
	\ref{lemma:pOneStepGenerates:height} and a variant of Lemma
	\ref{lemma:pbdresRed:size} (replacing $b$ with $\height e$).
	This is the path adopted in the extended abstract
	\cite{vaux:taylor-beta} presented at \emph{CSL 2017}.
}

\section{Simulating β-reduction under Taylor expansion}

\label{section:simulation:beta}

From now on, for all $M,N\in\algebraicTerms$, we write 
$M\vpresRed N$ if $\TaylorExp M\vpresRed\TaylorExp N$.
More generally, for all $M\in\algebraicTerms$ and all $σ\in\resourceVectors$,
we write $M\vpresRed σ$ (resp.\ $σ\vpresRed M$) if $\TaylorExp M\vpresRed σ$
(resp.\ $σ\vpresRed\TaylorExp M$).
We will show in Subsection \ref{subsection:simulation} that
$M\vpresRed N$ as soon as $M\pbetaRed N$
where $\pbetaRed$ is the parallel β-reduction defined as follows:

\begin{definition}
	\label{definition:pbetaRed}
	We define \emph{parallel β-reduction} on algebraic terms
	$\mathord{\pbetaRed}\subseteq\algebraicTerms\times\algebraicTerms$
	by the following inductive rules:
	\begin{itemize}
		\item $x\pbetaRed x$;
		\item if $S\pbetaRed M'$ then $\labs xS\pbetaRed\labs x{M'}$;
		\item if $S\pbetaRed M'$ and $N\pbetaRed N'$ then
			$\appl SN\pbetaRed\appl{M'}{N'}$;
		\item if $S\pbetaRed M'$ and $N\pbetaRed N'$ then
			$\appl {\labs xS}N\pbetaRed\subst{M'}x{N'}$;
		\item $0\pbetaRed 0$;
		\item if $M\pbetaRed M'$ then $a\sm M\pbetaRed a\sm M'$;
		\item if $M\pbetaRed M'$ and $N\pbetaRed N'$ then $M+N\pbetaRed M'+N'$.
	\end{itemize}
\end{definition}

\label{vpresRed:degenerate}
In particular, if $1\in\rigS$ admits an opposite element $-1\in\rigS$
then $\vpresRed[\boundedStructure]$ is degenerate. Indeed, we can consider $\pbetaRed$ up to 
the equality of vector λ-terms by setting $M\vpbetaRed N$ if
there are $M'\algEq M$ and $N'\algEq N$ such that $M'\pbetaRed N'$.
Since $\teq$ subsumes $\algEq$, the results of Subsection \ref{subsection:simulation}
will imply that $M\vpresRed[\boundedStructure] N$ as soon as $M\vpbetaRed N$.
If $-1\in\rigS$, we have  $M\RT\vpbetaRed N$ 
for all $M,N\in\algebraicTerms$ by Example \ref{example:inconsistency},
hence $M\RT\vpresRed[\boundedStructure] N$.

Using reduction structures, we will nonetheless be able to define a consistent
reduction relation containing β-reduction, but restricted to
those algebraic λ-terms that have a normalizable Taylor expansion, in the 
sense to be defined in Section \ref{section:normalization}.

On the other hand, even assuming $\rigS$ is zerosumfree,
Taylor expansions are not stable under $\vpresRed$:
if $M\vpresRed[\boundedTermStructure] σ'$, we know from the previous 
section that $σ'$ is bounded and $M\vpbsresRed σ'$,
but there is no reason why $σ'$ would be the Taylor expansion of 
an algebraic λ-term.

We do know, however, that $σ'\vpresRed[\boundedTermStructure]\fullReduct{\TaylorExp M}$, which 
will allow us to obtain a weak conservativity result w.r.t. parallel β-reduction:
for all reduction $M\RT\vpresRed[\boundedTermStructure] σ'$ 
there is a reduction $M\RT\pbetaRed M'$
such that $σ'\RT\vpresRed[\boundedTermStructure] M'$, \ie 
any $\vpresRed[\boundedTermStructure]$-reduction sequence from a Taylor expansion can be completed into a
parallel β-reduction sequence (Subsection \ref{subsection:vpresRed:conservative}).
Restricted to normalizable pure λ-terms, this will enable us to obtain an actual
conservativity result.

\subsection{Simulation of parallel β-reduction}

\label{subsection:simulation}

We show that $\vpresRed[\boundedStructure]$ allows to simulate $\pbetaRed$
on $\resourceVectors$, without any particular assumption
on $\rigS$.

\begin{lemma}
	\label{lemma:presRed:beta}
	If $σ\vpresRed[\calS] σ'$ and
	$\ms {τ}\vpresRed[\ms{\calT}] \ms{τ'}$ then
	$\rappl{\labs x{σ}}{\ms{τ}}\vpresRed[\rappl{\labs x{\calS}}{\ms{\calT}}]\lsubst{σ'}x{\ms{τ'}}$.
\end{lemma}
\begin{proof}
	Assume there are summable families
	$\pars{s_i}_{i\in I}$,
	$\pars{σ'_i}_{i\in I}$,
	$\pars{\ms t_j}_{j\in J}$ and
	$\pars{\ms{τ}'_j}_{j\in J}$, 
	and families of scalars
	$\pars{a_i}_{i\in I}\in\calS^I$ and $\pars{b_j}_{j\in J}\in\ms\calT^J$ 
	such that:
	\begin{itemize}
		\item  $σ=\sum_{i\in I} a_i\sm s_i$,
			$σ'=\sum_{i\in I} a_i\sm σ'_i$
			and $s_i\presRed σ'_i$ for all $i\in I$;
		\item  $\ms{τ}=\sum_{j\in J} b_j\sm \ms t_j$,
			$\ms{τ}'=\sum_{j\in J} b_j\sm \ms{τ}'_j$,
			and $\ms t_j\in\ms{\calT}$ and $\ms t_j\presRed \ms{τ}'_j$ for all $j\in J$.
	\end{itemize}
	By multilinear-continuity, the families 
	$\pars{\rappl{\labs x{s_i}}{\ms{t}_j}}_{i\in I,j\in J}$
	and 
	$\pars{\lsubst{σ'_i}x{\ms{τ}'_j}}_{i\in I,j\in J}$
	are sum\-ma\-ble,
	$\rappl{\labs x{σ}}{\ms{τ}}
		=\sum_{i\in I,j\in J}
		a_ib_j\sm \rappl{\labs x{s_i}}{\ms{t}_j}$
	and 
	$\lsubst {σ'}x{\ms{τ}'}
		=\sum_{i\in I,j\in J}
		a_ib_j\sm \lsubst{σ'_i}x{\ms{τ}'_j}$.
	It is then sufficient to observe that 
	$\rappl{\labs x{s_i}}{\ms{t}_j}\presRed \lsubst{σ'_i}x{\ms{τ}'_j}$
	for all $(i,j)\in I\times J$.

	The additional requirement on resource supports is straightforwardly satisfied, since
	$\rappl{\labs x{s_i}}{\ms t_j}\in\rappl{\labs x{\calS}}{\ms\calT}$ for all $(i,j)\in I\times J$.
\end{proof}

\begin{lemma}
	\label{lemma:presRed:contextual:abs-and-app}
	If $σ\vpresRed[\calS] σ'$ then 
	$\labs x{σ}\vpresRed[\labs x{\calS}] \labs x{σ'}$.
	If moreover 
	$\ms {τ}\vpresRed[\ms{\calT}] \ms{τ'}$ then
	$\rappl{σ}{\ms{τ}}\vpresRed[\rappl{\calS}{\ms{\calT}}] \rappl{σ'}{\ms{τ'}}$.
\end{lemma}
\begin{proof}
	Similarly to the previous lemma, each result follows from the
	multilinear-continuity of syntactic operators, and the contextuality of
	$\presRed$.
\end{proof}

\begin{lemma}
	\label{lemma:presRed:contextual:prom}
	If $σ\vpresRed[\calS]σ'$ then 
	$\prom{σ}\vpresRed[\prom{\calS}]\prom{{σ'}}$.
\end{lemma}
\begin{proof}
	Assume there are summable families
	$\pars{s_i}_{i\in I}$ and 
	$\pars{σ'_i}_{i\in I}$,
	and a family of scalars
	$\pars{a_i}_{i\in I}$
	such that  $σ=\sum_{i\in I} a_i\sm s_i$,
			$σ'=\sum_{i\in I} a_i\sm σ'_i$
			and $s_i\presRed σ'_i$ for all $i\in I$.

	Then by multilinear-continuity of the monomial construction,
	for all $n\in\naturals$,
	the families
	$\pars{\mset{s_{i_1},\dotsc,s_{i_n}}}_{i_1,\dotsc,i_n\in I}$
	and
	$\pars{\mset{σ'_{i_1},\dotsc,σ'_{i_n}}}_{i_1,\dotsc,i_n\in I}$
	are summable, and
	\[{σ}^n=\sum_{i_1,\dotsc,i_n\in I} a_{i_1}\cdots a_{i_n}\mset{s_{i_1},\dotsc,s_{i_n}}\]
	and \[{σ'}^n =\sum_{i_1,\dotsc,i_n\in I}
	a_{i_1}\cdots a_{i_n}\sm\mset{σ'_{i_1},\dotsc,σ'_{i_n}}.\]
	Since the supports of the monomial vectors $σ^n$ (resp.\ ${σ'}^n$)
	for $n\in\naturals$ are pairwise disjoint, 
	we obtain that the families
	$\pars{\mset{s_{i_1},\dotsc,s_{i_n}}}_{\substack{n\in\naturals\\i_1,\dotsc,i_n\in I}}$
	and
	$\pars{\mset{σ'_{i_1},\dotsc,σ'_{i_n}}}_{\substack{n\in\naturals\\i_1,\dotsc,i_n\in I}}$
	are summable, 
	and
	\[
		\prom{σ}
		=\sum_{n\in\naturals}\frac 1{\factorial n}\sm {σ}^n
		=\sum_{\substack{n\in\naturals\\i_1,\dotsc,i_n\in I}}
		\frac {a_{i_1}\cdots a_{i_n}}{\factorial n}\sm\mset{s_{i_1},\dotsc,s_{i_n}}
	\]
	and 
	\[
		\prom{{σ'}}
		=\sum_{n\in\naturals}\frac 1{\factorial n}\sm {σ'}^n
		=\sum_{\substack{n\in\naturals\\i_1,\dotsc,i_n\in I}}
		\frac {a_{i_1}\cdots a_{i_n}}{\factorial n}\sm\mset{σ'_{i_1},\dotsc,σ'_{i_n}}
	\]
	which concludes the proof since each
	$\mset{s_{i_1},\dotsc,s_{i_n}}\presRed\mset{σ'_{i_1},\dotsc,σ'_{i_n}}$.
\end{proof}

\begin{lemma}
	\label{lemma:presRed:contextual:semimodule}
	If $ε\vpresRed[\calE] ε'$ and $φ\vpresRed[\calF] φ'$ then 
	$a\sm ε\vpresRed[\calE] a\sm ε'$ and $ε+φ\vpresRed[\calE\union\calF] ε'+φ'$.
\end{lemma}
\begin{proof}
	Follows directly from the definitions, using the fact that 
	summable families form a $\rigS$-semimodule.
\end{proof}

\begin{lemma}
	\label{lemma:vpresRed:correct}
	If $M\pbetaRed M'$
	then $M\vpresRed[\TaylorSup M] M'$.
\end{lemma}
\begin{proof}
	By induction on the reduction $M\pbetaRed M'$
	using Lemmas \ref{lemma:presRed:beta} to \ref{lemma:presRed:contextual:prom}
	in the cases of reduction from a simple term,
	and Lemma \ref{lemma:presRed:contextual:semimodule}
	in the case of reduction from an algebraic term.
\end{proof}

Recalling that $\TaylorSup M\in\boundedTermStructure$ we obtain:

\begin{corollary}
	\label{corollary:vpresRed:correct}
	If $M\pbetaRed M'$
	then $M\vpresRed[\boundedTermStructure] M'$.
\end{corollary}

Observe that these results hold on Taylor supports as well, 
which will be useful in the treatment of Taylor normalizable
terms in Section \ref{section:normalization}:
\begin{lemma}
	\label{lemma:vpresRed:correct:TaylorSup}
	If $M\pbetaRed M'$
	then $\TaylorSup M\vpresRed[\TaylorSup M] \TaylorSup{M'}$
	in $\termVectors[\booleans]$.
\end{lemma}
\begin{proof}
	The proof is again by induction on the reduction $M\pbetaRed M'$ using Lemmas
	\ref{lemma:presRed:beta} to Lemma \ref{lemma:presRed:contextual:semimodule}
	in $\termVectors[\booleans]$.
\end{proof}

\subsection{Conservativity}

\label{subsection:vpresRed:conservative}

\begin{definition}
We define the \emph{full parallel reduct} of simple terms and algebraic terms
	inductively as follows:
	\begin{align*}
		\fullReduct{x}                      &\eqdef x&
		\fullReduct{0}                      &\eqdef 0\\
		\fullReduct{\labs xS}               &\eqdef \labs x{\fullReduct{S}}&
		\fullReduct{a\sm M}                 &\eqdef a\sm\fullReduct{M}\\
		\fullReduct{\appl{\labs xS}{N}}     &\eqdef \subst{\fullReduct{S}}x{\fullReduct{N}}&
		\fullReduct{M+N}                    &\eqdef \fullReduct{M}+\fullReduct{N}\\
		\fullReduct{\appl{S}{N}}&\eqdef \appl{\fullReduct{S}}{\fullReduct{N}}
		\lefteqn{\quad\text{(if $S$ is not an abstraction).}}
	\end{align*}
\end{definition}

As can be expected, we have 
	$M'\pbetaRed \fullReduct{M}$
as soon as 
	$M\pbetaRed M'$.
In this subsection, we will show that a similar property holds 
for $\vpresRed[\boundedStructure]$.

Recall that, by Lemma \ref{lemma:fullReduct:summable},
the full reduction operator $\fullReductSym$ on resource expressions
extends to bounded resource vectors.
We obtain:
\begin{lemma}
	\label{lemma:fullReduct:contextual}
	For all bounded $σ_0\in\termVectors$, $\ms{τ}\in\monomialVectors$, $ε,φ\in\resourceVectors$,
	\begin{align*}
		\fullReduct{x}                         &= x&
		\fullReduct{\prom{σ}}                  &= \prom{\fullReduct{σ}}\\
		\fullReduct{\labs x{σ}}                &= \labs x{\fullReduct{σ}}&
		\fullReduct{a\sm ε}                    &= a\sm {\fullReduct{ε}}\\
		\fullReduct{\rappl{\labs x{σ}}{\ms{τ}}}&= \lsubst{\fullReduct{σ}}x{\fullReduct{\ms{τ}}}&
		\fullReduct{ε+φ}                       &= \fullReduct{ε}+\fullReduct{φ}\\
		\fullReduct{\rappl{σ_0}{\ms{τ}}}       &= \rappl{\fullReduct{σ_0}}{\fullReduct{\ms{τ}}}
		\lefteqn{\quad\text{(if there is no abstraction term in $\support{σ_0}$).}}
	\end{align*}
\end{lemma}
\begin{proof} 
	The proofs of those identities are basically the same as those of Lemmas 
	\ref{lemma:presRed:beta} to \ref{lemma:presRed:contextual:semimodule},
	the necessary summability conditions following from Lemma
	\ref{lemma:fullReduct:summable}.
\end{proof}

\begin{lemma}
	For all $M\in\algebraicTerms$, 
	$\fullReduct{\TaylorExp M}=\TaylorExp{\fullReduct M}$.
\end{lemma}
\begin{proof}
	We know that $\TaylorExp{M}$ is bounded.
	The identity is then proved by induction on simple terms and algebraic terms,
	using the previous lemma in each case.
\end{proof}

\begin{lemma}
	For all bounded term reduction structure $\fS$
	and all $M\in\algebraicTerms$, 
	if $M\vpresRed[\fS] σ'$ then $σ'\vpresRed[\fS]\fullReduct{M}$.
\end{lemma}
\begin{proof}
	By Corollary \ref{corollary:splitVariant:presRed:fullReduct},
	$σ'\vpresRed[\fS]\fullReduct{\TaylorExp{M}}$
	and we conclude by the previous lemma.
\end{proof}

This result can then be generalized to sequences of $\vpresRed$-reductions.
\begin{lemma}
	\label{lemma:vpresRed:fullReduct}
	For all bounded term reduction structure $\fS$
	and all $M\in\algebraicTerms$, 
	if $M\redIterate n\vpresRed[\fS] σ'$
	then $σ'\redIterate n\vpresRed[\fS]\nthFullReduct n{M}$.
\end{lemma}
\begin{proof}
	By induction on $n$. The case $n=0$ is trivial, and the inductive case follows
	from the previous lemma and strong confluence of $\vpresRed[\fS]$:
	if $M\redIterate n\vpresRed[\fS] σ'\vpresRed[\fS] τ$
	then by induction hypothesis $σ'\redIterate n\vpresRed[\fS]\nthFullReduct n{M}$,
	hence by strong confluence, there exists $τ'$ such that 
	$τ\redIterate n\vpresRed[\fS] τ'$ and 
	$\nthFullReduct n{M}\vpresRed[\fS] τ'$;
	by the previous lemma, 
	$τ'\vpresRed \nthFullReduct {n+1}{M}$.
\end{proof}

We have thus obtained some weak kind of conservativity of $\vpresRed[\boundedTermStructure]$ w.r.t.
β-reduction, but it is not very satisfactory: the same result would hold 
for the tautological relation
$\finitaryVectors{\boundedTermStructure}\times\finitaryVectors{\boundedTermStructure}$,
which is indeed the same as $\vpresRed[\boundedTermStructure]$ if $1$ has an opposite 
element in $\rigS$.
Even when $\rigS$ is zerosumfree, the converse to Lemma
\ref{lemma:vpresRed:correct} cannot hold in general if only because there can
be distinct $β$-normal forms $M\not\algEq N$ such that $M\teq N$ (see Example
\ref{example:TaylorExp:not:injective}).
Under this hypothesis, we can nonetheless obtain an actual conservativity
result on normalizable pure λ-terms as follows.

We write $\betaEq$ for the symmetric, reflexive and transitive closure of
$\pbetaRed$.
Similarly, if $\fE$ is a reduction structure, we write $\resEq[\fE]$ for the
equivalence on $\finitaryVectors{\fE}$ induced by $\vpresRed[\fE]$.
\begin{lemma}
	\label{lemma:resEq:conservative:zerosumfree}
	Assume $\rigS$ is zerosumfree.
	Let $M,N\in\lambdaTerms$ be such that $M$ is normalizable.
	Then $M\resEq[\boundedTermStructure] N$ iff $M\betaEq N$.
\end{lemma}
\begin{proof}
  Corollary \ref{corollary:splitVariant:presRed:fullReduct} entails that,
  if $\fE$ is a bounded reduction structure, then 
  $ε\resEq[\fE]ε'$ iff $ε\RT\vpresRed[\fE]\nthFullReduct n{ε'}$ for some $n\in\naturals$.
  Now assume $M\in\algebraicTerms$ is normalizable and 
  write $\NormalForm M$ for its normal form:
  in particular $M\RT\vpresRed[\boundedTermStructure]\NormalForm M$, by 
  Corollary \ref{corollary:vpresRed:correct}.
  If $M\resEq[\boundedTermStructure]N$, we thus have
  $\NormalForm M\resEq[\boundedTermStructure]N$, hence
  $\NormalForm M\RT\vpresRed[\boundedTermStructure]\nthFullReduct n{N}$
  for some $n\in\naturals$.
  In particular, if $\rigS$ is zerosumfree, we obtain
  $\NormalForm M\teq\nthFullReduct n{N}$.
	If moreover $M,N\in\lambdaTerms$, we deduce $M\betaEq N$ by the injectivity
	of $\TaylorExpSym$ on $\lambdaTerms$.
\end{proof}

The next section will allow us to establish a similar conservativity result,
without any assumption on $\rigS$, at the cost of restricting the reduction
relation to normalizable resource vectors.

\section{Normalizing Taylor expansions}

\label{section:normalization}

Previous works on the normalization of Taylor expansions were restricted
\emph{a priori}, to a strict subsystem of the algebraic λ-calculus:
\begin{itemize}
	\item the uniform setting of pure λ-terms \cite{er:resource,er:bkt};
	\item the typed setting of an extension 
		of system $F$ to the algebraic λ-calculus \cite{ehrhard:finres};
	\item a λ-calculus extended with formal finite sums,
		rather than linear combinations \cite{ptv:taylorsn,tao:species}.
\end{itemize}
In all these, pathological terms were avoided, e.g. those involved
in the inconsistency Example \ref{example:inconsistency}.
Moreover observe that the very notion of normalizability is not compatible with
$\algEq$, and in particular the identity $0\algEq 0\sm M$: those previous works 
circumvented this incompatibility, either by imposing normalizability
via typing, or by excluding the formation of the term $0\sm M$.

Our approach is substantially different.
We introduce a notion of normalizability on resource vectors such that:
\begin{itemize}
	\item both pure λ-terms and normalizable algebraic λ-terms
		(in particular typed algebraic λ-terms and normalizable λ-terms with sums)
		have a normalizable Taylor expansion;
	\item the restriction of $\vpresRed$ to normalizable resource vectors
		is a consistent extension of both β-reduction on pure λ-terms and 
		normalization on algebraic λ-terms, without any assumption on the 
		underlying semiring of scalars.
\end{itemize}

\subsection{Normalizable resource vectors}

\label{subsection:normalizable}

We say $ε\in\resourceVectors$ is \emph{normalizable} whenever the family 
$\pars{\NormalForm{e}}_{e\in\support{ε}}$
is summable. In this case, we write
$\NormalForm{ε}\eqdef\sum_{e\in\resourceExpressions} ε_e\sm\NormalForm{e}$.

Normalizable vectors form a finiteness space.
Recall indeed from Subsection \ref{subsection:resourceExpressions}
that $e\generates e'$ iff $e\resRedRT ε'$ with $e'\in\support{ε'}$.
If $e\in\resourceExpressions$, we write
$\cone e\eqdef\set{e'\in\resourceExpressions\st e'\generates e}$.
Then $ε$ is normalizable iff
for each normal resource expression $e$, $\support{ε}\inter\cone e$ is finite:
writing $\normExpressions=\set{e\in\resourceExpressions\st\text{$e$ is normal}}$
and $\normStructure=\dual{\set{\cone e\st e\in\normExpressions}}\inter\resourceStructure$,
we obtain that $\normResourceVectors$ is the set of normalizable resource vectors.
Observe that $\NormalFormSym$ is defined on all $\normResourceVectors$
but is guaranteed to be linear-continuous only when restricted to
subsemimodules of the form $\vectors {\calE}$ with $\calE\in\normStructure$.

For our study of hereditarily determinable terms in Section \ref{section:determinable}, 
it will be useful to decompose $\normStructure$ into a decreasing sequence 
of finiteness structures.

\begin{definition}
	\label{definition:monomial:depth}
	We define the \emph{monomial depth} $\depth e\in\naturals$ of a
	resource expression $e\in\resourceExpressions$ as follows:
	\begin{align*}
		\depth x&\eqdef0&
		\depth{\rappl s{\ms t}}&\eqdef\max(\depth s,\depth{\ms t})\\
		\depth{\labs xs}&\eqdef\depth s&
		\depth{\mset{t_1,\dotsc,t_n}}&\eqdef1+\max\set{\depth{t_i}\st 1\le i\le n}
	\end{align*}
\end{definition}
We write $\kNormExpressions d=\set{e\in\normExpressions\st\depth e\le d}$
so that $\normExpressions=\Union_{d\in\naturals}\kNormExpressions d$.
We then write $\kNormStructure d=\dual{\set{\cone e\st e\in\kNormExpressions d}}\inter\resourceStructure$
so that $\normStructure=\Inter_{d\in\naturals}\kNormStructure d$.
Each finiteness structure $\kNormStructure d$ is moreover a reduction structure
for any reduction relation contained in $\RT\resRed$
(and so is $\normStructure$). Indeed,
writing $\genCone e=\set{e'\in\resourceExpressions\st e\generates e'}$
and $\genCone\calE=\Union_{e\in\calE}\genCone e$, we obtain:
\begin{lemma}
	\label{lemma:normStructure:reduction}
	If $\calE\in\kNormStructure d$ then 
	$\genCone\calE\in\kNormStructure d$.
\end{lemma}
\begin{proof}
	Let $e''\in\kNormExpressions d$ and 
	$e'\in\genCone \calE\inter\cone{e''}$.
	Necessarily, there is $e\in\calE$ such that $e\generates e'$.
	Then $e\in\calE\inter\cone{e''}$:
	since $\calE\in\kNormStructure d$, there are finitely many values for $e$ 
	hence for $e'$ by Lemma~\ref{lemma:resRed:SN}.
\end{proof}

It follows that normalizable vectors are stable under reduction:
\begin{lemma}
	\label{lemma:vpresRed:NormalForm}
	If $ε\vpresRed[\normStructure] ε'$ 
	then $ε'\in\normResourceVectors$
	and $\NormalForm{ε}=\NormalForm{ε'}$.
\end{lemma}
\begin{proof}
	Assume there exists $\calE\in\normStructure$ and  families $(a_i)_{i\in I}\in \rigS^I$, 
	$\pars{e_i}_{i\in I}\in\resourceExpressions^I$ 
	and $\pars{ε'_i}_{i\in I}\in \finiteResourceSums^I$ 
	such that:
	\begin{itemize}
		\item $\pars{e_i}_{i\in I}$ is summable and $ε=\sum_{i\in I}a_i\sm e_i$;
		\item $\pars{ε'_i}_{i\in I}$ is summable and $ε'=\sum_{i\in I}a_i\sm ε'_i$;
		\item for all $i\in I$, $e_i\in\calE$ and $e_i\vpresRed ε'_i$.
	\end{itemize}
	We obtain that $\calE'\eqdef\Union_{i\in I}\support{ε'_i}\in\normStructure$
	by Lemma \ref{lemma:normStructure:reduction},
	hence $ε'\in\normResourceVectors$ since $\support{ε'}\subseteq\calE'$.
	Then, by the linear-continuity of 
	$\NormalFormSym$ on $\vectors{\calE'}$,
	\[\NormalForm{ε}
	=\sum_{i\in I}a_i\sm\NormalForm{e_i}
	=\sum_{i\in I}a_i\sm\NormalForm{ε'_i}
	=\NormalForm{\sum_{i\in I}a_i\sm ε'_i}
	=\NormalForm{ε'}.\qedhere\]
\end{proof}

As a direct consequence, we obtain that $\resEq[\normStructure]$ is consistent,
without any additional condition on the semiring $\rigS$:
\begin{corollary}
	\label{corollary:resEq:normalStructure:NormalForm}
	If $ε\resEq[\normStructure] ε'$
	(in particular $ε,ε'\in\normResourceVectors$)
	then $\NormalForm{ε}=\NormalForm{ε'}$.
\end{corollary}

We can moreover show that the normal form of a Taylor normalizable
term is obtained as the limit of the parallel left reduction strategy.
Let us first precise the kind of convergence we consider.
With the notations of Subsection \ref{subsection:summable}, we say a sequence
$\vec{ξ}=\pars{ξ_n}_{n\in\naturals}\in\pars{\vectors{X}}^\naturals$
of vectors converges to $ξ'$ if, for all $x\in X$ there exists
$n_x\in\naturals$ such that, for all $n\ge n_x$, $ξ_{n,x}=ξ'_x$.
In other words we consider the product topology on $\vectors X$,
$\rigS$ being endowed with the discrete topology.
Similarly to the notion of summability, 
this notion of convergence coincides with that induced 
by the linear topology on $\vectors X$ associated 
with the maximal finiteness structure $\powerset X$ on $X$:
in this particular case, a base of neighbourhoods of $0$ is given by the sets
$\set{ξ\in\vectors X\st \support{ξ}\inter\calX'=\emptyset}$
for $\calX'\in\dual{\powerset X}=\finiteSubsets X$,
or equivalently by the the sets
$\set{ξ\in\vectors X\st x\not\in\support{ξ}}$
for $x\in X$.

The parallel left reduction strategy on resource vectors is defined as follows.
\begin{definition}
	We define the \emph{left reduct of a resource expression} inductively as follows:
	\begin{align*}
		\leftReduct {\labs xs} &\eqdef \labs x{\leftReduct s}\\
		\leftReduct {\mset{t_1,\dotsc,t_n}} &\eqdef \mset{\leftReduct{t_1},\dotsc,\leftReduct{t_n}}\\
		\leftReduct {\rappl {x}{\ms t_1\cdots\ms t_n}} &\eqdef 
			\rappl {x}{\leftReduct{\ms t_1}\cdots\leftReduct{\ms t_n}}\\
		\leftReduct {\rappl {\labs xs}{\ms t_0\,\ms t_1\cdots\ms t_n}} &\eqdef 
			\rappl {\lsubst sx{\ms t_0}}{\ms t_1\cdots\ms t_n}.
	\end{align*}
	This is extended to finite sums of resource expressions by linearity:
	$\leftReduct{\sum_{i=1}^n e_i}=\sum_{i=1}^n\leftReduct{e_i}$.
\end{definition}

\begin{lemma}
	\label{lemma:leftReduct:pbdresRed}
	For all resource expression $e\in\resourceExpressions$,
	$e\pbdresRed 1 \leftReduct e$.
\end{lemma}
\begin{proof}
	Easy by induction on $e$.
\end{proof}
In particular $\NormalForm{e} = \NormalForm{\leftReduct{e}}$ for all $e\in\resourceExpressions$.
By Lemma \ref{lemma:pbdresRed:size}, we moreover obtain
that if $e'\in\support{\leftReduct{e}}$ then
$\size{e}\le 4\size{e'}$ and $\fv e=\fv{e'}$.
As a consequence
$\pars{\leftReduct e}_{e\in\resourceExpressions}$ is summable.
For all $ε\in\resourceVectors$, we set 
\[ \leftReduct{ε}\eqdef\sum_{e\in\resourceExpressions} ε_e\sm\leftReduct e\]
and obtain a linear-continuous map on resource vectors.

For all $ε\in\resourceVectors$, we write $\normRestr{ε}$ 
for the \emph{projection of $ε$ on normal resource expressions}:
$\normRestr{ε}\eqdef\sum_{e\in\normExpressions} ε_e\sm e\in\vectors{\normExpressions}$.
We obtain:
\begin{theorem}
	\label{theorem:NormalForm:limit}
	For all normalizable resource vector $ε\in\normResourceVectors$,
	$\pars{\normRestr{\nthLeftReduct k{ε}}}_{k\in\naturals}$
	converges to $\NormalForm{ε}$ in $\vectors{\normExpressions}$.
\end{theorem}
\begin{proof}
	Fix ${e'}\in\normTerms$.
	Since $\support{ε}\in\normStructure$,
	$\calE\eqdef\support{ε}\inter\cone{e'}$ is finite.
	Let ${k'}$ be such that $\nthLeftReduct{k'}e$ is normal 
	for all $e\in\calE$.
	Then $\NormalForm{ε}_{e'}
	=\sum_{e\in\support{ε}} ε_e\sm{\NormalForm e}_{e'}
	=\sum_{e\in\calE} ε_e\sm{\NormalForm e}_{e'}
	=\sum_{e\in\calE} ε_e\sm{\nthLeftReduct {k'}e}_{e'}$.
	Moreover, by the linear-continuity of $\nthLeftReductSym k$
	on resource vectors,
	$\pars{\normRestr{\nthLeftReduct k{ε}}}_{e'}
	=\nthLeftReduct k{ε}_{e'}
	=\sum_{e\in\resourceTerms} ε_e\sm{\nthLeftReduct k{e}}_{e'}
	=\sum_{e\in\calE} ε_e\sm{\nthLeftReduct k{e}}_{e'}
	=\sum_{e\in\calE} ε_e\sm{\nthLeftReduct {k'}{e}}_{e'}$.
\end{proof}
Observe that the projection on normal expressions is essential:
\begin{example}
	Consider the looping term $Ω\eqdef\appl{\labs x{\appl xx}}{\labs x{\appl xx}}$:
	one can check that $\NormalForm{\TaylorExp{Ω}}=\normTermRestr{\TaylorExp{Ω}}=0$,
	but it will follow from the results of subsection
	\ref{subsection:Taylor:normalizable} that $\nthLeftReduct k{\TaylorExp{Ω}}
	=\TaylorExp{Ω}\not=0$ for all $k\in\naturals$.

	Analyzing this phenomenon was fundamental in the characterization of strongly
	normalizable λ-terms by a finiteness structure on resource terms,
	obtained by Pagani, Tasson and the author \cite{ptv:taylorsn}.
\end{example}

\subsection{Taylor normalizable terms}

\label{subsection:Taylor:normalizable}

It is possible to transfer some of the good properties of reduction on
normalizable vectors to those algebraic λ-terms that have a normalizable Taylor
expansion.
More precisely, we say $M\in\algebraicTerms$ is \emph{Taylor normalizable} if
$\TaylorSup M\in\normStructure$.
Then:
\begin{lemma}
	\label{lemma:betaRed:Taylor:normalizable}
	Assume $M,M'\in\algebraicTerms$ are such that $M\pbetaRed M'$.
	Then $M$ is Taylor normalizable iff $M'$ is Taylor normalizable.
\end{lemma}
\begin{proof}
	First observe that by Lemma \ref{lemma:vpresRed:correct:TaylorSup},
	we have $\TaylorSup M\vpresRed[\TaylorSup{M}]\TaylorSup{M'}$ in
	$\termVectors[\booleans]$.
	Moreover observe that $\normTermVectors[\booleans]$ is nothing but
	$\normTermStructure$.

	Assume $M$ is Taylor normalizable, \ie $\TaylorSup M\in\normTermStructure$:
	by Lemma \ref{lemma:vpresRed:NormalForm}, $\TaylorSup{M'}\in\normTermVectors[\booleans]$,
	\ie $M'$ is Taylor normalizable.

	Conversely, assume $M'$ is Taylor normalizable and
	let $s''\in\normTerms$ and $\calS\eqdef\TaylorSup M\inter\cone{s''}$:
	we prove $\calS$ is finite.
	Fix an enumeration $\pars{s_k}_{k\in K}\in\calS^K$ of $\calS$:
	$\calS=\set{s_k\st k\in K}$.
	Since $\TaylorSup M\vpresRed[\TaylorSup{M}]\TaylorSup{M'}$, we have 
	$\TaylorSup M=\set{t_i\st i\in I}$ and
	$\TaylorSup{M'}=\Union_{i\in I}\support{τ'_i}$ with $t_i\presRed τ'_i$ for all $i\in I$.
	Now for all $k\in K$, there exists $i\in I$ such that 
	$s_k=t_i$. Since $s_k\generates s''$, $τ'_i\not=0$ 
	and we can fix $s'_k\in\support{τ'_i}\subseteq\TaylorSup{M'}$
	such that $s_k\pOneStepGenerates s'_k\generates s''$.
	Since $\TaylorSup{M'}\in\normTermStructure$, the set $\set{s'_k\st k\in K}$ is finite.
	Then $\calS\subseteq\set{s\in\resourceTerms\st k\in K\text{, }s\pbdOneStepGenerates{\height M}s'_k}$
	which is finite by Lemma \ref{lemma:pbdresRed:size}.
\end{proof}

The consistency of β-reduction on Taylor normalizable terms follows.
\begin{theorem}
	\label{theorem:betaEq:NormalForm}
	Assume $M,M'\in\algebraicTerms$ are such that $M\betaEq M'$.
	Then $M$ is Taylor normalizable iff $M'$ is Taylor normalizable,
	and in this case $\NormalForm{\TaylorExp M}=\NormalForm{\TaylorExp{M'}}$.\footnote{
		In the standard terminology of denotational semantics,
		Theorem \ref{theorem:betaEq:NormalForm} expresses 
		the soundness of $\NormalForm{\TaylorExp \cdot}$ on 
		Taylor normalizable terms.
	}
\end{theorem}
\begin{proof}
	The first part is a direct corollary of Lemma \ref{lemma:betaRed:Taylor:normalizable}.
	By Lemma \ref{lemma:vpresRed:correct}, it follows that 
	$M\resEq[\normStructure] M'$, and then we conclude by Corollary
	\ref{corollary:resEq:normalStructure:NormalForm}.
\end{proof}
In other words, when restricted to Taylor normalizable terms,
the normal form of Taylor expansion is a valid notion of denotation.
Remark that, in general, it is not possible to generalize this result to those
terms $M$ such that $\TaylorExp M$ is normalizable because of the interaction 
with coefficients:
consider, \eg, $0\teq \appl{I}{\infty_x} +(-1)\sm \appl{I}{\infty_x}
\pbetaRed \infty_x +(-1)\sm\appl{I}{\infty_x}$, and observe that
$\TaylorExp{\infty_x +(-1)\sm\appl{I}{\infty_x}}\not\in\normTermVectors$.

\begin{definition}
	We define the \emph{left reduct of an algebraic λ-term} inductively as follows:
	\begin{align*}
		\leftReduct {\labs xS} &\eqdef \labs x{\leftReduct S}&
		\leftReduct 0 &\eqdef 0\\
		\leftReduct {\appl {x}{ M_1\cdots M_n}} &\eqdef 
		\appl {x}{\leftReduct{ M_1}\cdots\leftReduct{ M_n}}&
		\leftReduct {a\sm M} &\eqdef a\sm\leftReduct M\\
		\leftReduct {\appl {\labs xS}{ M_0\, M_1\cdots M_n}} &\eqdef 
		\appl {\subst Sx{M_0}}{ M_1\cdots M_n}&
		\leftReduct {M+N} &\eqdef \leftReduct M+\leftReduct N
	\end{align*}
\end{definition}
Observe that this definition is exhaustive by Fact
\ref{fact:simple:head}.
It should be clear that $M\pbetaRed\leftReduct M$ for all term $M$,
and that $\leftReduct M=M$ when $M$ is in normal form 
(although the converse may not hold).
Now we can establish that $\leftReductSym$ commutes with Taylor expansion.

\begin{lemma}
	\label{lemma:leftReduct:prom}
	For all $σ\in\termVectors$, $\leftReduct{\prom{σ}}=\prom{\leftReduct{σ}}$.
\end{lemma}
\begin{proof}
	First observe that by the definition of $\leftReductSym$ and 
	the linear-continuity of both $\leftReductSym$ and the monomial construction,
	for all $σ_1,\dotsc,σ_n\in\termVectors$, we have
	$\leftReduct{\mset{σ_1,\dotsc,σ_k}}=\mset{\leftReduct{σ_1},\dotsc,\leftReduct{σ_k}}$.
	In particular, $\leftReduct{σ^k}=\leftReduct{σ}^k$.
	We deduce that
	$\leftReduct{\prom{σ}}
	=\leftReduct{\sum_{k\in\naturals}\frac 1{\factorial k}\sm{σ^k}}
	=\sum_{k\in\naturals}\frac 1{\factorial k}\sm\leftReduct{σ}^k
	=\prom{\leftReduct{σ}}
	$, by the linear-continuity of $\leftReductSym$.
\end{proof}

\begin{lemma}
	\label{lemma:leftReduct:TaylorExp}
	For all $M\in\algebraicTerms$, $\leftReduct{\TaylorExp M}=\TaylorExp{\leftReduct M}$.
\end{lemma}
\begin{proof}
	By induction on the definition of $\leftReduct M$:
	in addition to the inductive hypothesis
	and the linear-continuity of $\leftReductSym$, we use
	Lemma \ref{lemma:leftReduct:prom} in the case of a head variable, 
	and Lemmas \ref{lemma:vlsubst:promotion:constructors},
  \ref{lemma:taylor:subst} and 
	\ref{lemma:leftReduct:prom}
	in the case of a head β-redex.
\end{proof}

As a direct corollary of Theorem \ref{theorem:NormalForm:limit},
we obtain:
\begin{theorem}
	\label{theorem:Taylor:normalizable:limit}
	For all Taylor normalizable term $M$, 
	the sequence of normal resource vectors
	$\pars{\normTermRestr{\TaylorExp{\nthLeftReduct k{M}}}}_{k\in\naturals}$
	converges to $\NormalForm{\TaylorExp{M}}$ in $\vectors{\normTerms}$.
\end{theorem}

This property is very much akin to the fact that the Böhm tree $\BohmTree M$ of
a pure λ-term $M$ is obtained as the limit (in an order theoretic sense) 
of normal form approximants of the left reducts of $M$.
This analogy will be made explicit in Section \ref{section:determinable}.
Before that, we apply our results to normalizable algebraic λ-terms.

\subsection{Taylor expansion and normalization commute on the nose}

\label{subsection:normalizable:terms}

By a general standardization argument, we can show that parallel reduction is a
normalization strategy:

\begin{lemma}
	\label{lemma:leftReduct:normalizes}
	An algebraic λ-term $M$ is normalizable iff there exists $k\in\naturals$,
	such that $\nthLeftReduct k M=\NormalForm M$.
\end{lemma}
\begin{proof}
	Recall that we consider algebraic λ-terms up to $\canEq$ only.
	Then one can for instance use the general standardization technique
	developed by Leventis for a slightly different presentation of the calculus
	\cite{leventis:phd}.
\end{proof}

A direct consequence is that $M$ normalizes iff the judgement $\converges M$
can be derived inductively by the following rules:\footnote{
		Moreover, it seems natural to conjecture that if $\converges M$ then $M$ (or,
		rather, its $\algEq$-class) is normalizable in the sense of Alberti
		\cite{alberti:phd}, and then the obtained normal forms are the same (up to
		$\algEq$).
	}
\begin{center}
	\begin{prooftree}
		\Hypo  {\converges S}
		\Infer1{\converges{\labs xS}}
	\end{prooftree}
	\quad
	\begin{prooftree}
		\Hypo  {\converges{M_1}}
		\Hypo  {\cdots}
		\Hypo  {\converges{M_n}}
		\Infer3{\converges{\appl{x}{M_1\cdots M_n}}}
	\end{prooftree}
	\quad
	\begin{prooftree}
		\Hypo  {\converges{\appl{\subst Sx{M_0}}{M_1\cdots M_n}}}
		\Infer1{\converges{\appl{\labs xS}{M_0\,M_1\cdots M_n}}}
	\end{prooftree}
	\quad
	\begin{prooftree}
		\Infer0{\converges 0}
	\end{prooftree}
	\quad
	\begin{prooftree}
		\Hypo  {\converges M}
		\Infer1{\converges{a\sm M}}
	\end{prooftree}
	\quad
	\begin{prooftree}
		\Hypo  {\converges M}
		\Hypo  {\converges N}
		\Infer2{\converges[]{M+N}}
	\end{prooftree}
\end{center}

In the remaining of this subsection, we prove that normalizable algebraic
λ-terms are Taylor normalizable, using a reducibility technique:
like in Ehrhard’s work for the typed case \cite{ehrhard:finres}, or
our previous work for the strongly normalizable case \cite{ptv:taylorsn},
$\normStructure$ is the analogue of a reducibility candidate.
We prove each key property (Lemmas \ref{lemma:normalizable:lambda}
to \ref{lemma:normalizable:head})
using the family of structures $\kNormStructure d$
rather than $\normStructure$ directly: this will be useful in section 
\ref{section:determinable}, while 
the corresponding results for $\normStructure$ are immediately derived from those.

\begin{lemma}
	\label{lemma:normalizable:lambda}
	If $\calS\in\kNormTermStructure d$ then $\labs x{\calS}\in\kNormTermStructure d$.
\end{lemma}
\begin{proof}
	Let $t'\in\kNormTerms d$ and $t\in\pars{\labs x{\calS}}\inter\cone {t'}$.
	Necessarily, $t=\labs xs$ and $t'=\labs x{s'}$ with $s\in{\calS}\inter\cone{s'}$
	which is finite by assumption.
\end{proof}

\begin{lemma}
	If $\calS\in\kNormTermStructure d$ then 
	$\prom{\calS}\in\kNormMonomialStructure{d+1}$.
\end{lemma}
\begin{proof}
	Let $\ms {t}'\in\kNormMonomials {d+1}$ and $\ms t\in{\prom{\calS}}\inter\cone {\ms{t}'}$.
	Write $n=\card{\ms t'}$.
	Without loss of generality, we can write
	$\ms t=\mset{t_1,\dotsc,t_n}$ and 
  $\ms t'=\mset{t'_1,\dotsc,t'_n}$ so that 
	$t_i\generates t'_i$ and $t'_i\in\kNormTerms d$, for all $i\in\set{1,\dotsc,n}$.
	Since $\ms t\in{\prom{\calS}}$, each $t_i\in{\calS}$.
	Since $\calS\in\kNormTermStructure d$, $t'_i$ being fixed, there are finitely
  many possible values for each $t_i$.
\end{proof}

\begin{lemma}
	If $\ms{\calT}_1,\dotsc,\ms{\calT}_n\in\kNormMonomialStructure d$ then 
	$\rappl {x}{\ms{\calT}_1\cdots\ms{\calT}_n}\in\kNormTermStructure{d}$.
\end{lemma}
\begin{proof}
	Let $t'\in\kNormTerms{d}$ and 
	$t\in\pars{\rappl {x}{\ms{\calT}_1\cdots\ms{\calT}_n}}\inter\cone {t'}$.
	Necessarily, $t = \rappl {x}{\ms t_1\cdots\ms t_n}$ and $t' = \rappl {x}{\ms t'_1\cdots\ms t'_n}$
	and, for each $i\in\set{1,\dotsc,n}$, $\ms t_i\in\ms{\calT}_i$, $\ms t_i\generates \ms t'_i$ and
	$\ms t'_i\in\kNormMonomials d$: since $\ms{\calT}_i\in\kNormMonomialStructure d$, 
	there are finitely many possible values for each $\ms t_i$.
\end{proof}

\begin{corollary}
	\label{corollary:normalizable:neutral}
	If $\calT_1,\dotsc,\calT_n\in\kNormTermStructure d$ then 
	$\rappl {x}{\prom{\calT_1}\cdots{\prom{\calT_n}}}\in\kNormTermStructure{d+1}$.
\end{corollary}

\begin{lemma}
	\label{lemma:normalizable:head}
	If
	$\rappl {\lsubst {\calS}x{\ms{\calT}_0}}{{\ms{\calT}_1}\cdots{\ms{\calT}_n}}\in\kNormTermStructure d$
	then
	$\rappl{\labs x{\calS}}{{\ms{\calT}_0}\,{\ms{\calT}_1}\cdots{\ms{\calT}_n}}\in\kNormTermStructure d$.
\end{lemma}
\begin{proof}
	Let $u'\in\kNormTerms d$,
	and let $u\in\pars{\rappl{\labs x{\calS}}{{\ms{\calT}_0}\,{\ms{\calT}_1}\cdots{\ms{\calT}_n}}}\inter\cone {u'}$.
	In other words, $u'\in\support{\NormalForm u}$ and we can write
	$u=\rappl{\labs x{s}}{{\ms{t}_0}\,{\ms{t}_1}\cdots{\ms{t}_n}}$
	with $s\in\support{\calS}$ and $\ms t_i\in\support{\ms{\calT}_i}$ for $i\in\set{0,\dotsc,n}$.
	Write $v=\rappl {\lsubst {s}x{\ms{t}_0}}{{\ms{t}_1}\cdots{\ms{t}_n}}$:
	Corollary \ref{corollary:resRed:strategies} entails 
	$v\generates u'$, hence
	we have $v\in\pars{\rappl {\lsubst {\calS}x{\ms{\calT}_0}}{{\ms{\calT}_1}\cdots{\ms{\calT}_n}}}\inter\cone {u'}$.
	By assumption, there are finitely many possible values for $v$.
	Then, $v$ being fixed, by Lemma \ref{lemma:reduction:size},
	we have $\fv{u}=\fv{v}$ and $\size{u}\le 2\size{v}+2$, hence
	there are finitely many possible values for $u$.
\end{proof}

\begin{theorem}
	\label{theorem:converges:normalizes}
	If $M$ is normalizable, then $\TaylorSup M\in\normTermStructure$, 
	and $\TaylorExp M\in\normTermVectors$.
\end{theorem}
\begin{proof}
	By induction on the derivation of $\converges M$:
	Lemma \ref{lemma:normalizable:lambda},
	Corollary \ref{corollary:normalizable:neutral} and
	Lemma \ref{lemma:normalizable:head}
	respectively entail the translation 
	of the first three inductive rules through Taylor expansion.
	The other three follow from the fact that $\normTermStructure$ 
	is a resource structure (because it is a finiteness structure).
\end{proof}

It remains to prove that in this case, $\TaylorExp{\NormalForm M}$ is indeed
the normal form of $\TaylorExp M$.

\begin{theorem}
	\label{theorem:NormalForm:TaylorExp}
	If $M$ is normalizable, then $\NormalForm{\TaylorExp M}=\TaylorExp{\NormalForm M}$.
\end{theorem}
\begin{proof}
	By Theorem \ref{theorem:converges:normalizes}, $M$ is Taylor normalizable.
	Then Theorem \ref{theorem:betaEq:NormalForm} entails
	$\NormalForm{\TaylorExp M}=\NormalForm{\TaylorExp{\NormalForm M}}
	=\TaylorExp{\NormalForm M}$.
\end{proof}

\subsection{Conservativity}

\label{subsection:conservativity}

The restriction to normalizable vectors allows us to prove 
an analogue of Lemma \ref{lemma:resEq:conservative:zerosumfree},
without any assumption on the semiring of scalars.

\begin{lemma}
	\label{lemma:resEq:conservative:normalizable}
	Let $M,N\in\lambdaTerms$ be normalizable.
	Then $M\resEq[\normTermStructure] N$ iff $M\betaEq N$.
\end{lemma}
\begin{proof}
	Assume $M\resEq[\normTermStructure] N$.
	By Corollary \ref{corollary:resEq:normalStructure:NormalForm},
	we have $\NormalForm{\TaylorExp{M}}=\NormalForm{\TaylorExp{M'}}$.
	By Theorem \ref{theorem:NormalForm:TaylorExp}, 
	we obtain $\NormalForm M\teq \NormalForm{M'}$.
	Since $M$ and $N$ are pure λ-terms, 
	we deduce $\NormalForm M=\NormalForm N$ 
	from the injectivity of $\TaylorExpSym$ on $\lambdaTerms$.

	The reverse direction is similar to Theorem
	\ref{theorem:betaEq:NormalForm}
	and does not depend on $M$ and $N$ being pure λ-terms:
	apply Lemmas \ref{lemma:vpresRed:correct},
	\ref{lemma:vpresRed:NormalForm} and
	\ref{lemma:betaRed:Taylor:normalizable}
	to the reduction path from $M$ to $N$
\end{proof}

We can adapt this result to non-normalizing pure λ-terms
thanks to previous work by Ehrhard and Regnier:\footnote{
	We could as well rely on Theorem \ref{theorem:taylor:determinable},
	to be proved in the next section.
}

\begin{thmC}[\cite{er:resource,er:bkt}]
	\label{theorem:taylor:pure}
	For all pure λ-term $M\in\lambdaTerms$, 
	$\TaylorSup M\in\normTermStructure$
	and $\NormalForm{\TaylorExp M}=\TaylorExp{\BohmTree M}$
	where $\BohmTree M$ denotes the Böhm tree of $M$.
\end{thmC}

Here \emph{Böhm tree} is to be understood as \emph{generalized normal form for left β-reduction}.
In particular it does not involve η-expansion. More formally,
the Böhm tree of a λ-term is the possibly infinite tree obtained 
coinductively as follows:
\begin{itemize}
	\item if $M$ is head normalizable and its head normal form is
		$\labs{x_1}{\cdots\labs{x_n}{\appl x{N_1\cdots N_k}}}$
		then $\BohmTree M\eqdef
		\labs{x_1}{\cdots\labs{x_n}{\appl x{\BohmTree{N_1}\cdots\BohmTree{N_k}}}}$
	\item otherwise $\BohmTree M\eqdef\bot$, where $\bot$ is a constant 
		representing unsolvability.
\end{itemize}
Taylor expansion can be generalized to Böhm trees \cite{er:bkt}, setting in
particular $\TaylorExp \bot =0$: this is still injective.

\begin{lemma}
	If $M,N\in\lambdaTerms$ and $M\resEq[\normTermStructure] N$
	then $\BohmTree M=\BohmTree N$.
\end{lemma}
\begin{proof}
	By Corollary \ref{corollary:resEq:normalStructure:NormalForm},
	we have $\NormalForm{\TaylorExp{M}}=\NormalForm{\TaylorExp{M'}}$.
	By Theorem \ref{theorem:taylor:pure}, 
	we obtain $\TaylorExp{\BohmTree M}=\TaylorExp{\BohmTree N}$.
	We conclude since $\TaylorExpSym$ is injective on Böhm trees.
\end{proof}

In the next and final section, we prove a generalization of
Theorem \ref{theorem:taylor:pure} to 
the non-uniform setting
which is made possible by the results
we have achieved so far.

\section{Normal form of Taylor expansion, \emph{façon} Böhm trees}

\label{section:determinable}

The Böhm tree construction is often introduced as the limit of an increasing sequence
$(\kBohmTree dM)_{d\in\naturals}$ of finite normal form approximants,
\emph{aka} finite Böhm trees, where $\kBohmTree dM$ is defined inductively as
follows:
\begin{itemize}
	\item $\kBohmTree 0M=\bot$;
	\item if $M$ is head normalizable and its head normal form is
		$\labs{x_1}{\cdots\labs{x_n}{\appl x{N_1\cdots N_k}}}$
		then $\kBohmTree {d+1}M\eqdef
		\labs{x_1}{\cdots\labs{x_n}{\appl x{\kBohmTree d{N_1}\cdots\kBohmTree d{N_k}}}}$
	\item otherwise $\kBohmTree {d+1}M\eqdef\bot$;
\end{itemize}
and the order on Böhm trees is the contextual closure of the inequality $\bot \le M$ for all $M$.

In this final section of our paper, we show that the \emph{normal form of
Taylor expansion} operator generalizes this construction to the class
of \emph{hereditarily determinable terms}: these encompass both all pure λ-terms and
all normalizable algebraic λ-terms, but exclude terms such as $\infty_x$, that
produce unbounded sums of head normal forms. More precisely, we show that any 
hereditarily determinable term $M$ is Taylor normalizable, and moreover admits
a sequence of approximants $\pars{\kApprox dM}_{d\in\naturals}$, such that 
each $\kApprox dM$ is an algebraic λ-term in normal form, and the
sequence of normal term vectors $\pars{\TaylorExp{\kApprox dM}}_{d\in\naturals}$
converges to $\NormalForm{\TaylorExp M}$.

The results in this section should not hide the fact that the more fundamental
notion is that of Taylor normalizable term, which arises naturally by combining
Taylor expansion with the normalization of resource terms, subject to a
summability condition. We believe this approach is quite robust, and may be
adapted modularly following both parameters: to other systems admitting Taylor
expansion; and to variants of summability, possibly associated with topological
conditions of the semiring of scalars.

By contrast, the definition of hereditarily determinable terms is essentially
\emph{ad-hoc}.
Its only purpose is to allow us to generalize Theorem
\ref{theorem:taylor:pure} and support our claim that:
\emph{the normal form of Taylor expansion extends the notion of Böhm tree
to the non-uniform setting}.

\subsection{Taylor unsolvability}

\label{subsection:unsolvable}

In the ordinary λ-calculus, head normalizable terms are exactly those with a
non trivial Böhm tree. This is reflected via Taylor expansion: it is easy to check that 
$\NormalForm{\TaylorExp M}=0$ iff $M$ has no head normal form. 
In the non uniform setting, a similar result holds, although
we need to be more careful about the interplay between reduction and coefficients.

\begin{definition}
	We say an algebraic λ-term $M$ (resp.\ simple term $S$) 
	is \emph{weakly solvable} if the judgement $\weakSolvable M$
	can be derived inductively by the following rules:
	\begin{center}
		\begin{prooftree}
			\Infer0{\weakSolvable {\appl{x}{M_1\cdots M_n}}}
		\end{prooftree}
		\hfill
		\begin{prooftree}
			\Hypo  {\weakSolvable S}
			\Infer1{\weakSolvable{\labs xS}}
		\end{prooftree}
		\hfill
		\begin{prooftree}
			\Hypo  {\weakSolvable{\appl{\subst Sx{M_0}}{M_1\cdots M_n}}}
			\Infer1{\weakSolvable{\appl{\labs xS}{M_0\,M_1\cdots M_n}}}
		\end{prooftree}
		\hfill
		\begin{prooftree}
			\Hypo  {\weakSolvable M}
			\Infer1{\weakSolvable {a\sm M}}
		\end{prooftree}
		\hfill
		\begin{prooftree}
			\Hypo  {\weakSolvable M}
			\Infer1{\weakSolvable {M+N}}
		\end{prooftree}
		\hfill
		\begin{prooftree}
			\Hypo  {\weakSolvable N}
			\Infer1{\weakSolvable {M+N}}
		\end{prooftree}
	\end{center}
\end{definition}

It should be clear that, if $M$ is a pure λ-term, $\weakSolvable M$ iff $M$ is
head normalizable. In the general case, we show that $\weakSolvable M$
iff normalizing the Taylor expansion of $M$ yields a non trivial result.
More formally:

\begin{definition}
We say an algebraic λ-term $M\in\algebraicTerms$ is \emph{Taylor unsolvable}
and write $\unsolvable M$ if $\NormalForm s=0$ for all $s\in\TaylorSup M$.
\end{definition}

In particular, if $\unsolvable M$ then $\TaylorExp M\in\normTermVectors$ 
and $\NormalForm{\TaylorExp M}=0$: indeed, 
$\support{\TaylorExp M}\subseteq\TaylorSup M$.
Beware that the reverse implication does not hold in general.
We can then show that $\weakSolvable M$ iff $M$ is Taylor solvable
(Lemmas \ref{lemma:Taylor:solvable:weakSolvable} and \ref{lemma:weakSolvable:Taylor:solvable}).

\begin{lemma}
	\label{lemma:Taylor:solvable:weakSolvable}
	If there exists $s\in\TaylorSup M$ such that $\NormalForm s\not=0$
	then $\weakSolvable M$.
\end{lemma}
\begin{proof}
	We prove by induction on $k\in\naturals$ then on $M\in\algebraicTerms$
	that if $\support{\nthLeftReduct ks}$ contains a normal resource term
	and $s\in\TaylorSup M$ then $\weakSolvable M$.

	If $M=\appl{x}{M_1\cdots M_n}$ we conclude directly.

	If $M=\labs xT$ then $s=\labs xt$ with $t\in\TaylorSup T$:
	necessarily $\support{\nthLeftReduct kt}$ contains a normal resource term
	and by induction hypothesis we obtain $\weakSolvable T$ hence 
	$\weakSolvable M$.

	If $M=\appl{\labs xT}{M_0\,M_1\cdots M_n}$ 
	then $s=\rappl{\labs xt}{\ms s_0\,\ms s_1\cdots\ms s_n}$ 
	with $t\in\TaylorSup T$ and $\ms s_i\in\prom{\TaylorSup{M_i}}$ for 
	$i\in\set{0,\dotsc,n}$. Necessarily $k>0$ and there is 
	$s'\in\support{\leftReduct s}=\support{\rappl{\lsubst tx{\ms s_0}}{\ms s_1\cdots\ms s_n}}$
	such that $\support{\nthLeftReduct{k-1}{s'}}$
	contains a normal resource term. By Lemma \ref{lemma:taylorsup:subst},
	$s'\in\TaylorSup{\appl{\subst Tx{M_0}}{M_1\cdots M_n}}$:
	we obtain $\weakSolvable{\appl{\subst Tx{M_0}}{M_1\cdots M_n}}$
	by induction hypothesis, and then $\weakSolvable M$.

	If $M=a\sm N$, $M=N+P$ or $M=P+N$ with $s\in\TaylorSup N$
	then we obtain $\weakSolvable N$ by induction hypothesis,
	and then $\weakSolvable M$.
\end{proof}

\begin{lemma}
	\label{lemma:weakSolvable:Taylor:solvable}
	If $\weakSolvable M$, then there exists $s\in\TaylorSup M$ such that
	$\NormalForm s\not=0$.
\end{lemma}
\begin{proof}
	By induction on the derivation of $\weakSolvable M$.

	If $M=\appl{x}{M_1\cdots M_n}$,
	set $s=\rappl x{\mset{}\cdots\mset{}}$ ($x$ applied $n$ times 
	to the empty monomial): $s\in\TaylorSup M$ and $s$ is normal.

	If $M=\labs xT$ with $\weakSolvable T$: by induction hypothesis, we obtain 
	$t\in\TaylorSup T$ with $\NormalForm t\not=0$ and set $s=\labs xt$.

	If $M=\labs xT{M_0\,M_1\cdots M_n}$ and $M'=\appl{\subst Tx{M_0}}{M_1\cdots M_n}$
	with $\weakSolvable {M'}$, the induction hypothesis gives 
	$s'\in\TaylorSup {M'}$ such that $\NormalForm{s'}\not=0$.
	By Lemma \ref{lemma:taylorsup:subst}, there exist $t\in\TaylorSup T$
	and $\ms{u}_i\in\prom{\TaylorSup{M_i}}$ for $i\in\set{0,\dotsc,n}$ such that
	$s'\in\support{\rappl{\lsubst{t}x{\ms u_0}}{\ms u_1\cdots\ms u_n}}$.
	We then set $s=\rappl{\labs x{t}}{\ms u_0\cdots\ms u_n}$.

	If $M=a\sm N$, $M=N+P$ or $M=P+N$ with $\weakSolvable N$: the induction hypothesis
	gives $s\in\TaylorSup N\subseteq\TaylorSup M$ with $\NormalForm s\not=0$
	directly.
\end{proof}

Taylor unsolvable terms are thus exactly those that are not weakly solvable.\footnote{
	If we restrict to non-deterministic λ-terms 
	(\ie{} only add a sum operator to the usual λ-term constructs)
	then we obtain $\unsolvable M$ iff $\NormalForm{\TaylorSup M}=\emptyset$,
	which states the adequacy of $\NormalForm{\TaylorSup{\cdot}}$ for 
	the observational equivalence associated with may-style head normalization.
}
They are moreover stable under $\betaEq$:

\begin{lemma}
	\label{lemma:unsolvable:pbetaRed}
	If $M\pbetaRed M'$ then $\unsolvable M$ iff $\unsolvable{M'}$.
\end{lemma}
\begin{proof}
	If $\calE\subseteq\resourceExpressions$,
	we write $\NormalForm{\calE}\eqdef\Union{\set{\support{\NormalForm e}\st e\in\calE}}$.
	We leave as an exercise to the reader the proof that 
	$\NormalForm{\TaylorSup{M}}=\NormalForm{\TaylorSup{M'}}$
	as soon as $M\pbetaRed M'$: this is the analogue of Lemma
	\ref{lemma:vpresRed:NormalForm} on Taylor supports 
	(in particular there is no summability condition, and scalars 
	play absolutely no rôle).
\end{proof}

\subsection{Hereditarily determinable terms}

The Böhm tree construction is based on the fact that, for a pure λ-term $M$,
either $M$ is unsolvable, or it reduces to a head normal form; and then the
same holds for the arguments of the head variable.
We will be able to follow a similar construction for the class of \emph{hereditarily
determinable terms}: intuitively, a simple term is in \emph{determinate form} if it is
either unsolvable or a head normal form; and a term is hereditarily
determinable if it reduces to a sum of determinate forms, and this holds 
hereditarily in the arguments of head variables. Formally:

\begin{definition}
	Let $M\in\algebraicTerms$ be an algebraic λ-term.
	We say $M$ is \emph{$d$-determinable} if the judgement $\kConverges d M$
	can be derived inductively from the following rules:
	\begin{center}
		\begin{prooftree}
			\Infer0{\kConverges 0 M}
		\end{prooftree}
		\hfill
		\begin{prooftree}
			\Hypo  {\unsolvable M}
			\Infer1{\kConverges d {M}}
		\end{prooftree}
		\hfill
		\begin{prooftree}
			\Hypo  {\kConverges d S}
			\Infer1{\kConverges d{\labs xS}}
		\end{prooftree}
		\hfill
		\begin{prooftree}
			\Hypo  {\kConverges d{M_1}}
			\Hypo  {\cdots}
			\Hypo  {\kConverges d{M_n}}
			\Infer3{\kConverges {d+1}{\appl{x}{M_1\cdots M_n}}}
		\end{prooftree}
		\hfill
		\begin{prooftree}
			\Hypo  {\kConverges d M}
			\Infer1{\kConverges d {a\sm M}}
		\end{prooftree}
		\hfill
		\begin{prooftree}
			\Hypo  {\kConverges d M}
			\Hypo  {\kConverges d N}
			\Infer2{\kConverges d {M+N}}
		\end{prooftree}
		\hfill
		\begin{prooftree}
			\Hypo  {\kConverges d{\appl{\subst Sx{M_0}}{M_1\cdots M_n}}}
			\Infer1{\kConverges d{\appl{\labs xS}{M_0\,M_1\cdots M_n}}}
		\end{prooftree}
	\end{center}

	We say $M$ is \emph{hereditarily determinable} and write $\wConverges M$
	if $\kConverges d M$ for all $d\in\naturals$.
	We say $M$ is in \emph{$d$-determinate form} and write $\kDetermined d M$
	if $\kConverges dM$ is derivable from the above rules excluding the 
	last one.
\end{definition}
It should be clear that $\converges M$ implies $\wConverges M$.
Observing that $\unsolvable M$ for all unsolvable pure λ-terms (\ie
those pure λ-terms having no head normal form), we moreover obtain
$\wConverges M$ for all $M\in\lambdaTerms$.

We can already prove that hereditarily determinable terms are Taylor normalizable:\footnote{
	Observe that this fails if we replace $\TaylorSup M$ with
	$\support{\TaylorExp{M}}$ in the definition of $\unsolvable M$:
	write $I\eqdef \labs xx$ and 
	consider, \eg, $M=\appl{\labs x{\appl I{\pars{x+(-1)\sm\infty_y}}}}{\infty_y}$
	which head-reduces to $\appl I{\pars{\infty_y+(-1)\sm\infty_y}}\teq\appl I0$,
	with $\unsolvable{\appl I0}$ but of course $\TaylorExp M\not\in\kNormTermStructure 0$.
	The very same problem would occur if we were to consider terms up to $\algEq$.
}
\begin{lemma}
	\label{lemma:kConverges:kNorm}
	If $\kConverges dM$ then $\TaylorSup M\in\kNormTermStructure d$.
	If moreover $\wConverges M$ then $\TaylorSup M\in\normTermStructure$.
\end{lemma}
\begin{proof}
	The second fact follows directly from the first one, 
	which we prove by induction on the derivation of $\kConverges dM$: 
	we use the definition of $\unsolvable M$ for the base case, 
	and rely on Lemma \ref{lemma:normalizable:lambda},
	Corollary \ref{corollary:normalizable:neutral},
	Lemma \ref{lemma:normalizable:head},
	or the fact that $\kNormTermStructure d$ is a resource structure
	to establish the induction in the other cases.
\end{proof}

On the other hand, there are Taylor normalizable terms that do not follow this
pattern: intuitively, hereditarily determinable terms rule out any
representation of an infinite sum of head normal forms, 
whereas Taylor normalizability allows to represent an infinite sum of normal forms
as long as their Taylor expansions are pairwise disjoint.
More formally:
\begin{example}
	Write $s_0\eqdef \labs xx$, 
	and $s_{n+1}\eqdef \labs x{s_n}$.
	Let $M_{step}=\labs y{\labs z z}
	+\labs y{\labs z{\labs x{\appl y{y\,z}}}}$ and then
	$M_{loop}=\appl{M_{step}}{M_{step}\,\labs xx}$.
	Write $u=\labs y{\labs z z}$ and
	$v_{n,k}=\labs y{\labs z{\labs x{\rappl y{y^n\,z^k}}}}$ so that 
	$\TaylorSup{M_{step}}=\set{u}\union\set{v_{n,k}\st n,k\in\naturals}$.
	Let $s\in\TaylorSup{M_{loop}}$ be such that $\NormalForm s\not=0$:
	a simple inspection shows that either $s=\rappl u{\mset{}\,\mset{s_0}}$
	and then $\NormalForm s=s_0$,
	or $s=\rappl{v_{n,1}}{\mset{v_{0,1},\dotsc,v_{n-1,1},u}\,\mset{s_0}}$
	and then $\NormalForm s=s_{n+1}$.
	It follows that $M_{loop}$ is Taylor normalizable.
	On the other hand, observe that $\nthLeftReduct 2{M_{loop}}=
	\labs xx+\labs x{M_{loop}}$, which is not $1$-determinate:
	hence no $\nthLeftReduct{2k}{M_{loop}}$ is $1$-determinate and 
	it will follow from Lemma \ref{lemma:kConverges:leftReduct:kDetermined}
	that $M_{loop}$ is not $1$-determinable.
\end{example}
Hence hereditarily determinable terms form a strict subclass of Taylor
normalizable terms, containing both pure λ-terms and normalizable algebraic
λ-terms.
For each level $d\in\naturals$, the class of $d$-determinable terms (resp.\ of
$d$-determinate terms) is moreover stable under left reduction:
\begin{lemma}
	\label{lemma:kConverges:leftReduct}
	If $\kConverges dM$ (resp.\ $\kDetermined dM$)
	then $\kConverges d{\leftReduct M}$ (resp.\ $\kDetermined d{\leftReduct M}$).
\end{lemma}
\begin{proof}
	We give the proof for $d$-determinable terms,
	by induction on the derivation of $\kConverges dM$: 
	the case of $d$-determinate terms is similar, except 
	that we do not consider head redexes.
	
	If $d=0$ the result is direct. Otherwise, write $d=d'+1$.

	If $\unsolvable M$ then $\unsolvable{\leftReduct M}$
	by Lemma \ref{lemma:unsolvable:pbetaRed},
	and we conclude directly.

	If $M=\labs xS$ with $\kConverges dS$: by induction hypothesis
	$\kConverges d{\leftReduct  S}$, and then
	$\kConverges d{\labs x{\leftReduct S}}$.

	If $M=\appl {x}{ M_1\cdots M_n}$ 
	with $\kConverges{d'}{M_i}$ for $i\in\set{1,\dotsc,n}$:
	by induction hypothesis $\kConverges{d'}{\leftReduct{M_i}}$
	for $i\in\set{1,\dotsc,n}$, and then
	$\kConverges{d}{\appl{x}{{\leftReduct {M_1}}\cdots {\leftReduct{M_n}}}}$.

	If $M=\appl{\labs xS}{M_0\,M_1\cdots M_n}$
	with $\kConverges d{\appl {\subst Sx{M_0}}{ M_1\cdots M_n}}$
	then we conclude directly since
	$\leftReduct M=\appl {\subst Sx{M_0}}{ M_1\cdots M_n}$.

	If $M=a\sm N$ with $\kConverges dN$: by induction hypothesis
	$\kConverges d{\leftReduct  N}$, and then 
	$\kConverges d{a\sm\leftReduct  N}$.

	If $M=N+P$ with $\kConverges dN$ and $\kConverges dP$: by induction hypothesis
	$\kConverges d{\leftReduct  N}$ and $\kConverges d{\leftReduct  P}$,
	and then $\kConverges d{\leftReduct  N+\leftReduct  P}$.
\end{proof}

Now we can formally prove that applying the parallel left reduction strategy to
$d$-determinable terms does reach $d$-determinate forms.
\begin{lemma}
	\label{lemma:kConverges:leftReduct:kDetermined}
	If $\kConverges dM$ then there exists $k\in\naturals$
	such that $\kDetermined d{\nthLeftReduct k M}$.
\end{lemma}
\begin{proof}
	By induction on the derivation of $\kConverges dM$.

	If $d=0$ or $\unsolvable M$, then $\kDetermined dM$.
	
	If $M=\labs x S$ with $\kConverges{d} S$:
	by induction hypothesis, we have $k\in\naturals$ such that 
	$\kDetermined {d}{\nthLeftReduct k{S}}$
	and then
	$
		\nthLeftReduct k{M}
		=\labs x{\nthLeftReduct k{S}}
	$ hence $\kDetermined {d}{\nthLeftReduct k{M}}$.

	If $M=\appl{x}{M_1\cdots M_n}$ with $d>0$ and $\kConverges{d-1} {M_i}$ for each $i\in\set{1,\dotsc,n}$:
	by induction hypothesis, we obtain $k_i\in\naturals$ such that 
	$\kDetermined{d-1}{\nthLeftReduct{k_i}{M_i}}$
	for each $i\in\set{1,\dotsc,n}$. Let $k=\max\set{k_i\st 1\le i\le n}$:
	by Lemma \ref{lemma:kConverges:leftReduct}, we also have
	$\kDetermined{d-1}{\nthLeftReduct{k}{M_i}}$ for all $i\in\set{1,\dotsc,n}$.
	Since
	$\nthLeftReduct k{M}
	=\nthLeftReduct k{\appl x{M_1\cdots M_n}}
	=\appl x{\nthLeftReduct k{M_1}\cdots\nthLeftReduct k{M_n}}
	$ we conclude that 
	$\kDetermined{d}{\nthLeftReduct k{M}}$.

	If $M=\appl{\labs x{S}}{M_0\,M_1\cdots M_n}$ 
	with $\kConverges{d} {\appl{\subst Sx{M_0}}{M_1\cdots M_n}}$:
	by induction hypothesis, we have $k_0\in\naturals$ such that
	$\kDetermined{d}{\nthLeftReduct{k_0}{\appl{\subst Sx{M_0}}{M_1\cdots M_n}}}$.
	It is then sufficient to observe that 
	$\leftReduct{M} =\appl{\subst Sx{M_0}}{M_1\cdots M_n}$
	and set $k=k_0+1$.

	If $M=a\sm N$ with $\kConverges d N$:
	by induction hypothesis, we have $k\in\naturals$ such that 
	$\kDetermined {d}{\nthLeftReduct k{S}}$
	and then
	$
		\nthLeftReduct k{M}
		=a\sm{\nthLeftReduct k{S}}
	$ hence $\kDetermined {d}{\nthLeftReduct k{M}}$.

	If $M=N+P$ with $\kConverges d N$ and $\kConverges d P$:
	by induction hypothesis, we have $k_0,k_1\in\naturals$ such that 
	$\kDetermined {d}{\nthLeftReduct {k_0}{N}}$
	and
	$\kDetermined {d}{\nthLeftReduct {k_1}{P}}$
	and then, setting $k=\max(k_0,k_1)$,
	$
		\nthLeftReduct k{M}
		=\nthLeftReduct k{N}+\nthLeftReduct k{P}
	$ hence $\kDetermined {d}{\nthLeftReduct k{M}}$ by the previous lemma.
\end{proof}

\subsection{Approximants of the normal form of Taylor expansion}

Now we introduce the analogue of finite Böhm trees for hereditarily
determinable terms:
\begin{definition}
	If $\kConverges d M$ then we define the 
	\emph{normal $d$-approximant} $\kApprox dM$ of $M$ inductively as follows:
	$\kApprox{d}M\eqdef 0$ if $d=0$ or $\unsolvable M$, and 
	\begin{align*}
		\kApprox{d} {\labs xS} &\eqdef \labs x{\kApprox{d} S}
		\\
		\kApprox{d} {\appl {x}{ M_1\cdots M_n}} &\eqdef 
		\appl {x}{\kApprox{d-1}{ M_1}\cdots\kApprox{d-1}{ M_n}}
		\\
		\kApprox{d} {\appl {\labs xS}{ M_0\, M_1\cdots M_n}} &\eqdef 
		\kApprox{d} {\appl {\subst Sx{M_0}}{ M_1\cdots M_n}} 
		\\
		\kApprox{d} {a\sm M} &\eqdef a\sm\kApprox{d} M
		\\
		\kApprox{d} {M+N} &\eqdef \kApprox{d} M+\kApprox{d} N
	\end{align*}
	otherwise.
\end{definition}

First observe that $d$ approximants are stable under parallel left reduction:
\begin{lemma}
	\label{lemma:kApprox:leftReduct}
	If $\kConverges dM$ then $\kApprox dM=\kApprox d{\leftReduct M}$.
\end{lemma}
\begin{proof}
	Recall indeed that, by Lemma \ref{lemma:kConverges:leftReduct},
	$\kConverges d{\leftReduct M}$ so that $\kApprox d{\leftReduct M}$
	is well defined.
	The proof is then straightforward, by induction on $\kConverges dM$.
\end{proof}

We do not prove here that $d$-determinable terms and the associated
$d$-approximants are stable under arbitrary reduction:
if $\kConverges dM$ and $M\pbetaRed M'$ then $\kConverges d{M'}$
and then $\kApprox dM=\kApprox d{M'}$.
We believe it is a very solid conjecture, but it would require us to develop
a full standardization argument: in our non-deterministic setting, this is
known to be tedious at best \cite{alberti:phd,leventis:phd}.
Since we introduced hereditarily determinable terms \emph{ad-hoc}, only to be
able to define normal $d$-approximants, we feel that the general 
study of their computational behaviour is not worth the effort.

Our next step is to show that if $M$ is in $d+1$ determinate form, 
then $\kNormTermRestr d{\TaylorExp M}$ depends only on $\kApprox{d+1}M$.
\begin{lemma}
	\label{lemma:kApprox:TaylorExp}
	If $\kDetermined {d+1}M$
	then, for all $s\in\kNormTerms d$,  $\TaylorExp M_s=\TaylorExp{\kApprox{d+1}M}_s$.
\end{lemma}
\begin{proof}
	By induction on the derivation of $\kDetermined{d+1}M$,
	writing $M'=\kApprox{d+1}M$.

	If $\unsolvable M$ then $M'=0$ and
	$\TaylorExp M_s=0$ for all $s\in\normTerms$,
	hence the result holds.

	If $M=\labs xT$ with $\kDetermined{d+1}T$ then 
	$M'=\labs x{\kApprox{d+1}T}$
	and we can assume $s=\labs xt$:
	otherwise $\TaylorExp M_s=0=\TaylorExp{M'}_s$.
	Then $t\in\kNormTerms d$ and by induction hypothesis
	$\TaylorExp M_s
	=\TaylorExp T_t
	=\TaylorExp{\kApprox{d+1}T}_t
	=\TaylorExp{M'}_s$.

	If $M=\appl x{N_1\cdots N_n}$ with $\kDetermined{d}{N_i}$ 
	for all $i\in\set{1,\dotsc,n}$ then 
	$M'=\appl x{N'_1\cdots N'_n}$ with $N'_i=\kApprox{d}{N_i}$
	and we can assume $s=\rappl{x}{\ms t_1\cdots\ms t_n}$:
	otherwise $\TaylorExp M_s=0=\TaylorExp{M'}_s$.
	If $d=0$, $s\in\kNormTerms 0$, hence $n=0$ and then $M=x=M'$.
	Otherwise write $d=d'+1$. For each $i\in\set{1,\dotsc,n}$,
	$\support{\ms t_i}\subseteq\kNormTerms{d'}$.
	By induction hypothesis we obtain 
	$\TaylorExp{N_i}_{u}=\TaylorExp{N'_i}_{u}$
	for all $u\in\support{\ms t_i}$: it follows that 
	$\prom{\TaylorExp{N_i}}_{\ms t_i}=\prom{\TaylorExp{N'_i}}_{\ms t_i}$
	by the definition of promotion. Then 
	$\TaylorExp{M}_s
	=\prod_{i=1}^n\TaylorExp{N_i}_{\ms t_i}
	=\prod_{i=1}^n\TaylorExp{N'_i}_{\ms t_i}
	=\TaylorExp{M'}_s
	$.

	If $M=a\sm N$ with $\kDetermined{d+1}N$ then 
	$\TaylorExp M_s
	=a\sm\TaylorExp N_s
	=a\sm\TaylorExp{\kApprox{d+1}N}_s
	=\TaylorExp{M'}_s$ by induction hypothesis.

	Similarly, if $M=N+P$ with $\kDetermined{d+1}N$ and $\kDetermined{d+1}P$ then
	$\TaylorExp M_s
	=\TaylorExp N_s+\TaylorExp P_s
	=\TaylorExp{\kApprox{d+1}N}_s+\TaylorExp{\kApprox{d+1}P}_s
	=\TaylorExp{M'}_s$
	by induction hypothesis.
\end{proof}

We obtain our final theorem:
\begin{theorem}
	\label{theorem:taylor:determinable}
	For all hereditarily determinable term $M$,
	the sequence
	$\pars{\TaylorExp{\kApprox dM}}_{d\in\naturals}$
	of normal vectors 
	converges to $\NormalForm{\TaylorExp{M}}$ in $\vectors{\normTerms}$.
\end{theorem}
\begin{proof}
	First observe that each
	$\TaylorExp{\kApprox dM}\in\vectors{\normTerms}$, because 
	$\kApprox dM$ is in normal form. Let $s\in\normTerms$
	and fix $d\ge\depth{s}+1$:
	by Lemmas \ref{lemma:kConverges:leftReduct},
	\ref{lemma:kConverges:leftReduct:kDetermined} and 
	\ref{lemma:kApprox:leftReduct},
	there exists $k_0\in\naturals$ such that 
	$\kDetermined {d}{\nthLeftReduct kM}$ and $\kApprox{d}{\nthLeftReduct kM}=\kApprox{d}{M}$
	whenever $k\ge k_0$.
	By Lemma \ref{lemma:kApprox:TaylorExp}, we moreover have
	$\TaylorExp{\kApprox{d}{M}}_s
	=\TaylorExp{\kApprox{d}{\nthLeftReduct kM}}_s
	=\TaylorExp{\nthLeftReduct kM}_s$.
	It follows that $\TaylorExp{\kApprox{d}{M}}_s=\NormalForm{\TaylorExp{M}}_s$,
	by Theorem \ref{theorem:Taylor:normalizable:limit}.
	Since this holds for any $d\ge\depth{s}+1$, we have just proved that 
	$\pars{\TaylorExp{\kApprox{d}{M}}_s}_{d\in\naturals}$ converges to $\NormalForm{\TaylorExp{M}}_s$,
	for the discrete topology.
\end{proof}

In the case of pure λ-terms, by identifying $0$ with the unsolvable Böhm tree
$\bot$, it should be clear that the sequence $\pars{\kApprox
dM}_{d\in\naturals}$ is nothing but the increasing sequence of finite
approximants of $\BohmTree M$: Theorem \ref{theorem:taylor:determinable} is thus 
a proper generalization of Theorem \ref{theorem:taylor:pure}
of which it provides a new proof.

\bibliographystyle{alpha}
\bibliography{bib.bib}

\begin{thebibliography}{KKSdV97}

\bibitem[AD08]{ad:lineal}
Pablo Arrighi and Gilles Dowek.
\newblock Linear-algebraic lambda-calculus: higher-order, encodings, and
  confluence.
\newblock In Andrei Voronkov, editor, {\em Rewriting Techniques and
  Applications, 19th International Conference, {RTA} 2008, Hagenberg, Austria,
  July 15-17, 2008, Proceedings}, volume 5117 of {\em Lecture Notes in Computer
  Science}, pages 17--31. Springer, 2008.

\bibitem[Alb14]{alberti:phd}
Michele Alberti.
\newblock {\em {On operational properties of quantitative extensions of
  $\lambda$-calculus}}.
\newblock PhD thesis, {Aix Marseille Universit{\'e} ; Universit{\`a} di
  Bologna}, December 2014.

\bibitem[BEM07]{bem:enough}
Antonio Bucciarelli, Thomas Ehrhard, and Giulio Manzonetto.
\newblock Not enough points is enough.
\newblock In {\em Computer Science Logic}, volume 4646 of {\em Lecture Notes in
  Computer Science}, pages 298--312. Springer Berlin, 2007.

\bibitem[Bou93]{boudol:resource}
G\'{e}rard Boudol.
\newblock The lambda-calculus with multiplicities (abstract).
\newblock In {\em CONCUR '93: Proceedings of the 4th International Conference
  on Concurrency Theory}, pages 1--6, London, UK, 1993. Springer-Verlag.

\bibitem[CA18]{cv:antireduits-csl}
Jules Chouquet and Lionel~Vaux Auclair.
\newblock An application of parallel cut elimination in unit-free
  multiplicative linear logic to the taylor expansion of proof nets.
\newblock In Dan~R. Ghica and Achim Jung, editors, {\em 27th {EACSL} Annual
  Conference on Computer Science Logic, {CSL} 2018, September 4-7, 2018,
  Birmingham, {UK}}, volume 119 of {\em LIPIcs}, pages 15:1--15:17. Schloss
  Dagstuhl - Leibniz-Zentrum fuer Informatik, 2018.

\bibitem[Car11]{carraro:phd}
Alberto Carraro.
\newblock {\em Models and theories of pure and resource lambda calculas}.
\newblock PhD thesis, 2011.
\newblock Thèse de doctorat dirigée par Salibra, Antonino et Bucciarelli,
  Antonio Informatique Paris 7 2011.

\bibitem[CES10]{ces:multiplicities}
Alberto Carraro, Thomas Ehrhard, and Antonino Salibra.
\newblock Exponentials with infinite multiplicities.
\newblock In Anuj Dawar and Helmut Veith, editors, {\em Computer Science Logic,
  24th International Workshop, {CSL} 2010, 19th Annual Conference of the EACSL,
  Brno, Czech Republic, August 23-27, 2010. Proceedings}, volume 6247 of {\em
  Lecture Notes in Computer Science}, pages 170--184. Springer, 2010.

\bibitem[dC08]{carvalho:exec}
Daniel de~Carvalho.
\newblock Execution time of lambda-terms via denotational semantics and
  intersection types.
\newblock Technical report, 2008.
\newblock Rapport de recherche INRIA n° 6638.

\bibitem[DC11]{diazcaro:phd}
Alejandro D{\'\i}az-Caro.
\newblock {\em Du typage vectoriel}.
\newblock PhD thesis, Universit\'e de Grenoble, France, September~23, 2011.

\bibitem[DE11]{de:pcoh}
Vincent Danos and Thomas Ehrhard.
\newblock Probabilistic coherence spaces as a model of higher-order
  probabilistic computation.
\newblock {\em Information and Computation}, 2011.

\bibitem[DLL19]{dll:pbt}
Ugo Dal~Lago and Thomas Leventis.
\newblock On the {T}aylor expansion of probabilistic λ-terms.
\newblock Submitted, 2019.

\bibitem[{\relax DLMF}]{NIST:DLMF}
{NIST Digital Library of Mathematical Functions}.
\newblock http://dlmf.nist.gov/, Release 1.0.11 of 2016-06-08.
\newblock Online companion to \cite{Olver:2010:NHMF}.

\bibitem[dLP95]{deliguoro-piperno:ndlc}
Ugo de' Liguoro and Adolfo Piperno.
\newblock {N}ondeterministic extensions of untyped $\lambda$-calculus.
\newblock {\em Information and Computation}, 149-177(122), 1995.

\bibitem[Ehr05]{ehrhard:fs}
Thomas Ehrhard.
\newblock Finiteness spaces.
\newblock {\em Mathematical Structures in Computer Science}, 15(4):615--646,
  2005.

\bibitem[Ehr10]{ehrhard:finres}
Thomas Ehrhard.
\newblock A finiteness structure on resource terms.
\newblock In {\em Proceedings of the 25th Annual {IEEE} Symposium on Logic in
  Computer Science, {LICS} 2010, 11-14 July 2010, Edinburgh, United Kingdom},
  pages 402--410. {IEEE} Computer Society, 2010.

\bibitem[Ehr16]{ehrhard:dill}
Thomas Ehrhard.
\newblock An introduction to differential linear logic: proof-nets, models and
  antiderivatives.
\newblock {\em CoRR}, abs/1606.01642, 2016.

\bibitem[ER03]{er:tdlc}
Thomas Ehrhard and Laurent Regnier.
\newblock The differential lambda-calculus.
\newblock {\em Theoretical Computer Science}, 309:1--41, 2003.

\bibitem[ER05]{er:diffnets}
Thomas Ehrhard and Laurent Regnier.
\newblock Differential interaction nets.
\newblock {\em Electr. Notes Theor. Comput. Sci.}, 123:35--74, 2005.

\bibitem[ER06]{er:bkt}
Thomas Ehrhard and Laurent Regnier.
\newblock B{\"o}hm trees, {K}rivine's machine and the {T}aylor expansion of
  $\lambda$-terms.
\newblock In Arnold Beckmann, Ulrich Berger, Benedikt L{\"o}we, and John~V.
  Tucker, editors, {\em CiE}, volume 3988 of {\em Lecture Notes in Computer
  Science}, pages 186--197. Springer, 2006.

\bibitem[ER08]{er:resource}
Thomas Ehrhard and Laurent Regnier.
\newblock Uniformity and the taylor expansion of ordinary lambda-terms.
\newblock {\em Theor. Comput. Sci.}, 403(2-3):347--372, 2008.

\bibitem[Gir86]{girard:fifteen}
Jean{-}Yves Girard.
\newblock The system {F} of variable types, fifteen years later.
\newblock {\em Theor. Comput. Sci.}, 45(2):159--192, 1986.

\bibitem[Gir87]{girard:ll}
Jean-Yves Girard.
\newblock Linear logic.
\newblock {\em Theoretical Computer Science}, 50:1--102, 1987.

\bibitem[Gir88]{girard:quantitative}
Jean-Yves Girard.
\newblock Normal functors, power series and lambda-calculus.
\newblock {\em Annals of Pure and Applied Logic}, 37(2):129--177, 1988.

\bibitem[Gol13]{golan:semirings}
J.S. Golan.
\newblock {\em Semirings and their Applications}.
\newblock SpringerLink : B{\"u}cher. Springer Netherlands, 2013.

\bibitem[Has96]{hasegawa:generating}
Ryu Hasegawa.
\newblock The generating functions of lambda terms.
\newblock In Douglas~S. Bridges, Cristian~S. Calude, Jeremy Gibbons, Steve
  Reeves, and Ian~H. Witten, editors, {\em First Conference of the Centre for
  Discrete Mathematics and Theoretical Computer Science, {DMTCS} 1996,
  Auckland, New Zealand, December, 9-13, 1996}, pages 253--263.
  Springer-Verlag, Singapore, 1996.

\bibitem[Joy86]{joyal:especes}
Andr{\'e} Joyal.
\newblock {\em Foncteurs analytiques et esp{\`e}ces de structures}, pages
  126--159.
\newblock Springer Berlin Heidelberg, Berlin, Heidelberg, 1986.

\bibitem[KKSdV97]{kksv:infinitary}
J.R. Kennaway, J.W. Klop, M.R. Sleep, and F.J. de~Vries.
\newblock Infinitary lambda calculus.
\newblock {\em Theoretical Computer Science}, 175(1):93 -- 125, 1997.

\bibitem[Kri90]{krivine:lc}
Jean-Louis Krivine.
\newblock {\em Lambda-calcul, types et mod\`eles}.
\newblock Masson, Paris, 1990.

\bibitem[Lai16]{laird:fixed}
J.~Laird.
\newblock Fixed points in quantitative semantics.
\newblock In Martin Grohe, Eric Koskinen, and Natarajan Shankar, editors, {\em
  {LICS} '16, New York, NY, USA, July 5-8, 2016}, pages 347--356. {ACM}, 2016.

\bibitem[Lev16]{leventis:phd}
Thomas Leventis.
\newblock {\em {Probabilistic λ-theories}}.
\newblock PhD thesis, {Aix-Marseille Université}, December 2016.

\bibitem[LMMP13]{lmcp:weighted}
Jim Laird, Giulio Manzonetto, Guy McCusker, and Michele Pagani.
\newblock Weighted relational models of typed lambda-calculi.
\newblock In {\em 28th Annual {ACM/IEEE} Symposium on Logic in Computer
  Science, {LICS} 2013, New Orleans, LA, USA, June 25-28, 2013}, pages
  301--310. {IEEE} Computer Society, 2013.

\bibitem[OLBC10]{Olver:2010:NHMF}
F.~W.~J. Olver, D.~W. Lozier, R.~F. Boisvert, and C.~W. Clark, editors.
\newblock {\em {NIST Handbook of Mathematical Functions}}.
\newblock Cambridge University Press, New York, NY, 2010.
\newblock Print companion to \cite{NIST:DLMF}.

\bibitem[PTV16]{ptv:taylorsn}
Michele Pagani, Christine Tasson, and Lionel Vaux.
\newblock Strong normalizability as a finiteness structure via the taylor
  expansion of λ-terms.
\newblock In Bart Jacobs and Christof L{\"{o}}ding, editors, {\em Foundations
  of Software Science and Computation Structures - 19th International
  Conference, {FOSSACS} 2016, Eindhoven, The Netherlands, April 2-8, 2016,
  Proceedings}, volume 9634 of {\em Lecture Notes in Computer Science}, pages
  408--423. Springer, 2016.

\bibitem[TAO17]{tao:species}
Takeshi Tsukada, Kazuyuki Asada, and C.{-}H.~Luke Ong.
\newblock Generalised species of rigid resource terms.
\newblock In {\em 32nd Annual {ACM/IEEE} Symposium on Logic in Computer
  Science, {LICS} 2017, Reykjavik, Iceland, June 20-23, 2017}, pages 1--12.
  {IEEE} Computer Society, 2017.

\bibitem[Tas09]{tasson:phd}
Christine Tasson.
\newblock {\em Sémantiques et syntaxes vectorielles de la logique linéaire}.
\newblock PhD thesis, Université Paris Diderot -- Paris 7, December 2009.

\bibitem[TdF03]{tortora:obsessional}
Lorenzo Tortora~de Falco.
\newblock Obsessional experiments for linear logic proof-nets.
\newblock {\em Mathematical Structures in Computer Science}, 13(6):799--855,
  2003.

\bibitem[TV16]{tv:transport}
Christine Tasson and Lionel Vaux.
\newblock Transport of finiteness structures and applications.
\newblock {\em Mathematical Structures in Computer Science}, page 1–36, 2016.

\bibitem[Vau07]{vaux:alglam2}
Lionel Vaux.
\newblock On linear combinations of $\lambda$-terms.
\newblock In Franz Baader, editor, {\em RTA}, volume 4533 of {\em Lecture Notes
  in Computer Science}, pages 374--388. Springer, 2007.

\bibitem[Vau09]{vaux:alglam}
Lionel Vaux.
\newblock The algebraic lambda calculus.
\newblock {\em Mathematical Structures in Computer Science}, 19(5):1029--1059,
  2009.

\bibitem[Vau17]{vaux:taylor-beta}
Lionel Vaux.
\newblock Taylor expansion, lambda-reduction and normalization.
\newblock In Valentin Goranko and Mads Dam, editors, {\em 26th {EACSL} Annual
  Conference on Computer Science Logic, {CSL} 2017, August 20-24, 2017,
  Stockholm, Sweden}, volume~82 of {\em LIPIcs}, pages 39:1--39:16. Schloss
  Dagstuhl - Leibniz-Zentrum fuer Informatik, 2017.

\end{thebibliography}

\end{document}